\documentclass[11pt,draftclsnofoot,twoside,onecolumn,letter]{IEEEtran}
\usepackage{graphicx}
\usepackage{epstopdf}
\DeclareGraphicsExtensions{.eps,.pdf} 
\graphicspath{ {figures/} }
\usepackage{cite}
\usepackage{amssymb,amsmath,array}
\usepackage{subfigure}
\usepackage{amssymb,amsmath}
\usepackage{algorithm}
\usepackage{algorithmic}
\usepackage{color}
\usepackage{multirow}
\usepackage{multicol}

\usepackage{fancyhdr}
\fancypagestyle{mahmood}{%
	\fancyhf{} % clear all fields
	
	\fancyhead[C]{\footnotesize This work has been accepted for publication in IEEE Transactions on Wireless Communications. Copyright may be transferred without notice, after which this version may no longer be accessible.}
}%

\makeatletter
\let\ps@IEEEtitlepagestyle\ps@mahmood
\makeatother

\makeatletter
\newcommand\restartchapters{\par
  \setcounter{chapter}{0}%
  \setcounter{section}{0}%
  \gdef\@chapapp{\chaptername}%
  \gdef\thechapter{\@arabic\c@chapter}}
\makeatother

\newtheorem{remark}{\underline{\it Remark}}

\newtheorem{lemma}{\underline{\it Lemma}}

\newtheorem{proposition}{\underline{\it Proposition}}

\newcommand{\tr}{{\mathrm{tr}}}
\newcommand{\diag}{{\mathrm{diag}}}

\newcommand{\thmend}{\hspace*{\fill}~\QEDopen\par\endtrivlist\unskip}

\newcommand{\lbn}{\par\medskip\noindent}

\newcommand{\ULU}{\mathtt{U}_{\ell}^{\mathtt{u}}}
\newcommand{\ul}{\mathtt{u}}
\newcommand{\SI}{\mathtt{SI}}
\newcommand{\dl}{\mathtt{d}}
\newcommand{\DLU}{\mathtt{U}_{k}^{\mathtt{d}}}

\newcommand{\ds}{\displaystyle}

\usepackage{xcolor}

\newcommand*{\hili}{\color{black}}
\newcommand*{\hilise}{\color{black}}
\newenvironment{bsmallmatrix}
{\left[\begin{smallmatrix}} {\end{smallmatrix}\right]}

\usepackage{capt-of}% or \usepackage{caption}
\usepackage{dblfloatfix}
\usepackage{varwidth}
\makeatletter
\g@addto@macro\normalsize{%
	\setlength\abovedisplayskip{1pt}
	\setlength\belowdisplayskip{1pt}
	\setlength\abovedisplayshortskip{1pt}
	\setlength\belowdisplayshortskip{1pt}
}
\makeatother

\newcommand{\subparagraph}{}

\usepackage{titlesec}
\titlespacing{\section}{0pt}{1pt}{0pt}

\begin{document}

\title{\huge {\hili Joint Antenna Array Mode Selection and User Assignment for Full-Duplex MU-MISO Systems}}
\author{
	\IEEEauthorblockN{ Hieu V. Nguyen, Van-Dinh Nguyen, Octavia A. Dobre, \\ Yongpeng Wu, and Oh-Soon Shin}\\
%	\IEEEauthorblockA{$^{\dag}$School of Electronic Engineering, Soongsil University, Korea \\
%		$^{\dag\dag}$Department of ICMC Convergence Technology, Soongsil University, Korea \\
%		$^{\ddag}$Faculty of Engineering and Applied Science, Memorial University of Newfoundland, St. John's, NL, Canada\\
%		$^*$Corresponding author (E-mail: osshin@ssu.ac.kr)\vspace{-25pt}
%	}
	\thanks{H. V. Nguyen, V.-D. Nguyen, and O.-S. Shin are with the School of Electronic Engineering \& Department of ICMC Convergence Technology, Soongsil University, Seoul 06978, Korea (e-mail: \{hieuvnguyen, nguyenvandinh, osshin\}@ssu.ac.kr).
	\newline \indent O. A. Dobre is with the Faculty of Engineering and Applied Science, Memorial University of Newfoundland, St. John's, NL, Canada (e-mail: adobre@mun.ca). 
	\newline \indent Y. Wu is with the Department of Electronic Engineering, Shanghai Jiao Tong University, China, Minhang 200240, China (Email: yongpeng.wu@sjtu.edu.cn). 
	\newline \indent Part of this work was presented at the IEEE Globecom 2017,  Singapore \cite{Hieu:GLOBECOM:2017}. 
	}
}
%	\thanks{This work was supported in part by an Institute for Information \& Communications Technology Promotion (IITP) grant funded by the Korean government (MSIT) (No. 2017-0-00724, Development of Beyond 5G Mobile Communication Technologies (Ultra-Reliable, Low-Latency, and Massive Connectivity) and Combined Access Technologies for Cellularbased Industrial Automation Systems), in part by the National Research Foundation of Korea (NRF) grant funded by the Korean government (MSIT) (NRF-2017R1A5A1015596), and in part by the Natural Sciences and Engineering Research Council of Canada (NSERC), through its Discovery program.}

\maketitle
\vspace*{-40pt}
\begin{abstract}
This paper considers a full-duplex (FD) multiuser multiple-input single-output system where  a base station  simultaneously serves both uplink (UL) and downlink (DL) users on the same time-frequency resource. The crucial barriers in implementing FD systems reside  in the residual self-interference  and  co-channel interference. To accelerate the use of FD radio in future wireless networks, we aim at managing the network interference more effectively by jointly designing the selection of half-array antenna modes (in the transmit or receive mode) at the BS with time phases and user assignments. The first problem of interest is to maximize the overall sum rate subject to quality-of-service  requirements, which is formulated as a highly non-concave utility function followed by non-convex constraints. To address the design problem, we propose an iterative low-complexity algorithm by developing new inner approximations, and its convergence to a stationary point is guaranteed. To provide more insights into the solution of the proposed design, a general max-min rate optimization is further considered to maximize the minimum per-user rate while satisfying a given ratio between UL and DL rates. Furthermore, a robust algorithm is devised to verify that the proposed scheme works well under channel uncertainty. Simulation results demonstrate that the proposed algorithms exhibit fast convergence and substantially outperform existing schemes.
\end{abstract}

\section{Introduction} \label{Introduction} 

The rapid growth of the demand for high data rate in the next-generation communication systems requires innovative technologies that make the maximal use of radio spectrum. Even though advanced techniques such as multiple-input multiple-output (MIMO) antennas have been implemented to improve network throughput \cite{MietznerCST09}, half-duplex (HD) systems, where uplink (UL) and downlink (DL) communications are carried out orthogonally in time domain  or in frequency domain, may not be able to provide  sufficient improvements in the spectral efficiency (SE). Full-duplex (FD) communications, which allow for simultaneous UL and DL transmissions in the same  time-frequency resource, theoretically double the SE when compared to HD. As a result, FD communications have been recognized as a promising technology for the forthcoming 5G wireless networks   \cite{Sabharwal:JSAC:Feb2014,ZhangCM15,Wong5Gbook17,Yadav:VehCOMM:18}.

Although the potential gains of FD  can be easily foreseen, its major drawback  is self-interference (SI) from the transmit (Tx) to  receive (Rx) antennas at  an FD wireless transceiver (e.g., a base station (BS)), which is typically much stronger than the signal of interest.  Recently, there have been many efforts to develop analog and digital SI cancellation techniques  to bring the SI power at the noise level when the Tx power is relatively low \cite{Duarte:TWC:12,Bharadia13,KorpiCoMag16,LaughlinCoMag15,Yadav:LSP:2018}. Such results have spurred research at the system level, considering small cell-based systems. In practice, complete elimination of the SI is not possible due to the limitations of hardware design \cite{Day:TSP:Jul2012, Zlatanov:TCOMM:Mar2017}, and thus, recent works have considered FD systems under the effects of \textit{residual} SI. Moreover, FD cellular networks also suffer from  the user-to-user (UE-to-UE) interference caused by UL users to the ones in the DL channel,  referred to as co-channel interference (CCI), which may become a performance bottleneck. The CCI is unavoidable in a cellular network, and it largely depends on the user position in the cell. While CCI was ignored in  earlier FD designs \cite{Dan-SP-13}, the latest research has taken it into account, especially in dense user-deployed small cell systems \cite{Dinh:TCOMM:2017, Yadav:Access, Dan:TWC:14, Tam:TCOM:16, Aquilina:TCOMM:2017, Dinh:Access, Cirik:WCL:June2016, Cirik:TCOMM:March2018, Nguyen:JSAC:18}.

\subsection{Related Works}

By considering both the residual SI and CCI, recent research has mainly focused on improving the SE \cite{Dinh:TCOMM:2017, Yadav:Access, Dan:TWC:14, Tam:TCOM:16, Aquilina:TCOMM:2017, Dinh:Access,Cirik:WCL:June2016, Cirik:TCOMM:March2018}. In \cite{Dinh:TCOMM:2017} and \cite{ Yadav:Access}, the sum rate (SR) maximization was studied in an FD system for both information and energy transfer. In \cite{Dan:TWC:14}, joint beamforming and power allocation were investigated to maximize the SE for a single-cell scenario.\cite{Tam:TCOM:16} further explored a multi-cell scenario, along with an additional quality-of-service (QoS) constraint. {\hili Under channel uncertainty, a robust design for the weighted SE optimization problem for a single-cell system was considered in \cite{Cirik:WCL:June2016}. In an FD multi-cell system, the weighted SE maximization problem was further examined in \cite{Aquilina:TCOMM:2017}, while the authors in \cite{Cirik:TCOMM:March2018} focused on the user fairness design through a max-min problem for the SINRs of all users in the entire network.} The performance of these systems degrades quickly as the residual SI and CCI increase, and to overcome this drawback, an approach integrating user grouping and time allocation was proposed in \cite{Dinh:Access}. Nevertheless, the optimization under many slots of the communication time block (or the coherence time) results in an extremely high complexity, while the gain achieved from increasing the number of time slots suddenly decreases. {\hili If the SI and CCI become more severe, the system performance is still worse than the traditional HD scheme even with a large number of groups, since the full antenna array is not exploited in the transmission. Moreover, the number of binary variables to optimize the BS-user association scales exponentially with the number of DL and UL users.}

Recently, a joint design of user assignment, power control and antenna selection has been applied to FD systems that further improves SE. In \cite{Nam:TWC:Jun2015, Di:TWC:Dec2016, Silva:TWC:Oct2017}, the authors proposed an FD design based on user assignment, where  users are assigned to one of the sub-carriers in an orthogonal frequency-division multiple access network. However, the improvement of such FD systems mainly relies on the frequency resource allocation, while the effect of SI still remains strong given that FD communication occurs in the same band. {\hili Regarding user assignment design, Ahn \textit{et al.} in \cite{Ahn:TWC:Nov2016} proposed a user selection scheme where a set of users is selected to maximize the overall SR. The number of selected UL and DL users is assumed to be less than or equal to the number of receive (Rx) and transmit (Tx) antennas at the base station (BS), respectively.} Regarding antenna selection techniques, they were studied in some simple FD end-to-end communications. For example, \cite{Yang:TWC:Jul2015} considered an FD relaying system where two-antenna relays can switch their antenna modes (Tx and Rx) to improve the SE in forwarding the messages from the source to the destination. A point-to-point FD MIMO system was investigated in \cite{Jang:TVT:Dec2016}, in which Tx and Rx antennas are allowed to be active or not, to enhance the energy efficiency. More challengingly, in addition to the inherently present SI, FD multiuser systems also experience mutual interferences among users, and therefore, they should be considered in designing an FD system. Furthermore, a per-antenna selection causes high computational complexity, which motivates us to consider a half-array (HA) antenna selection in this work.

\subsection{Motivation and Contributions}
Despite advancements made so far to SI cancellation, the residual SI still has a substantially negative impact on the system performance, and thus, needs to be mitigated. As evident from \cite{Dan:TWC:14, Tam:TCOM:16, Aquilina:TCOMM:2017,Dinh:TCOMM:2017,Dinh:Access,Yadav:Access}, recent advanced FD designs may provide worse performance than HD when the residual SI becomes severe. Besides, as previously mentioned, there exists a small, but not negligible amount of CCI, especially when considering an FD system in a small cell scenario with UL and DL users in close proximity \cite{Dan:TWC:14, Tam:TCOM:16}.
Therefore, it is highly desirable to develop an efficient design that utilizes the same frequency resource for all users within one transmission time block while improving the SE. Moreover, the FD system should achieve at least the same performance as HD counterpart even in unexpected conditions.

In contrast with \cite{Dinh:TCOMM:2017, Yadav:Access, Dan:TWC:14, Tam:TCOM:16, Aquilina:TCOMM:2017, Dinh:Access}, this paper considers a small cell FD system, which is enabled with an HA antenna mode selection and two-phase transmission. In particular, the transmission time block is split into two phases, and for each phase, two HAs of the antennas at the FD-enabled BS can be independently selected for UL and DL transmissions. By using this mode selection, our approach automatically switches among HD, conventional FD \cite{Dan:TWC:14}, two-phase FD \cite{Dinh:Access}, and even hybrid transmissions where the HD scheme (UL or DL transmission) is used in one phase and the FD scheme is applied in the other phase. Compared to the prior work, the proposed FD design has several advantages. Firstly, the proposed two-phase transmission helps better exploit the spatial degrees-of-freedom (DoF) since the number of users served at the same time is effectively reduced. Secondly, HA antenna selection supports the BS in finding better channel conditions and enabling the hybrid transmission. Moreover, by a joint design of the HA antenna mode selection and the two-phase transmission, we can achieve a better system performance, with the harmful effects of residual SI and CCI significantly reduced when compared with a typical FD cellular network \cite{Dan:TWC:14,Tam:TCOM:16,Dinh:Access,Aquilina:TCOMM:2017}.

The advantages of the proposed FD design are paralleled with several challenges in the practical implementation. As the SI and CCI are strongly related to the performance of the DL and UL channels, it usually requires introducing integer variables to manage the BS-UE association. Consequently, the joint design problems of interest belong to the difficult class of mixed-integer non-convex problem, where both types of the binary and continuous variables are coupled to each other in the rate functions. In addition, the design problems are established in a blend of the UL (DL) channels and antenna selection variables, making the problems even more challenging. Furthermore, the acquisition of perfect channel state information (CSI), including the interference channels from the multiple users, is relatively unrealistic, and thus, the consideration of a joint robust FD transmission strategy becomes more critical.  In this respect, a rigorous optimization and analysis are required to develop a low-complexity joint design.

Specifically, the main contributions of this paper are briefly described as follows:
{\hili 
\begin{enumerate}
	\item We first propose a new FD transmission model in which two HAs at the BS are independently selected between the Rx and Tx modes in two phases of a time block. This model not only improves the achievable rate by selecting the modes for two HAs, but also mitigates effects of both residual SI and CCI by separating users into two phases. Interestingly, the new model can enable  hybrid modes of the HD and FD to utilize a full-array antenna, which is not investigated in earlier works. 
	\item We formulate an SR maximization problem with QoS constraints under perfect CSI, which jointly optimizes HA mode selection, dynamically-timed phase, user assignment and power allocation. The design problem is a mixed-integer non-convex problem, and we present an iterative algorithm to solve it optimally in a centralized manner. In an effort to reduce the computational complexity, we first derive a set of subproblems without the integer constraints and tighten  bounds for the constraints of subproblems. With the newly developed inner approximations, these  subproblems can be efficiently solved via successive convex approximations. The optimal solution of the original problem is determined to correspond to the best optimal value of subproblems.
	\item We further introduce a general max-min problem which aims at guaranteeing the achievable rate fairness. Although the objective function of this problem is not only non-concave but also non-smooth, the developments made to address the SR maximization problem  are useful to obtain  its solution.
	\item For the imperfect CSI case, we consider a robust counterpart of the SR maximization problem when the errors are mainly due to estimation inaccuracies. We show that the worst-case robust SR maximization problem can be addressed in a similar manner as the SR maximization.
	\item Numerical results  show that the proposed algorithms  have a fast convergence speed and outperform existing schemes. The new scheme is further shown to provide robustness against the estimation error under imperfect CSI.
\end{enumerate}}

\subsection{Paper Organization and Notation}

The remainder of this paper is organized as follows. In Section \ref{System Model and Problem Formulation}, we introduce the system model and formulate the design problem. The proposed algorithms for the SR maximization and max-min rate optimization are presented in Sections \ref{sec: sum rate maximization} and \ref{sec: max-min rate optimization}, respectively. In Section \ref{sec: Channel Uncertainty}, we study the robust FD transmission under  uncertain channel information. Numerical results are  presented in Section \ref{NumericalResults}, and Section \ref{Conclusion} concludes the paper.

\emph{Notation}:  $\mathbf{X}^{T}$, $\mathbf{X}^{H}$ and $\tr(\mathbf{X})$ are the transpose, Hermitian transpose and trace of a matrix $\mathbf{X}$, respectively. $\left\langle \mathbf{x},\mathbf{y}\right\rangle=\tr(\mathbf{x}^H\mathbf{y})$ denotes the inner product of vectors $\mathbf{x}$ and $\mathbf{y}$. $\|\cdot\|$ denotes the Euclidean norm of a matrix or vector, while $|\cdot|$ stands for the absolute value of a complex scalar. $\mathbb{E}[\cdot]$ denotes the statistical expectation and $\Re\{\cdot\}$ returns the real part of the argument. The notation $\mathbf{X}\succeq\mathbf{0}$ and $\mathbf{X}\succ\mathbf{0}$ represent that
$\mathbf{X}$ is a positive-semidefinite and positive-definite matrix, respectively. $ \mathbf{x} \preceq \mathbf{y}  $ denotes the element-wise comparison of vectors, in which a certain element of $ \mathbf{x} $ is not larger than the element of $ \mathbf{y} $ at the same index. $ \diag(\mathbf{x}) $ stands for a diagonal matrix containing the elements of vector $ \mathbf{x} $ as its main diagonal. $\mathbf{x}\sim\mathcal{CN}(\boldsymbol{\eta},\boldsymbol{Z})$ means that $\mathbf{x}$ is a random vector following a complex circularly symmetric Gaussian distribution with mean  $\boldsymbol{\eta}$ and covariance matrix $\boldsymbol{Z}$. $ \mathbf{dom}\;f $ stands for the domain of the function $ f $. $ \otimes $ is the Kronecker operator, while $ \mathbf{and} $ and $ \mathbf{nor} $ denote the \textit{and} and \textit{negative or} operations in Boolean algebra, respectively.

\section{Preliminaries} \label{System Model and Problem Formulation}

\subsection{System Model}
\vspace{-20pt}
\begin{figure}[htbp] 
	\begin{minipage}[b]{0.45\columnwidth}
		\centering
		\includegraphics[width=0.58\columnwidth,trim={2.5cm 0.0cm 0cm 0.0cm}]{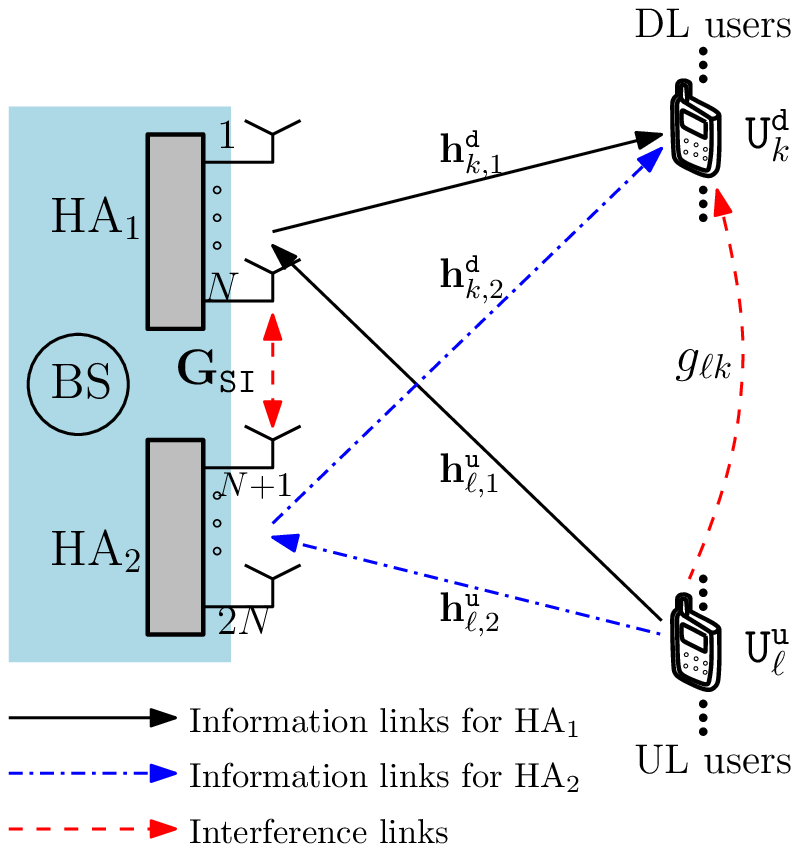}
		\caption{A small cell FD MU-MISO system.  At the BS, HA$_1$ is from the $ 1 $-st to the $ N $-th antenna, while HA$_2$ is from the $ (N+1) $-th to the $ 2N $-th antenna; at a certain moment, either one of the links from or to a HA is active, depending on the Tx/Rx-antenna modes (Rx mode for UL and Tx mode for DL).}
		\label{fig: system model}
	\end{minipage}
	\hfill
	\begin{minipage}[b]{0.5\columnwidth}
		\centering
			\includegraphics[width=1.2\columnwidth,trim={1.0cm 0.0cm -2cm 1.0cm}]{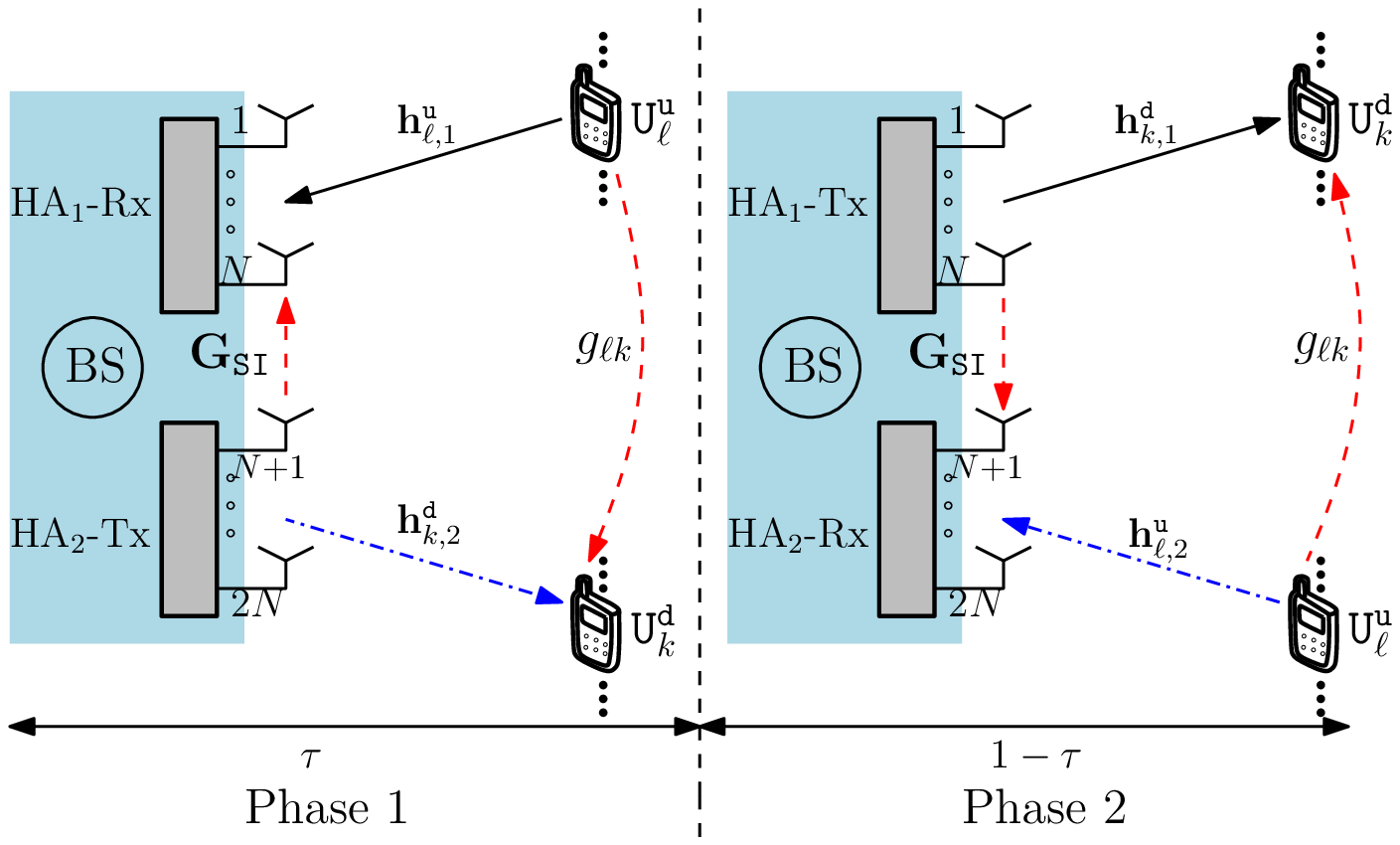}
			\caption{An example of two-phase transmission in the FD system. In Phase 1, HA$_1$ is set to the Rx mode, while HA$_2$ is set to the Tx mode, and vice versa in Phase 2. The links shown in the figure represent the active channels.}
		\label{fig: system model - 2 phases}
		\end{minipage}
\end{figure}

We consider a small cell FD multiuser multiple-input single-output (MU-MISO) system  consisting of an FD-BS equipped with $ 2N $ antennas, $ L $ single-antenna UL users, and $ K $ single-antenna DL users, as illustrated in Fig. \ref{fig: system model}. We denote $\ULU, \; \forall \ell \in \mathcal{L}\triangleq \{1,2,\dots,L\} $ and $\DLU, \; \forall k \in \mathcal{K}\triangleq \{1,2,\dots,K\} $, the $ \ell $-th UL and $ k $-th DL user, respectively. The array of $ 2N $ antennas at the BS is split into two HAs of $ N $ antennas: HA$_1$ from the $ 1 $-st to the $ N $-th antenna and HA$_2$ from the $ (N+1) $-th to the $ 2N $-th antenna. Each HA can switch between Rx and Tx modes, corresponding to UL and DL, respectively. The  channel vectors from $\ULU$  to HA$_i$ and from HA$_i$ to $\DLU$ are denoted by $ \mathbf{h}_{\ell,i}^\ul \in \mathbb{C}^{N\times1}$ and $ \mathbf{h}_{k,i}^\dl\in \mathbb{C}^{N\times1} $, respectively.  At a certain time, the channels in Fig. \ref{fig: system model} are either active or inactive, depending on the Tx/Rx mode selection for HA. The matrix $ \mathbf{G}_{\SI} \in \mathbb{C}^{N\times N} $ represents the SI channel between HA$_1$ and HA$_2$, while $ g_{\ell k} $ stands for the CCI channel from $\ULU$ to $\DLU$. The channel vectors capture the effects of both large- and small-scale fading. 

%{\hili Herein, we assume that an aperiodic feedback is used as channel quality index (CQI) signaling with aperiodic reporting, which allows the BS to adapt the CQI transmission time. In the rest of the paper, we will develop an FD scheme and devise a low-complexity algorithm, which use the CSI as predetermined in the channel information estimation, to improve the system performance through better exploiting channel condition.}

We consider that the FD transmission takes place in two phases within each transmission time block, which is normalized to 1. Fig. \ref{fig: system model - 2 phases} describes an example for the FD transmission in two phases, with transmission intervals of $ \tau \text{ and } 1-\tau\; (0<\tau<1)$, respectively. In the first phase, HA$_1$ and HA$_2$ are set to the Rx and Tx modes, respectively. Hence, $ \mathbf{h}_{\ell,1}^\ul $ and $ \mathbf{h}_{k,2}^\dl $ are active, while $ \mathbf{h}_{k,1}^\dl $ and $ \mathbf{h}_{\ell,2}^\ul $ are inactive. During the second phase, the Tx/Rx modes of the HAs are interchanged, i.e., $ \mathbf{h}_{\ell,2}^\ul $ and $ \mathbf{h}_{k,1}^\dl $ are active, while $ \mathbf{h}_{k,2}^\dl $ and $ \mathbf{h}_{\ell,1}^\ul $ are not. {\hili Herein, Fig. \ref{fig: system model - 2 phases} illustrates one possible case to indicate that the HAs can be changed across two phases. The system model tends to exploit the channel conditions through the combination of the HA modes selection and two-phase transmission, resulting in three different transmission modes: FD, HD, and hybrid mode.} In summary, the mode selection for HA$_i$ in phase $ j $ $ (j=1,2) $ is optimized using the variables $ \omega_{ij} \in \{0, 1\}$, which form a binary matrix $ \boldsymbol{\Omega} \in \mathbb{R}^{2\times2} $ as
\begin{equation} \label{eq: Omega matrix - HA assignment}
\boldsymbol{\Omega} = 
\begin{bmatrix}
\omega_{11}  & \omega_{12} \\		
\omega_{21} & \omega_{22} \\
\end{bmatrix}
\vspace{5pt}
\end{equation}
where $ \omega_{ij}=0 $ or 1 indicates that HA$_i$ in phase $ j $ is set to the Rx mode (UL) or Tx mode (DL).  We let the $ j $-th column of $ \boldsymbol{\Omega} $ be $ \boldsymbol{\omega}_j $. {\hili Note that the HA modes in two phases depend on the entries of matrix $ \boldsymbol{\Omega} $. Therefore, $ \boldsymbol{\Omega} $ represents the transmission states of the system in one time block. Although \eqref{eq: Omega matrix - HA assignment} describes a model with a $ 2\times2 $ matrix, a scalable system can be enabled by changing the number of rows and/or that of columns of $ \boldsymbol{\Omega} $ to create a general transmission model with multiple sub-array antennas and multiple time phases. However, solving the general case may not be suitable for practical implementations, as more RF chains and high complexity are required. Therefore, we consider the HA mode selection and two-phase transmission in the rest of the paper.}

%{\hilise
%\begin{remark}
%	The proposed model is to exploit the physical channel condition and manage the interference for the FD system, where one time block is split into two transmission phases. Three modes including FD, HD and hybrid are generated by the combination of HA antenna selection and two phases. Especially, the hybrid mode is not a concatenation of one time block from the conventional FD and two time blocks from HD. This mode is selected when two HAs are assigned to different Tx/Rx-antenna modes in one phase, but they are assigned to the same Tx/Rx-antenna mode in the other phase.
%\end{remark}
%}

\subsection{Problem Formulation}
In the DL transmission of phase $ j $,  the information signal $ x_{k,j}^{\dl} $, with $ \mathbb{E}\bigl[|x_{k,j}^{\dl}|^2\bigr]=1$, intended to $\DLU$, is precoded using a beamforming vector $ \mathbf{w}_{k,j} \in \mathbb{C}^{2N\times1} $ at the BS. The signals received at the BS and $\DLU$ in phase $j$ are, respectively, given as
\begin{align} 
\mathbf{y}_{j}^{\ul} & = \sum\nolimits_{\ell=1}^{L} p_{\ell,j}  \boldsymbol{\bar{\Lambda}}_j \mathbf{h}_{\ell}^{\ul} x_{\ell,j}^{\ul} + \underbrace{\rho\sum\nolimits_{k=1}^{K} \mathbf{\bar{G}}_{\SI}^H \boldsymbol{\Lambda}_j \mathbf{w}_{k,j} x_{k,j}^{\dl}}_{\text{SI}}  + \mathbf{n}_{j}, \label{eq: received signal at BS} \\
y_{k,j}^{\dl} & = \sum\nolimits_{k'=1}^{K} (\mathbf{h}_{k}^\dl)^H \boldsymbol{\Lambda}_j \mathbf{w}_{k',j} x_{k',j}^{\dl} + \underbrace{\sum\nolimits_{\ell=1}^{L} p_{\ell,j} g_{\ell k} x_{\ell,j}^{\ul}}_{\text{CCI}} + n_{k,j}\label{eq: received signal at user}
\end{align}
where $ \mathbf{h}_{\ell}^\ul \triangleq [(\mathbf{h}_{\ell,1}^\ul)^H \; (\mathbf{h}_{\ell,2}^\ul)^H]^H $ and $ \mathbf{h}_{k}^\dl \triangleq [(\mathbf{h}_{k,1}^\dl)^H \; (\mathbf{h}_{k,2}^\dl)^H]^H $. $ \boldsymbol{\Lambda}_j  $ and $ \boldsymbol{\bar{\Lambda}}_j $ are diagonal matrices defined as $ \boldsymbol{\Lambda}_j \triangleq \diag(\boldsymbol{\omega}_j) \otimes \mathbf{I}_N $ and $ \boldsymbol{\bar{\Lambda}}_j \triangleq \diag(\boldsymbol{\bar{\omega}}_j) \otimes \mathbf{I}_N $, where $ \boldsymbol{\bar{\omega}}_j $ is a bit-wise complement of the binary vector $ \boldsymbol{\omega}_j $.  $ p_{\ell,j}^2 \in \mathbb{R} $ and $ x_{\ell,j}^{\ul} $, with $ \mathbb{E}\bigl[|x_{\ell,j}^{\ul}|^2\bigr]=1 $, are the UL transmit power and information signal from $\ULU$ to BS in phase $j$, respectively. $ n_{k,j} \sim \mathcal{CN}(0, \sigma_k^2) $ and $ \mathbf{n}_{j} \sim \mathcal{CN}(\mathbf{0}, \sigma_{\mathtt{U}}^2\mathbf{I}) $ represent the additive white Gaussian noise. The parameter $ \rho \in \bigl[0,1\bigr] $ represents the degree of SI propagation \cite{Riihonen-SP-11}, while the block anti-diagonal matrix $ \mathbf{\bar{G}}_{\SI} $ is defined as $ \mathbf{\bar{G}}_{\SI} = 
\begin{bsmallmatrix}
\mathbf{0} & \mathbf{G}_{\SI} \\
\mathbf{G}_{\SI}^H & \mathbf{0} \\
\end{bsmallmatrix}. $

Let $ \mathbf{w}_{j} \triangleq \bigl[\mathbf{w}_{1,j}^H\;\mathbf{w}_{2,j}^H\cdots\mathbf{w}_{K,j}^H\bigr]^H \in \mathbb{C}^{2NK\times1} $ and $ \mathbf{p}_{j} \triangleq \bigl[p_{1,j}\;p_{2,j}\cdots p_{L,j}\bigr]^H \in \mathbb{C}^{L\times1} $. By using minimum mean square error and successive interference cancellation (MMSE-SIC) decoder \cite{Tse:book:05}, the  signal-to-interference-plus-noise ratio (SINR) of $\ULU$ at the BS in phase $j$ can be written as 
	\begin{equation} \label{eq: SINR ULUs}
	\gamma_{\ell}^{\ul}\bigl(\boldsymbol{\omega}_j, \mathbf{w}_j, \mathbf{p}_j\bigr) = p_{\ell,j}^2 (\mathbf{\tilde{h}}_{\ell,j}^{\ul})^H \boldsymbol{\Psi}_{\ell,j}^{-1} \mathbf{\tilde{h}}_{\ell,j}^{\ul}
	\end{equation}
where 
\[ \boldsymbol{\Psi}_{\ell,j} \triangleq  \sum\nolimits_{\ell'=\ell+1}^{L} p_{\ell',j}^2 \mathbf{\tilde{h}}_{\ell',j}^{\ul}(\mathbf{\tilde{h}}_{\ell',j}^{\ul})^H + \rho^2\sum\nolimits_{k=1}^{K} \mathbf{\tilde{G}}_{j}^H \mathbf{w}_{k,j} \mathbf{w}_{k,j}^H\mathbf{\tilde{G}}_{j} + \sigma_{\mathtt{U}}^2\mathbf{I},
\]
and $ \mathbf{\tilde{h}}_{\ell,j}^{\ul} \triangleq \boldsymbol{\bar{\Lambda}}_j\mathbf{h}_{\ell}^{\ul}  $ and $ \mathbf{\tilde{G}}_{j} \triangleq \boldsymbol{\Lambda}_j \mathbf{\bar{G}}_{\SI} $ are the effective UL  and residual SI channels, respectively. Next, we can express the received SINR of  $\DLU$ in phase $j$ as
	\begin{equation} \label{eq: SINR DLUs}
	\gamma_{k}^{\dl}\bigl(\boldsymbol{\omega}_j, \mathbf{w}_j, \mathbf{p}_j\bigr) = \frac{\bigl|(\mathbf{\tilde{h}}_{k,j}^\dl)^H \mathbf{w}_{k,j}\bigr|^2}{\psi_{k,j}}
	\end{equation}
where 
\[ \psi_{k,j} \triangleq \sum\nolimits_{k'=1,k'\neq k}^{K} \bigl|\bigl(\mathbf{\tilde{h}}_{k,j}^\dl\bigr)^H \mathbf{w}_{k',j}\bigr|^2 + \sum\nolimits_{\ell=1}^{L} p_{\ell,j}^2 |g_{\ell k}|^2 + \sigma_k^2,
\] and $ \mathbf{\tilde{h}}_{k,j}^\dl\triangleq \boldsymbol{\Lambda}_j \mathbf{h}_{k}^\dl $ is the effective channel for $\DLU$. 
 
%{\hili Despite the simpler ZFBF or MRT (MRC) used for both UL and DL transmission in many works, we use beamforming design for DL and MMSE-SIC for UL since they take some advantages. (i) The beamforming design and MMSE-SIC provides a better performance than using ZFBF for DL or both UL/DL. Considering the entries of vectors as the optimized variables, the beamforming design can better exploit channel condition to maximize the signal-to-interference-plus-noise ratio (SINR). Meantime, MMSE-SIC which allows the BS to successively cancel the interference improves the UL performance. (ii) The beamforming design and MMSE-SIC efficiently supports a joint of user assignments and power control, in which the users are automatically active or not due to their transmit power. (iii) Based on the optimization, the beamforming design works well in the lack degree of freedom.}

{\hili 
\begin{remark}
	It is clear that the SINRs in \eqref{eq: SINR ULUs} and \eqref{eq: SINR DLUs} can express all possible cases of $ \boldsymbol{\Omega} $. Remarkably, the SINR is separately computed for each phase, depending on the columns of matrix $ \boldsymbol{\Omega} $, i.e., $ \boldsymbol{\omega}_j $. If the FD mode is selected, half of the effective channel vector is active corresponding to the HA mode in $ \boldsymbol{\omega}_j $. The HD mode is selected when both HAs are set to either Rx or Tx. When both HAs are in Rx mode (or Tx mode), the full vector of the effective UL channel in \eqref{eq: SINR ULUs} (or effective DL channel in \eqref{eq: SINR DLUs}) is active, while the DL SINRs in \eqref{eq: SINR DLUs} (or the UL SINRs in \eqref{eq: SINR ULUs}) are equal to 0 due to the full vector of the DL effective channel (or the effective UL channel) being off.
\end{remark}
}

From \eqref{eq: SINR ULUs} and \eqref{eq: SINR DLUs}, the achievable rates (nats/s/Hz) of $\ULU$ and $\DLU$ in one transmission time block are, respectively, expressed as
\begin{align} 
R_{\ell}^{\ul}\bigl(\boldsymbol{\Omega}, \mathbf{w}, \mathbf{P},\tau,\boldsymbol{\alpha}_{\ell}^{\ul}\bigr) & = \sum_{j=1}^{2} \alpha_{\ell,j}^{\ul}\bar{\tau}_j \ln \bigl(1+\gamma_{\ell}^{\ul}(\boldsymbol{\omega}_j, \mathbf{w}_j, \mathbf{p}_j)\bigr), \label{eq: nonconvex rates uplink} \\
R_{k}^{\dl}\bigl(\boldsymbol{\Omega}, \mathbf{w}, \mathbf{P},\tau,\boldsymbol{\alpha}_{k}^{\dl}\bigr) & = \sum_{j=1}^{2} \alpha_{k,j}^{\dl}\bar{\tau}_j \ln \bigl(1+\gamma_{k}^{\dl}(\boldsymbol{\omega}_j, \mathbf{w}_j, \mathbf{p}_j)\bigr) \label{eq: nonconvex rates downlink} 
\end{align}
where $ \mathbf{w} \triangleq [\mathbf{w}_1^H\;\mathbf{w}_2^H]^H $ and $ \mathbf{P} \triangleq [\mathbf{p}_1\;\mathbf{p}_2] $. The new binary variables $ \boldsymbol{\alpha}_{\ell}^{\ul} \triangleq [\alpha_{\ell,1}^{\ul}\;\alpha_{\ell,2}^{\ul}],\; \forall\ell\in\mathcal{L}$ and $ \boldsymbol{\alpha}_{k}^{\dl} \triangleq [\alpha_{k,1}^{\dl}\;\alpha_{k,2}^{\dl}],  \forall k\in\mathcal{K} $ are used for user assignments, and they indicate whether $\ULU$ and $\DLU$ are respectively served by the BS in  phase $j$. If they are, $ \alpha_{\ell,j}^{\ul} $ and $ \alpha_{k,j}^{\dl} $ are set to 1; otherwise, $ \alpha_{\ell,j}^{\ul} $ and $ \alpha_{k,j}^{\dl} $ are equal to 0. The communication time is defined as $\bar{\tau}_j\triangleq(2-j)\tau+(j-1)(1-\tau)$, ensuring that phases 1 and 2 are within $ \tau $ and $ (1-\tau) $, respectively. 

We consider a general objective function, denoted as $ \mathcal{R}\bigl(\{R_{\ell}^{\ul}\},\{R_{k}^{\dl}\}\bigr) $. For simplicity, $ R_{\ell}^{\ul} $ and $ R_{k}^{\dl} $ are used to represent $ R_{\ell}^{\ul}\bigl(\boldsymbol{\Omega}, \mathbf{w}, \mathbf{P},\tau,\boldsymbol{\alpha}_{\ell}^{\ul}\bigr) $ and $ R_{k}^{\dl}\bigl(\boldsymbol{\Omega}, \mathbf{w}, \mathbf{P},\tau,\boldsymbol{\alpha}_{k}^{\dl}\bigr) $ given in \eqref{eq: nonconvex rates uplink} and \eqref{eq: nonconvex rates downlink}, respectively.
In particular, a joint design problem is mathematically formulated  as
\begin{subequations} \label{eq: prob. general form}
	\begin{IEEEeqnarray}{cl}
		\underset{\boldsymbol{\Omega},\mathbf{w}, \mathbf{P},\tau,\boldsymbol{\alpha}_{\ell}^{\ul},\boldsymbol{\alpha}_{k}^{\dl}}{\text{maximize}} & \quad \mathcal{R}\bigl(\{R_{\ell}^{\ul}\},\{R_{k}^{\dl}\}\bigr)\label{eq: prob. general form a} \\
		\text{subject to} & \quad \omega_{ij} \in \{0,1\}, \forall i,j=1,2, \label{eq: prob. general form b} \\
		& \quad 1 \leq \sum\nolimits_{i=1}^{2}\sum\nolimits_{j=1}^{2}\omega_{ij} \leq 3, \label{eq: prob. general form c} \\
		& \quad \sum\nolimits_{j=1}^{2}\bar{\tau}_j p_{\ell,j}^2 \leq P_{\ell}^{\text{max}},\; \forall \ell \in \mathcal{L}, \label{eq: prob. general form d} \\
		& \quad p_{\ell,j}^2 \leq P_{\ell}^{\infty},\;\forall \ell \in \mathcal{L}, j=1,2, \label{eq: prob. general form e}\\
		& \quad p_{\ell,j} \geq 0, \;\forall \ell \in \mathcal{L}, j=1,2, \label{eq: prob. general form f} \\
		& \quad \sum\nolimits_{j=1}^{2}\bar{\tau}_j \|\mathbf{w}_j\|^2 \leq P_t^{\text{max}}, \label{eq: prob. general form g} \\
		& \quad \|\mathbf{w}_j\|^2 \leq P_t^{\infty}, \; j=1,2, \label{eq: prob. general form h}  \\
		& \quad 0 < \tau < 1, \label{eq: prob. general form i} \\
		& \quad \alpha_{\ell,j}^{\ul},\; \alpha_{k,j}^{\dl} \in \{0,1\}, \; \forall \ell \in \mathcal{L},\forall k \in \mathcal{K}, j=1,2. \label{eq: prob. general form j}
	\end{IEEEeqnarray}							
\end{subequations}
Constraints \eqref{eq: prob. general form b} and \eqref{eq: prob. general form c} are  the mode selections of HAs, in which constraint \eqref{eq: prob. general form c} guarantees that UL and DL transmissions are executed at least in one phase. {\hili In fact, when the summation in \eqref{eq: prob. general form c} is equal to 2, the system operates in either FD or HD mode. If this summation equals 1 or 3, the hybrid mode is activated.} Constraints \eqref{eq: prob. general form d} and \eqref{eq: prob. general form g} cap the maximum transmit power of  $\ULU$ and BS, respectively. As in \cite{Dinh:Access,Nguyen:JSAC:18}, the coupling between communication time $ \bar{\tau}_j $ and transmit power in \eqref{eq: prob. general form d} and \eqref{eq: prob. general form g} indicates that the transmit power in one phase is associated with its interval and independent of the other phase, under the condition that the total transmit power does not exceed the power budgets, $ P_{\ell}^{\text{max}} $ at $\ULU$ and $ P_t^{\text{max}} $ at the BS. The physical constraints \eqref{eq: prob. general form e} and \eqref{eq: prob. general form h} are imposed to ensure that it is not possible to transmit an arbitrary high power at every instant. {\hili Herein, $ P_{\ell}^{\infty} $ and $ P_t^{\infty} $ represent the maximum transmit power per phase at $ \ULU $ and at the BS, respectively. Without loss of generality, we choose $ P_{\ell}^{\infty}\equiv P_{\ell}^{\text{max}} $ and $ P_t^{\infty} \equiv P_t^{\text{max}} $.}
Finally, constraints \eqref{eq: prob. general form j} are associated with the users' assignments in each phase. {\hili Through problem \eqref{eq: prob. general form}, the SI and CCI are significantly reduced since the proposed two-phase scheme allows the system to serve its users in various transmission time $ \bar{\tau}_j $, i.e., one of two phases or the whole transmission block. In addition, the HA modes are dynamically selected in the optimization, and thus, this scheme can easily switch to FD scheme; or asymmetric HD scheme where the UL and DL transmissions independently operating in two phases of one time block can occupy different time intervals; or hybrid mode where FD is used in one phase, and HD using full antenna array for either uplink (UL) or downlink (DL) transmission is employed in the other phase.}

{\hili
\begin{remark}
	Based on problem \eqref{eq: prob. general form}, the modes for HAs are dynamically selected to maximize the objective function (i.e., SR or max-min rate). Moreover, the FD-enabled system can be switched to HD transmission if the SI and CCI become strong. Therefore, the HD mode case corresponds to one of those cases that are generated by HA modes in two phases. When the HD mode case is selected, the full-antenna array ($ 2N $ antennas) is used in each phase for UL or DL transmission while the SI and CCI automatically disappear due to the joint optimization of HA mode selection and power allocation.
\end{remark}
}

\section{Proposed Algorithm for Sum Rate Maximization}\label{sec: sum rate maximization}
In this section, we investigate the following joint design problem for SR maximization:
	\begin{subequations} \label{eq: prob. max sum SE}
		\begin{IEEEeqnarray}{cl}
			\underset{\boldsymbol{\Omega},\mathbf{w}, \mathbf{P},\tau,\boldsymbol{\alpha}_{\ell}^{\ul},\boldsymbol{\alpha}_{k}^{\dl}}{\text{maximize}} & \quad R_{\Sigma}\triangleq \sum\nolimits_{\ell=1}^{L}R_{\ell}^{\ul} + \sum\nolimits_{k=1}^{K}R_{k}^{\dl} \label{eq: prob. max sum SE a} \\
			\text{subject to} 
			& \quad \eqref{eq: prob. general form b}-\eqref{eq: prob. general form j}, \\
			& \quad R_{\ell}^{\ul} \geq \bar{R}_\ell^{\ul}, \forall \ell \in \mathcal{L}, \label{eq: prob. max sum SE c} \\
			& \quad R_{k}^{\dl} \geq \bar{R}_k^{\dl}, \forall k \in \mathcal{K} \label{eq: prob. max sum SE d}
		\end{IEEEeqnarray}
	\end{subequations}
where the objective function $ \mathcal{R}\bigl(\{R_{\ell}^{\ul}\},\{R_{k}^{\dl}\}\bigr)$ in \eqref{eq: prob. general form} is replaced by  $R_{\Sigma}=\sum\nolimits_{\ell\in\mathcal{L}}R_{\ell}^{\ul} + \sum\nolimits_{k\in\mathcal{K}}R_{k}^{\dl} $. $\bar{R}_\ell^{\ul}\geq 0$ in \eqref{eq: prob. max sum SE c} and $\bar{R}_k^{\dl}\geq 0$ in \eqref{eq: prob. max sum SE d} specify the minimum rate requirements for $\ULU$ and $\DLU$, respectively. It is clear that problem \eqref{eq: prob. max sum SE} is a mixed-integer non-convex  programming, which is an NP-hard problem \cite{Tawarmalani:2002:CGO:640669}. In addition, the  binary variables $(\boldsymbol{\Omega},\boldsymbol{\alpha}_{\ell}^{\ul},\boldsymbol{\alpha}_{k}^{\dl})$ and the continuous variables  $(\mathbf{w}, \mathbf{P},\tau)$ are strongly  coupled, making the design problem even more challenging. In what follows, we develop an efficient successive  algorithm with low  complexity to solve problem \eqref{eq: prob. max sum SE}. Specifically, we first decouple the original problem \eqref{eq: prob. max sum SE} into subproblems of lower dimensions, and then use a bound-tightening process to further relieve the computational complexity. Although each subproblem is still non-convex, we iteratively approximate each subproblem by a sequence of convex programs based on  the inner approximation method \cite{Marks:78}, for which an iterative algorithm is then proposed.

\subsection{Subproblems for Problem \eqref{eq: prob. max sum SE}}
Since the complexity of the BS-UE association will increase exponentially with the number of users, it is not practical to seek all possible BS-UE combinations. {\hili In particular, the number of possible cases for  two HAs is the product of $ 4^2 $ cases for the HA mode selection, $ 2^{2L} $ for UL and $ 2^{2K} $ for DL users assignments. To overcome this issue, we first utilize the symmetry property of problem \eqref{eq: prob. max sum SE} to deal with the binary variables $\boldsymbol{\Omega}$ in \eqref{eq: prob. general form b} and \eqref{eq: prob. general form c} by introducing  the following proposition.}
{\hili 
	\lbn
\begin{proposition} \label{prop: Constraint Reduction}
	A finite set is defined as $ \mathcal{B}\triangleq\{\mathrm{b}`b_1b_0|b_0,b_1\in\{0,1\}\} $, where $ \mathrm{b}`b_1b_0 $ denotes a binary number with two bits including the most significant bit, $ b_1 $, and the least significant bit, $ b_0 $. Then, we consider a function $ f_\omega $ as
	\begin{align} \label{eq: function f_omega}
		f_\omega: & \; \{0,1\}^{2\times 1}  \rightarrow \mathcal{B}, \nonumber \\ 
		& \qquad \mathbf{z} \quad \;\; \rightarrow \hat{z}, \; f_\omega(\mathbf{z})=\mathrm{b}`z_{1}z_{2}. 
	\end{align}
	From $ f_{\omega} $ in \eqref{eq: function f_omega}, a mapping is established as
	\begin{align}
	g_\omega: & \; \{0,1\}^{2\times 1}\times \{0,1\}^{2 \times 1} \rightarrow \mathcal{B}\times \mathcal{B}, \nonumber \\ 
	& \quad \boldsymbol{\Omega}=\boldsymbol{\omega}_1 \times \boldsymbol{\omega}_2 \quad \quad \;  \rightarrow \quad \boldsymbol{\hat{\omega}}, \; g_\omega(\boldsymbol{\Omega})=[f_\omega(\boldsymbol{\omega}_1)\; f_\omega(\boldsymbol{\omega}_2)]. \nonumber
	\end{align}
	We define a lexicographically ordered set as 
	\begin{align} 
		\mathcal{B}'\triangleq\{\boldsymbol{\hat{\omega}}=[\hat{\omega}_1\;\hat{\omega}_2]\;|\;\hat{\omega}_1<\hat{\omega}_2\;\vee\;\hat{\omega}_1=\hat{\omega}_2\notin\{\mathrm{b}`00,\mathrm{b}`11\} \}\subset \mathcal{B}\times\mathcal{B}. \label{eq: fixing omega - HA assignment}
	\end{align}
	Without loss of optimality, problem \eqref{eq: prob. max sum SE} can be decoupled into subproblems without constraints \eqref{eq: prob. general form b} and \eqref{eq: prob. general form c} by fixing $ \boldsymbol{\Omega} $, such that $ \boldsymbol{\Omega} \in \mathcal{S}_{\boldsymbol{\Omega}} = g_\omega^{-1}(\mathcal{B}') $. \thmend
\end{proposition}
}

\begin{proof}
	Please see Appendix \ref{app: Constraint Reduction}.
\end{proof}
  
\begin{remark} Since $ \hat{\omega}_j $'s are binary numbers, the inequality $ \hat{\omega}_1 \leq \hat{\omega}_2 $ is not equivalent to $ \boldsymbol{\omega}_1 \preceq \boldsymbol{\omega}_2 $. Therefore, $ (\hat{\omega}_1, \hat{\omega}_2) $ must be searched in the area of $ \mathcal{S}_{\hat{\omega}} $ to find all possible cases of HA modes. Then, the corresponding argument $ (\boldsymbol{\omega}_1, \boldsymbol{\omega}_2) $ obtained via the inverse of $ g_\omega $ is applied to problem \eqref{eq: prob. max sum SE} to derive a subproblem without the binary matrix variable $ \boldsymbol{\Omega} $.
\end{remark}

{\hili
\begin{remark} It is realized that $ \mathcal{S}_{\hat{\omega}} $ can be considered as an ordered set where a lexicographical order is applied to pairs of phases $ (\hat{\omega}_1,\hat{\omega}_2) $. This has two advantages: \textit{(i)} it requires to partly search for the optimal solution, instead of all symmetric regions; \textit{(ii)} without loss of generality, the property of lexicographical order is very useful to extend pairs of two phases to a sequence of multiple phases, while the HA mode selection is easily transformed into multiple sub-array antennas selection by replacing the set $ \mathcal{B} $ with $ \mathcal{B}_n\triangleq \{\mathrm{b}`b_{n-1}\dots b_0|b_0,b_1,\dots,b_{n-1}\in\{0,1\}\} $.
\end{remark}
}

Next, we deal with the binary nature of ($\boldsymbol{\alpha}_{\ell}^{\ul},\boldsymbol{\alpha}_{k}^{\dl}$) in \eqref{eq: prob. general form j}. For $ \boldsymbol{\alpha} \triangleq \bigl[ [\boldsymbol{\alpha}_{\ell}^{\ul}]_{\ell\in \mathcal{L}}^H\; [\boldsymbol{\alpha}_{k}^{\dl}]_{k\in \mathcal{K}}^H \bigr] $ and  $ \mathbf{R} \triangleq [ \mathbf{r}_{\ul}^H\;\mathbf{r}_{\dl}^H]^H $, with $ \mathbf{r}_{\ul}=\bigl[\bar{\tau}_1 \ln\bigl(1+\gamma_{\ell}^{\ul}(\mathbf{w}_1, \mathbf{p}_1,\boldsymbol{\omega}_1)\bigl) $  $ \bar{\tau}_2 \ln\bigl(1+\gamma_{\ell}^{\ul}(\mathbf{w}_2, \mathbf{p}_2,\boldsymbol{\omega}_2)\bigr)\bigr]_{\ell\in\mathcal{L}} $ and $ \mathbf{r}_{\dl}=\bigl[\bar{\tau}_1 \ln\bigl(1+\gamma_{k}^{\dl}(\mathbf{w}_1, \mathbf{p}_1,\boldsymbol{\omega}_1)\bigr) \; \bar{\tau}_2 \ln\bigl(1+\gamma_{k}^{\dl}(\mathbf{w}_2, \mathbf{p}_2,\boldsymbol{\omega}_2)\bigr)\bigr]_{k\in\mathcal{K}} $,  the SR is $ \tr(\boldsymbol{\alpha}\mathbf{R}) $. Since each element of $ \mathbf{R} $ is greater than zero, the SR is maximized if all entries of $ \boldsymbol{\alpha} $ are set to 1. Therefore, fixing $ \alpha_{\ell,j}^{\ul} $ and $ \alpha_{k,j}^{\dl} $ to 1 always provides an upper bound for an arbitrary feasible point. {\hili In other words, we assume an optimal solution with some $ \alpha_{\ell,j}^{\mathtt{u}}=0 $ and $ \alpha_{k,j}^{\mathtt{d}}=0,\; \ell \in \mathcal{L},k \in \mathcal{K}, j=1,2 $, denoted by $ \mathbf{X}^{*}=\{\boldsymbol{\Omega}^{*},\mathbf{w}^{*},\mathbf{P}^{*},\tau^{*},(\boldsymbol{\alpha}_{\ell}^{\mathtt{u}})^{*},(\boldsymbol{\alpha}_{k}^{\mathtt{d}})^{*}|\;\text{some } \alpha_{\ell,j}^{\mathtt{u}}=0 \;\&\; \alpha_{k,j}^{\mathtt{d}}=0\} $. Let $ \mathbf{X}^{*}|_{\alpha=1} $ be $ \mathbf{X}^{*}$ with the values of all $ \alpha_{\ell,j}^{\mathtt{u}} $ and $ \alpha_{k,j}^{\mathtt{d}} $ replaced by 1. It is realized that $ \mathbf{X}^{*}|_{\alpha=1} $ is also an optimal solution for \eqref{eq: prob. max sum SE}, since the objective function value given by $ \mathbf{\hat{X}}|_{\alpha=1} $ is always equal to or larger than that given by $ \mathbf{\hat{X}} $, where $ \mathbf{\hat{X}} $ is an arbitrary feasible point, and $ \mathbf{\hat{X}}|_{\alpha=1} $ is $ \mathbf{\hat{X}} $ with the values of all $ \alpha_{\ell,j}^{\mathtt{u}} $ and $ \alpha_{k,j}^{\mathtt{d}} $ replaced by 1. Intuitively, the user assignment variables aim to effectively switch users on-off, and thus direct computation of the values of user assignment variables as in \cite{Dinh:Access, Nam:TWC:Jun2015} is unnecessary and results in high complexity. As a result, we can consider the constraint \eqref{eq: prob. general form j} in problem \eqref{eq: prob. max sum SE} as feasibility by fixing $ \alpha_{\ell,j}^{\ul} $ and $ \alpha_{k,j}^{\dl} $ to 1 at the beginning of the optimization, and removing them from problem \eqref{eq: prob. max sum SE}.} After solving problem \eqref{eq: prob. max sum SE}, an expression for two cases is used for  user assignments as
\begin{equation} \label{eq: users assignments}
\alpha_{\ell,j}^{\ul}{(\text{or }\alpha_{k,j}^{\dl})} = \begin{cases}
1, & \text{if } \gamma_{\ell, j}^{\ul,(*)}{(\text{or }\gamma_{k,j}^{\dl,(*)}) } \geq \gamma_\varepsilon ,\\
0, & \text{otherwise}
\end{cases}
\end{equation}
where $ \gamma_{\ell,j}^{\ul} $ and $ \gamma_{k,j}^{\dl} $ merely represent $ \gamma_{\ell}^{\ul}\bigl(\boldsymbol{\omega}_j, \mathbf{w}_j, \mathbf{p}_j\bigr) $ and $ \gamma_{k}^{\dl}\bigl(\boldsymbol{\omega}_j, \mathbf{w}_j, \mathbf{p}_j\bigr) $, respectively. $ \gamma_\varepsilon $ is a threshold defined such that if the obtained SINR optimal values, $ \gamma_{\ell, j}^{\ul,(*)} {(\text{or }\gamma_{k, j}^{\dl,(*)})}  $, $ \forall \ell \in \mathcal{L}, \forall k\in \mathcal{K}, j=1,2 $, are lower than $ \gamma_\varepsilon $, they are small enough to be negligible and thus $ \alpha_{\ell,j}^{\ul}{(\text{or }\alpha_{k,j}^{\dl})} $ is set to 0.

By fixing ($ \boldsymbol{\Omega} $, $ \boldsymbol{\alpha}_{\ell}^{\ul} $, $ \boldsymbol{\alpha}_{k}^{\dl} $) as previously discussed, the achievable rates for $\ULU$ and $ \DLU $ are equivalently expressed as $ R_{\ell}^{\ul}\bigl(\mathbf{w}, \mathbf{P},\tau|\boldsymbol{\Omega}\bigr) = \sum\nolimits_{j=1}^{2} \bar{\tau}_j \ln \bigl(1+\gamma_{\ell}^{\ul}(\mathbf{w}_j, \mathbf{p}_j|\boldsymbol{\omega}_j)\bigr) $ and $ R_{k}^{\dl}\bigl(\mathbf{w}, \mathbf{P},\tau|\boldsymbol{\Omega}\bigr) = \sum\nolimits_{j=1}^{2} \bar{\tau}_j \ln\bigl(1+\gamma_{k}^{\dl}(\mathbf{w}_j, \mathbf{p}_j|\boldsymbol{\omega}_j)\bigr) $.
For a given value of $ \boldsymbol{\Omega} \in \mathcal{S}_{\boldsymbol{\Omega}} $,  problem  \eqref{eq: prob. max sum SE}  can be  decoupled into the following subproblems:
\begin{subequations} \label{eq: prob. max sum SE 1}
	\vspace{7pt}
	\begin{IEEEeqnarray}{cl}
		\underset{\mathbf{w}, \mathbf{P},\tau}{\text{maximize}} & \quad \sum\nolimits_{\ell=1}^{L}R_{\ell}^{\ul}\bigl(\mathbf{w}, \mathbf{P},\tau|\boldsymbol{\Omega}\bigr) + \sum\nolimits_{k=1}^{K}R_{k}^{\dl}\bigl(\mathbf{w}, \mathbf{P},\tau|\boldsymbol{\Omega}\bigr), \qquad \\
		\text{subject to} 
		& \quad \eqref{eq: prob. general form d}-\eqref{eq: prob. general form i}, \\
		& \quad R_{\ell}^{\ul}\bigl(\mathbf{w}, \mathbf{P},\tau|\boldsymbol{\Omega}\bigr) \geq \bar{R}_\ell^{\ul}, \; \forall \ell \in \mathcal{L}, \label{eq: prob. max sum SE 1 c} \\
		& \quad R_{k}^{\dl}\bigl(\mathbf{w}, \mathbf{P},\tau|\boldsymbol{\Omega}\bigr) \geq \bar{R}_k^{\dl},\; \forall k \in \mathcal{K}. \label{eq: prob. max sum SE 1 d}
%		\vspace{7pt}
	\end{IEEEeqnarray}
\end{subequations}

\subsection{Bound Tightening for Power Constraints}
We now exploit the properties of HA selection to tighten  bounds for power constraints \{\eqref{eq: prob. general form d}, \eqref{eq: prob. general form e}, \eqref{eq: prob. general form g},  \eqref{eq: prob. general form h}\}, so that the search region containing the optimal solution is downsized significantly. {\hili For instance, if a certain HA is set to Rx, its effective channel to all DL users is obviously inactive. Intuitively, the corresponding beamforming vector can be arbitrarily chosen in a restricted region while holding the optimality. Therefore, the following lemma is derived to limit the search region as it dynamically tightens the bounds for power constraints.}
\begin{lemma} \label{lem: Bound Tightening}
	Let $ \mathbf{w}_{k,i,j} $ be the $ i $-th half of beamforming vector $ \mathbf{w}_{k,j} $. For a subproblem with $ \omega_{ij}\; (i,j=1,2)$ being set to 0, $ \mathbf{w}_{k,i,j} $ can be fixed at a certain value without any effect on the optimal value. \thmend
\end{lemma}
\begin{proof}
	Please see Appendix \ref{app: Bound Tightening}.
\end{proof}

Using Lemma \ref{lem: Bound Tightening}, we can tighten the bound of $ \mathbf{w}_{k,i,j} $ to reduce the search region when $ \omega_{ij}=0 $. To preserve the optimal value, the feasible set for other beamforming vectors needs to be kept as large as possible, meaning that $ \|\mathbf{w}_{k,i,j}^*\|^2=0 $. For further analysis, we consider the viewpoint that the beamforming is transmitted in one phase, i.e., switching to the HD mode. In the case of the proposed model, the power budget $ P_t^{\text{max}} $ is issued within one time interval of the transmission block. Meanwhile, a traditional HD system switches between UL and DL modes, and each mode occupies one transmission block, meaning that $ P_t^{\text{max}} $ is issued for every two transmission blocks (UL and DL). To make it general, the upper bound of the beamforming vectors in the proposed scheme is tightened through the time scale. In particular, for $ \chi_j \triangleq \omega_{1j}\; \mathbf{nor}\; \omega_{2j}, j=1,2 $, the DL power constraints \eqref{eq: prob. general form g} and \eqref{eq: prob. general form h} can be rewritten~as
\begin{subequations} \label{eq: DL power constraints}
	\vspace{-15pt}
	\begin{align}
		\sum\nolimits_{j=1}^{2}\tau_j^{\dl} \|\mathbf{w}_j\|^2 &\leq P_t^{\text{max}}, \label{eq: DL power constraints a} \\
		\sum\nolimits_{k=1}^{K} \|\mathbf{w}_{k,i,j}\|^2 &\leq P_t^{\infty}\omega_{ij},\; i,j=1,2 \label{eq: DL power constraints b}
	\end{align}
\end{subequations}
where $ \tau_j^{\dl} \triangleq \bar{\tau}_j+\chi_{3-j}\bar{\tau}_{3-j}, j=1,2 $. It is clear that $ \|\mathbf{w}_{k,i,j}\|^2 $ is automatically upper bounded by 0 when $ \omega_{ij}=0 $, and the number of variables for the beamforming vector is thus reduced to half per phase. Similarly, the power constraints for $\ULU$ in \eqref{eq: prob. general form d} and \eqref{eq: prob. general form e} can be expressed as
\begin{subequations} \label{eq: UL power constraints}
	\begin{align}
		& \sum\nolimits_{j=1}^{2}\tau_j^{\ul} p_{\ell,j}^2 \leq P_{\ell}^{\text{max}},\; \forall \ell \in \mathcal{L}, \label{eq: UL power constraints a} \\
		& p_{\ell,j}^2 \leq P_{\ell}^{\infty}\bar{\beta}_{j},\; \forall \ell \in \mathcal{L},j=1,2 \label{eq: UL power constraints b}
	\end{align}
\end{subequations}
where $ \tau_j^{\ul} \triangleq \bar{\tau}_j+\beta_{3-j}\bar{\tau}_{3-j}, j=1,2 $, with $ \beta_j \triangleq \omega_{1j}\; \mathbf{and}\; \omega_{2j}, $ and $ \bar{\beta}_{j} $ is the bit complement of $ \beta_{j} $.

\begin{remark} The power allocation for UL users is upper bounded by 0 only if both HAs are switched to the Tx mode. Therefore, constraints \eqref{eq: UL power constraints b} use $ \bar{\beta}_{j} $ on the right-hand side (RHS), instead of $ \omega_{ij} $ in \eqref{eq: DL power constraints b}. With the appearance of $ \beta_j $ and $ \chi_j $, the HD-switching mode achieves at least the sum rate of the traditional HD system because the proposed method can dynamically set the transmission time for $ \tau_j^{\ul} $ and $ \tau_j^{\dl} $ under  given values of $ \beta_j $ and $ \chi_j $, while the traditional HD system sets the same interval for both UL and DL. 
\end{remark}

Finally, subproblem \eqref{eq: prob. max sum SE 1} is reformulated as
\begin{subequations} \label{eq: prob. max sum SE 1'}
	\vspace{7pt}
	\begin{IEEEeqnarray}{cl}
		\underset{\mathbf{w}, \mathbf{P},\tau}{\text{maximize}} & \quad \sum\nolimits_{\ell=1}^{L}R_{\ell}^{\ul}\bigl(\mathbf{w}, \mathbf{P},\tau|\boldsymbol{\Omega}\bigr) + \sum\nolimits_{k=1}^{K}R_{k}^{\dl}\bigl(\mathbf{w}, \mathbf{P},\tau|\boldsymbol{\Omega}\bigr), \qquad \label{eq: prob. max sum SE 1':objective}\\
		\text{subject to} 
		& \quad \eqref{eq: prob. general form f}, \eqref{eq: prob. general form i}, \eqref{eq: prob. max sum SE 1 c}, \eqref{eq: prob. max sum SE 1 d}, \eqref{eq: DL power constraints}, \eqref{eq: UL power constraints}.
	\end{IEEEeqnarray}
\end{subequations}
For subproblem \eqref{eq: prob. max sum SE 1'} at hand, the power allocation for the optimization in $(\mathbf{w}, \mathbf{P},\tau)$ still remains highly non-convex. Specifically, the objective \eqref{eq: prob. max sum SE 1':objective} is a non-concave function while constraints \eqref{eq: prob. max sum SE 1 c}, \eqref{eq: prob. max sum SE 1 d}, \eqref{eq: DL power constraints a} and \eqref{eq: UL power constraints a} are non-convex.

\subsection{Proposed Convex Approximation-based Iterative Algorithm}
Herein, we will solve subproblem \eqref{eq: prob. max sum SE 1'} via a sequence of convex programs that provides the minorant solution. As in \cite{Tuybook}, a function $ \tilde{h} $ is a convex majorant (or concave minorant, respectively) of a function $ h $ at a point $ \bar{x}\in \mathbf{dom}\; h $ iff $ \tilde{h}(\bar{x})=h(\bar{x}) $ and $ \tilde{h}(x)\geq h(x) $ (or $ \tilde{h}(x)\leq h(x) $, respectively), $ \forall x \in \mathbf{dom}\;h $. According to this definition, subproblem \eqref{eq: prob. max sum SE 1'} can be convexified as follows.

Let us start  by introducing the alternative variables  $ \mu_j, j=1,2$  as
	\begin{IEEEeqnarray}{rCl}
		\bar{\tau}_2 & = &1/\mu_2, \label{eq: altering variables tau mu2}\\
		\bar{\tau}_1 & = &1 - 1/\mu_2 \geq 1/\mu_1. \label{eq: altering variables tau mu1} 
		\vspace{-5pt}
	\end{IEEEeqnarray}
For $ 0<\bar{\tau}_j<1 $, the following additional linear and convex constraints are imposed:
	\begin{IEEEeqnarray}{rCl}
		 \mu_j > 1, j=1,2,&&\label{eq: alternative var. constraint 1}\\
		 1/\mu_1 + 1/\mu_2 \leq 1&&. \label{eq: alternative var. constraint 2}
		\end{IEEEeqnarray}
To handle the non-concave function $ R_{\ell}^{\ul}\bigl(\mathbf{w}, \mathbf{P},\tau|\boldsymbol{\Omega}\bigr)$ in \eqref{eq: prob. max sum SE 1':objective}, we can generally examine $R_{\ell,j}^{\ul}(\mathbf{w}_j,\mathbf{p}_j, \mu_j)\triangleq \frac{1}{\mu_j}\ln\bigl(1+\gamma_{\ell}^{\ul}\bigl(\mathbf{w}_j, \mathbf{p}_j|\boldsymbol{\omega}_j\bigr)\bigr) $ instead of $ \bar{\tau}_j \ln\bigl(1+\gamma_{\ell}^{\ul}\bigl(\mathbf{w}_j, \mathbf{p}_j|\boldsymbol{\omega}_j\bigr)\bigr) $. By the Schur complement, the epigraph of $ \gamma_{\ell}^{\ul}\bigl(\mathbf{w}_j, \mathbf{p}_j|\boldsymbol{\omega}_j\bigr) $ can be expressed with a linear matrix inequality  \cite{Boyd-04-B}. In the spirit of \cite{Dinh:TCOMM:2017}, at a feasible point $(\mathbf{w}_j^{(\kappa)}, \mathbf{p}_j^{(\kappa)},\mu_j^{(\kappa)})$ found at  iteration $\kappa$ of the proposed iterative algorithm, a concave quadratic minorant of $ R_{\ell,j}^{\ul}(\mathbf{w}_j, \mathbf{p}_j,\mu_j)$, denoted by $\mathcal{R}_{\ell, j}^{\ul,(\kappa)}$, is 
	\begin{IEEEeqnarray} {cl} \label{eq: convexifying UL rate}
		 R_{\ell,j}^{\ul}(\mathbf{w}_j, \mathbf{p}_j,\mu_j) \geq  -\varpi(z_{\ell, j}^{(\kappa)},\mu_j^{(\kappa)})\mu_j  + \zeta(z_{\ell, j}^{(\kappa)},\mu_j^{(\kappa)}) - \frac{\varphi_{\ell}^{(\kappa)}(\mathbf{w}_j, \mathbf{p}_j|\boldsymbol{\omega}_j)}{\mu_j^{(\kappa)}} +  \frac{2z_{\ell, j}^{(\kappa)}}{\mu_j^{(\kappa)}p_{\ell,j}^{(\kappa)}}p_{\ell,j}:= \mathcal{R}_{\ell, j}^{\ul,(\kappa)}\qquad
	\end{IEEEeqnarray}
	where 
	\begin{align}
	&	z_{\ell, j}^{(\kappa)}\triangleq\gamma_{\ell}^{\ul}(\mathbf{w}_j^{(\kappa)}, \mathbf{p}_j^{(\kappa)}|\boldsymbol{\omega}_j),\; \varpi(z_{\ell, j}^{(\kappa)},\mu_j^{(\kappa)})\triangleq \frac{\ln\bigl(1+z_{\ell, j}^{(\kappa)}\bigr)}{(\mu_j^{(\kappa)})^2},\;  
		\zeta(z_{\ell, j}^{(\kappa)},\mu_j^{(\kappa)})  \triangleq \frac{2\ln\bigl(1+z_{\ell, j}^{(\kappa)}\bigr) -z_{\ell, j}^{(\kappa)}}{\mu_j^{(\kappa)}}, \nonumber \\
&		\varphi_{\ell}^{(\kappa)}(\mathbf{w}_j, \mathbf{p}_j|\boldsymbol{\omega}_j)  \triangleq
		\ds\tr\Bigl(\bigl(p_{\ell,j}^2\mathbf{\tilde{h}}_{\ell,j}^{\ul}(\mathbf{\tilde{h}}_{\ell,j}^{\ul})^H + \boldsymbol{\Psi}_{\ell,j}\bigl)\bigl( (\boldsymbol{\Psi}_{\ell,j}^{(\kappa)})^{-1} - (\boldsymbol{\Psi}_{(\ell-1),j}^{(\kappa)})^{-1}\bigr)\Bigr),   \nonumber \\
		& \boldsymbol{\Psi}_{\ell,j}^{(\kappa)} \triangleq  \sum\nolimits_{\ell'=\ell+1}^{L} \bigl(p_{\ell',j}^{(\kappa)}\bigr)^2 \mathbf{\tilde{h}}_{\ell',j}^{\ul}\bigl(\mathbf{\tilde{h}}_{\ell',j}^{\ul}\bigr)^H + \rho^2\sum\nolimits_{k=1}^{K} \mathbf{\tilde{G}}_{j}^H \mathbf{w}_{k,j}^{(\kappa)} \bigl(\mathbf{w}_{k,j}^{(\kappa)}\bigr)^H\mathbf{\tilde{G}}_{j} +\sigma_{\mathtt{U}}^2\mathbf{I}.\nonumber
	\end{align}

Next, we address the non-concave function $ R_{k}^{\dl}\bigl(\mathbf{w}, \mathbf{P},\tau|\boldsymbol{\Omega}\bigr) $. For $R_{k,j}^{\dl}(\mathbf{w}_j,\mathbf{p}_j, $ $ \mu_j)\triangleq \frac{1}{\mu_j}\ln\bigl(1+\gamma_{k}^{\dl}(\mathbf{w}_j, \mathbf{p}_j|\boldsymbol{\omega}_j)\bigr) $, it follows that
\begin{equation}\label{eq:dltrans}
R_{k,j}^{\dl}(\mathbf{w}_j,\mathbf{p}_j,  \mu_j) \geq \frac{1}{\mu_j}\ln\bigl(1+\frac{1}{\vartheta_{k,j}}\bigr)
\end{equation}
which imposes the following constraints.
\begin{subequations} \label{eq: vartheta constraints:a}
	\begin{gather}
	\gamma_{k}^{\dl}(\mathbf{w}_j, \mathbf{p}_j|\boldsymbol{\omega}_j) \geq \frac{1}{\vartheta_{k,j}}\; \forall k\in \mathcal{K},\; j=1,2, \label{eq: SOC constraint for SINR DLUs:a} \\
	\vartheta_{k,j}>0,\; \forall k\in \mathcal{K},\; j=1,2 \label{eq: positive vartheta constraint}
	\end{gather}
\end{subequations}
where $ \vartheta_{k,j} $ are  newly introduced variables. Note that constraint \eqref{eq: SOC constraint for SINR DLUs:a} will hold with equality at optimum leading to the equality of \eqref{eq:dltrans}. In addition, constraint \eqref{eq: SOC constraint for SINR DLUs:a} is convexified as
\begin{equation} \label{eq: vartheta constraints}
	\psi_{k,j} \leq \vartheta_{k,j} \gamma_{\text{s}}^{(\kappa)}\bigl(\mathbf{w}_{k,j}\bigr),\; \forall k\in \mathcal{K},\; j=1,2
\end{equation}
over the trust region
\begin{gather}
 \gamma_{\text{s}}^{(\kappa)}(\mathbf{w}_{k,j}) \triangleq 2\Re\{(\mathbf{\tilde{h}}_{k,j}^\dl)^H\mathbf{w}_{k,j}^{(\kappa)}\}\Re\{(\mathbf{\tilde{h}}_{k,j}^\dl)^H \mathbf{w}_{k,j}\}  - \bigl(\Re\{(\mathbf{\tilde{h}}_{k,j}^\dl)^H \mathbf{w}_{k,j}^{(\kappa)}\}\bigr)^2  > 0, \; \forall k, j\label{eq: positive effective channel constraint} 
\end{gather}
where $|(\mathbf{\tilde{h}}_{k,j}^\dl)^H \mathbf{w}_{k,j}|^2$ is innerly approximated by $\gamma_{\text{s}}^{(\kappa)}(\mathbf{w}_{k,j})$ due to its convexity \cite{Dinh:TCOMM:2017}. It can be seen that constraint \eqref{eq: vartheta constraints} is convex and also admits the second-order cone (SOC) representation. Furthermore, the concave  minorant of $R_{k,j}^{\dl}(\mathbf{w}_j,\mathbf{p}_j, $ $ \mu_j)$ in \eqref{eq:dltrans} at the $ (\kappa+1) $-th iteration, denoted by $\mathcal{R}_{k, j}^{\dl, (\kappa)}$, is found as
\begin{align} \label{eq: convexifying DL rate}
R_{k,j}^{\dl}(\mathbf{w}_j,\mathbf{p}_j,  \mu_j) & \geq \nu(\vartheta_{k,j}^{(\kappa)},\mu_j^{(\kappa)}) + \xi(\vartheta_{k,j}^{(\kappa)},\mu_j^{(\kappa)})\vartheta_{k,j} + \lambda(\vartheta_{k,j}^{(\kappa)},\mu_j^{(\kappa)})\mu_j := \mathcal{R}_{k, j}^{\dl, (\kappa)}
\end{align}
where $ \nu(\vartheta_{k,j}^{(\kappa)},\mu_j^{(\kappa)}) $, $ \xi(\vartheta_{k,j}^{(\kappa)},\mu_j^{(\kappa)}) $ and $ \lambda(\vartheta_{k,j}^{(\kappa)},\mu_j^{(\kappa)}) $ are respectively defined as
\begin{align}
\nu(\vartheta_{k,j}^{(\kappa)},\mu_j^{(\kappa)}) & \triangleq \frac{2}{\mu_j^{(\kappa)}}\ln\Bigl(1+\frac{1}{\vartheta_{k,j}^{(\kappa)}}\Bigr) + \frac{1}{\bigl(\vartheta_{k,j}^{(\kappa)}+1\bigr)\mu_j^{(\kappa)}}, \nonumber \quad
\xi(\vartheta_{k,j}^{(\kappa)},\mu_j^{(\kappa)}) \triangleq -\frac{1}{\vartheta_{k,j}^{(\kappa)}\bigl(\vartheta_{k,j}^{(\kappa)}+1\bigr)\mu_j^{(\kappa)}}, \nonumber \\
\lambda(\vartheta_{k,j}^{(\kappa)},\mu_j^{(\kappa)}) & \triangleq - \frac{1}{\bigl(\mu_j^{(\kappa)}\bigr)^2}\ln\Bigl(1+\frac{1}{\vartheta_{k,j}^{(\kappa)}}\Bigr). \nonumber
\end{align}
Obviously, we can iteratively replace $R_{\ell}^{\ul}\bigl(\mathbf{w}, \mathbf{P},\tau|\boldsymbol{\Omega}\bigr)$ and $R_{k}^{\dl}\bigl(\mathbf{w}, \mathbf{P},\tau|\boldsymbol{\Omega}\bigr)$ by $\sum_{j=1}^2 \mathcal{R}_{\ell,j}^{\ul, (\kappa)} $ and $\sum_{j=1}^2  \mathcal{R}_{k,j}^{\dl, (\kappa)} $ to arrive at the convex constraints for \eqref{eq: prob. max sum SE 1 c} and \eqref{eq: prob. max sum SE 1 d}, respectively.

Finally, let us convexify the power constraints \eqref{eq: DL power constraints a} and \eqref{eq: UL power constraints a}. By substituting \eqref{eq: altering variables tau mu2} into \eqref{eq: DL power constraints a} and \eqref{eq: UL power constraints a}, we obtain two new constraints as
	\begin{subequations} \label{eq: power constraints - altered variables}
	\begin{gather} 
		\mathcal{T}_1(\chi_2)\|\mathbf{w}_1\|^2  + \mathcal{T}_2(\chi_1)\|\mathbf{w}_2\|^2 \leq P_t^{\text{max}} + f_{\text{c}}^{d}\bigl(\mathbf{w}, \mu_2 \bigr), \label{eq: power constraints - altered variables - downlink} \\
		\mathcal{T}_1(\beta_2)p_{\ell,1}^2 + \mathcal{T}_2(\beta_1) p_{\ell,2}^2 \leq P_{\ell}^{\text{max}} + f_{\text{c}}^{u}\bigl(\mathbf{\hat{p}}_{\ell}, \mu_2 \bigr), \; \forall \ell \in \mathcal{L} \label{eq: power constraints - altered variables - uplink}
	\end{gather}
	\end{subequations}
where $ \mathcal{T}_1(x) = \bigl(1+\frac{x}{\mu_2}\bigr) $ and $ \mathcal{T}_2(x) = \bigl(\frac{1}{\mu_2}+x\bigr) $, while $ \chi_j $ and $ \beta_j $, $ j=1,2 $ are  constant, depending on a fixed value of $ \boldsymbol{\Omega} $. We define $ f_{\text{c}}^{d}\bigl(\mathbf{w}, \mu_2 \bigr) \triangleq \frac{1}{\mu_2}\|\mathbf{w}_1\|^2 + \frac{\chi_1}{\mu_2}\|\mathbf{w}_2\|^2 $ and $ f_{\text{c}}^{u}\bigl(\mathbf{\hat{p}}_{\ell}, \mu_2 \bigr) \triangleq \frac{1}{\mu_2}p_{\ell,1}^2 +  \frac{\beta_1}{\mu_2}p_{\ell,2}^2 $, where $ \mathbf{\hat{p}}_{\ell} \in \mathbb{R}^{2\times1} $ is the transpose of the $ \ell $-th row of $ \mathbf{P} $. We remark that both left-hand sides of \eqref{eq: power constraints - altered variables} are convex while their RHSs are quadratic-over-linear functions. By applying the first-order approximation to $ f_{\text{c}}^{d}\bigl(\mathbf{w}, \mu_2 \bigr) $ around the point $(\mathbf{w}^{(\kappa)}, \mu_2^{(\kappa)} \bigr)$ and $ f_{\text{c}}^{u}\bigl(\mathbf{\hat{p}}_{\ell}, \mu_2 \bigr) $ around the point $\bigl(\mathbf{\hat{p}}_{\ell}^{(\kappa)}, \mu_2^{(\kappa)} \bigr)$, the functions $ f_{\text{c}}^{d}\bigl(\mathbf{w}, \mu_2 \bigr) $  and  $ f_{\text{c}}^{u}\bigl(\mathbf{\hat{p}}_{\ell}, \mu_2 \bigr)  $ are innerly approximated as
\begingroup\allowdisplaybreaks
	\begin{subequations} \label{eq: appox. power contraints}
	\begin{align}
		f_{\text{c}}^{d}\bigl(\mathbf{w}, \mu_2 \bigr) & \geq f_{\text{c}}^{d}\bigl(\mathbf{w}^{(\kappa)}, \mu_2^{(\kappa)} \bigr) + \Bigl<\nabla f_{\text{c}}^{d}\bigl(\mathbf{w}^{(\kappa)}, \mu_2^{(\kappa)}\bigr),\bigl(\mathbf{w}, \mu_2 \bigr)-\bigl(\mathbf{w}^{(\kappa)}, \mu_2^{(\kappa)} \bigr) \Bigr> \nonumber \\
		& = 2\frac{\Re\{\bigl(\mathbf{w}_1^{(\kappa)}\bigr)^H\mathbf{w}_1\}}{\mu_2^{(\kappa)}} + 2\frac{\chi_1\Re\{\bigl(\mathbf{w}_2^{(\kappa)}\bigr)^H\mathbf{w}_2\}}{\mu_2^{(\kappa)}} - \frac{f_{\text{c}}^{d}\bigl(\mathbf{w}^{(\kappa)}, \mu_2^{(\kappa)} \bigr)}{\mu_2^{(\kappa)}}\mu_2 \nonumber \\ 
		& := \hat{f}_{\text{c}}^{d,(\kappa)}\bigl(\mathbf{w}, \mu_2 \bigr) \label{eq: appox. power contraint - downlink} \\
		f_{\text{c}}^{u}\bigl(\mathbf{\hat{p}}_{\ell}, \mu_2 \bigr) & \geq f_{\text{c}}^{u}\bigl(\mathbf{\hat{p}}_{\ell}^{(\kappa)}, \mu_2^{(\kappa)} \bigr) + \Bigl<\nabla f_{\text{c}}^{u}\bigl(\mathbf{\hat{p}}_{\ell}^{(\kappa)}, \mu_2^{(\kappa)} \bigr),\bigl(\mathbf{\hat{p}}_{\ell}, \mu_2 \bigr)-\bigl(\mathbf{\hat{p}}_{\ell}^{(\kappa)}, \mu_2^{(\kappa)} \bigr) \Bigr> \nonumber \\
		& = 2\frac{p_{\ell,1}^{(\kappa)}}{\mu_2^{(\kappa)}}p_{\ell,1} + 2\frac{\beta_1 p_{\ell,2}^{(\kappa)}}{\mu_2^{(\kappa)}}p_{\ell,2} - \frac{f_{\text{c}}^{u}\bigl(\mathbf{\hat{p}}_{\ell}^{(\kappa)}, \mu_2^{(\kappa)} \bigr)}{\mu_2^{(\kappa)}}\mu_2 := \hat{f}_{\text{c}}^{u,(\kappa)}\bigl(\mathbf{\hat{p}}_{\ell}, \mu_2 \bigr). \label{eq: appox. power contraint - uplink}
	\end{align}
	\end{subequations}
\endgroup
As a result,  the convex constraints for \eqref{eq: power constraints - altered variables} read as
	\begin{subequations} \label{eq: power constraints approx. - altered variables}
		\begin{gather}
		\mathcal{T}_1(\chi_2)\|\mathbf{w}_1\|^2  + \mathcal{T}_2(\chi_1)\|\mathbf{w}_2\|^2 \leq P_t^{\text{max}} + \hat{f}_{\text{c}}^{d,(\kappa)}\bigl(\mathbf{w}, \mu_2 \bigr), \label{eq: power constraints approx. - altered variables - downlink} \\
		\mathcal{T}_1(\beta_2)p_{\ell,1}^2 + \mathcal{T}_2(\beta_1) p_{\ell,2}^2 \leq P_{\ell}^{\text{max}} + \hat{f}_{\text{c}}^{u,(\kappa)}\bigl(\mathbf{\hat{p}}_{\ell}, \mu_2 \bigr). \label{eq: power constraints approx. - altered variables - uplink}
		\end{gather}
	\end{subequations}

In summary, the convex program providing a minorant maximization for \eqref{eq: prob. max sum SE 1'} solved at  iteration $(\kappa+1)$ of the proposed algorithm is given by
\begin{subequations} \label{eq: prob. max sum SE 2}
	\begin{IEEEeqnarray}{cl}
		\underset{\mathbf{w}, \mathbf{P},\boldsymbol{\mu}, \boldsymbol{\vartheta}}{\text{maximize}} & \quad \mathcal{R}_{\Sigma}^{(\kappa+1)} \triangleq \sum\nolimits_{\ell=1}^{L}\sum\nolimits_{j=1}^{2}\mathcal{R}_{\ell, j}^{\ul, (\kappa)} + \sum\nolimits_{k=1}^{K}\sum\nolimits_{j=1}^{2}\mathcal{R}_{k, j}^{\dl, (\kappa)} \label{eq: prob. max sum SE 2 a} \\
		\text{subject to} 
		& \quad \eqref{eq: prob. general form f}, \eqref{eq: DL power constraints b}, \eqref{eq: UL power constraints b}, \eqref{eq: alternative var. constraint 1}, \eqref{eq: alternative var. constraint 2}, \eqref{eq: positive vartheta constraint},  \eqref{eq: vartheta constraints}, \eqref{eq: positive effective channel constraint}, \eqref{eq: power constraints approx. - altered variables}, \label{eq: prob. max sum SE 2 b} \\
		& \quad \sum\nolimits_{j=1}^{2}\mathcal{R}_{\ell, j}^{\ul, (\kappa)} \geq \bar{R}_\ell^{\ul}, \forall \ell \in \mathcal{L}, \label{eq: prob. max sum SE 2 c} \\
		& \quad \sum\nolimits_{j=1}^{2}\mathcal{R}_{k, j}^{\dl, (\kappa)} \geq \bar{R}_k^{\dl}, \forall k \in \mathcal{K} \label{eq: prob. max sum SE 2 d}
	\end{IEEEeqnarray}
\end{subequations}	
where $ \boldsymbol{\mu} \triangleq [\mu_1\; \mu_2] $ and $ \boldsymbol{\vartheta} \triangleq [\vartheta_{k,j}]_{k\in \mathcal{K}, j=1,2} $. Note that the convex program \eqref{eq: prob. max sum SE 2} can be efficiently solved per iteration by existing  solvers (e.g., MOSEK \cite{MOSEK} or SeDuMi \cite{SeDuMi:2002}) and the optimal solution for  \eqref{eq: prob. max sum SE 1'} is  one of eight cases of $ \boldsymbol{\Omega} $ in \eqref{eq: Omega matrix - HA assignment}. By comparing the eight optimal values of  subproblems, the global solution $ \bigl(\boldsymbol{\Omega}^{*}, \mathbf{W}^{*}, \mathbf{P}^{*}, \tau^*\bigr) $ for \eqref{eq: prob. max sum SE} is derived, and it corresponds to the maximum among eight optimal values. The proposed iterative algorithm for solving  problem \eqref{eq: prob. max sum SE} is briefly described in Algorithm \ref{alg: for max sum SE}.

\begin{algorithm}[t]
	\begin{algorithmic}[1]
		\fontsize{10}{12}\selectfont
		\protect\caption{Proposed Iterative Algorithm for the SR Maximization Problem \eqref{eq: prob. max sum SE}}

		\label{alg: for max sum SE}
		
%		\global\long\def\algorithmicrequire{\textbf{Initialization:}}
		\STATE \textbf{Initialization:} Set $ \mathcal{R}_{\text{SR}}^{\max} := -\infty $, $ \varepsilon:=10^{-3} $, $ \gamma_{\varepsilon} := 10^{-3}$, $ \boldsymbol{\alpha}_{\ell}^{\ul} = \boldsymbol{\alpha}_k^{\dl} := \mathbf{1} $, and $ \bigl(\boldsymbol{\Omega}^{*}, \mathbf{w}^{*},\mathbf{P}^{*},\tau^{*}\bigr)$\footnotemark $:=\mathbf{0} $.
		
		\FOR[solving subproblem \eqref{eq: prob. max sum SE 1'}] {each $ \boldsymbol{\Omega} \in \mathcal{S}_{\boldsymbol{\Omega}} $}	
		
		\STATE\textbf{Generating an initial point:} Set $\kappa:=0$ and solve \eqref{eq: prob. max sum SE - initialization} to generate $ (\mathbf{w}^{(0)},\mathbf{P}^{(0)},\boldsymbol{\mu}^{(0)},\boldsymbol{\vartheta}^{(0)}) $.
		
		\REPEAT
		\STATE Solve \eqref{eq: prob. max sum SE 2} to obtain  $ \bigl(\mathbf{w}^{\star},\mathbf{P}^{\star},\boldsymbol{\mu}^{\star},\boldsymbol{\vartheta}^{\star}\bigr) $ and  $\mathcal{R}_{\Sigma}^{(\kappa+1)}$.
		
		\STATE Update $(\mathbf{w}^{(\kappa+1)},\mathbf{P}^{(\kappa+1)}, \boldsymbol{\mu}^{(\kappa+1)}, \boldsymbol{\vartheta}^{(\kappa+1)}) :=(\mathbf{w}^{\star},   \mathbf{P}^{\star}, \boldsymbol{\mu}^{\star}, \boldsymbol{\vartheta}^{\star} $).
		\STATE Set $ \kappa := \kappa + 1 $.
		\UNTIL Convergence
		\IF[if true, then there exists a better solution for \eqref{eq: prob. max sum SE}] {$\mathcal{R}_{\Sigma}^{(\kappa)}>\mathcal{R}_{\text{SR}}^{\max} $}
		\STATE Update $ \mathcal{R}_{\text{SR}}^{\max}:= \mathcal{R}_{\Sigma}^{(\kappa)} $ and $ \bigl(\boldsymbol{\Omega}^{*}, \mathbf{w}^{*},\mathbf{P}^{*},\tau^{*}\bigr) := \bigl(\boldsymbol{\Omega}, \mathbf{w}^{(\kappa)},\mathbf{P}^{(\kappa)},(1-1/\mu_2^{(\kappa)} ) \bigr) $.
		\STATE Update $ \boldsymbol{\alpha}_{\ell}^{\ul} $ and $ \boldsymbol{\alpha}_k^{\dl} $ as in \eqref{eq: users assignments}, and  $ R_{\Sigma} $ in \eqref{eq: prob. max sum SE a}.
		\ENDIF
		\ENDFOR
%		\STATE {\color{red}\textbf{Output:} The objective value $ \mathcal{R}_{\Sigma} $, and optimal solution $ \bigl(\boldsymbol{\Omega}^{*}, \mathbf{w}^{*},\mathbf{P}^{*},\tau^{*},\boldsymbol{\alpha}_{\ell}^{\ul}, \boldsymbol{\alpha}_k^{\dl}\bigr) $.}
	\end{algorithmic} 
\end{algorithm}
\footnotetext{Note that $ \bigl(\boldsymbol{\Omega}^{*}, \mathbf{w}^{*},\mathbf{P}^{*},\tau^{*}\bigr) $ represents the optimal solution for \eqref{eq: prob. max sum SE}, while $ \bigl(\mathbf{w}^{\star},\mathbf{P}^{\star},\boldsymbol{\mu}^{\star},\boldsymbol{\vartheta}^{\star}\bigr) $ only denotes a per-iteration optimal solution for subproblem \eqref{eq: prob. max sum SE 2}.}

\textit{Generating an initial feasible point:} To execute the loop for solving \eqref{eq: prob. max sum SE 1'}, Algorithm \ref{alg: for max sum SE} needs to initialize a feasible point $ (\mathbf{w}^{(0)},\mathbf{P}^{(0)},\boldsymbol{\mu}^{(0)},\boldsymbol{\vartheta}^{(0)}) $ by consecutively solving a simpler problem as
%\begingroup\allowdisplaybreaks
\begin{subequations} \label{eq: prob. max sum SE - initialization}
	\begin{IEEEeqnarray}{cl}
		\underset{\mathbf{w}, \mathbf{P},\boldsymbol{\mu}, \boldsymbol{\vartheta}, \varrho}{\text{maximize}} & \quad \varrho \label{eq: prob. max sum SE - initialization a} \\
		\text{subject to} 
		& \quad \eqref{eq: prob. max sum SE 2 b}, \\
		& \quad \sum\nolimits_{j=1}^{2}\mathcal{R}_{\ell, j}^{\ul, (\kappa)}-\bar{R}_\ell^{\ul} \geq \varrho, \forall \ell \in \mathcal{L}, \label{eq: prob. max sum SE - initialization c} \\
		& \quad \sum\nolimits_{j=1}^{2}\mathcal{R}_{k, j}^{\dl, (\kappa)}-\bar{R}_k^{\dl} \geq \varrho, \forall k \in \mathcal{K} \label{eq: prob. max sum SE - initialization d}
	\end{IEEEeqnarray}
\end{subequations}
%\endgroup
where $ \varrho $ is a new variable.  The goal of solving \eqref{eq: prob. max sum SE - initialization} is not to focus on maximizing $ \varrho $, but rather to find a feasible point for \eqref{eq: prob. max sum SE 2}. Simply, the break condition for generating a starting point in \eqref{eq: prob. max sum SE - initialization} is $ \varrho > 0 $.

\textit{Convergence analysis:} As in \cite{Boyd-04-B}, the proposed algorithm obtains the optimal point satisfying the Karush-Kuhn-Tucker (KKT) conditions or KKT-invexity. In fact, since Algorithm \ref{alg: for max sum SE} is based on the inner approximation, it is true that for the sequence of convex programs  in \eqref{eq: prob. max sum SE 2} every point in a feasible set $ \mathcal{F}^{(\kappa)} $ of the considered problem at  iteration $ \kappa $ is also feasible for the problem at iteration $ (\kappa+1) $, i.e., $ \mathcal{F}^{(\kappa+1)} \supseteq \mathcal{F}^{(\kappa)} $ \cite{Marks:78}. This means that $ \mathcal{F}^{(\kappa)}, \forall \kappa $, is considered as a convex subset in the order topology, and thus the set $ \mathcal{F}^{(\kappa)} $ is connected \cite{Munkres:Topo}. According to \cite{Bestuzheva:InvexOp:Jul2017}, it satisfies the KKT-invex condition when $ \kappa \rightarrow \infty $. As a result, the optimal solution of this sequence is monotonically improved toward the KKT point. For the practical implementation, Algorithm \ref{alg: for max sum SE} terminates upon reaching $ \mathcal{R}_{\Sigma}^{(\kappa+1)}-\mathcal{R}_{\Sigma}^{(\kappa)}<\varepsilon $ after a finite number of iterations.

\textit{Complexity analysis:} The proposed method of solving  \eqref{eq: prob. max sum SE 2} has a low complexity since it only contains the SOC and linear constraints. In particular, it takes a polynomial time complexity $ \mathcal{O}\bigl((6L+7K+8)^{2.5}(2L+(4N+2)K+2)^2+(6L+7K+8)^{3.5}\bigr) $, resulting from $ 6L+7K+8 $ constraints and $ 2L+(4N+2)K+2 $ scalar variables \cite{SeDuMi:2002}.

\section{Proposed Algorithm for Max-Min Rate Optimization}\label{sec: max-min rate optimization}

The BS in problem \eqref{eq: prob. max sum SE} will favor users with good channel conditions to maximize the total SR whenever QoS requirements for all users are satisfied. It merely means that users with poor channel conditions will be easily disconnected from the service when their QoS requirements become stringent.  In order to support the most vulnerable users,   the following general max-min optimization problem  is considered:
\begin{subequations} \label{eq: prob. max-min SE}
	\begin{IEEEeqnarray}{cl}
		\underset{\boldsymbol{\Omega},\mathbf{w}, \mathbf{P},\tau,\boldsymbol{\alpha}_{\ell}^{\ul},\boldsymbol{\alpha}_{k}^{\dl}}{\text{maximize}} & \quad R_{\mathtt{min}}\triangleq \underset{{\ell\in\mathcal{L},k \in \mathcal{K}}}{\min}\bigl\{R_{\ell}^{\ul},\frac{1}{\eta}R_{k}^{\dl}\bigr\} \label{eq: prob. max-min SE a} \\
		\text{subject to} 
		& \quad \eqref{eq: prob. general form b}-\eqref{eq: prob. general form j}
	\end{IEEEeqnarray}
\end{subequations}	
where  the objective function $ \mathcal{R}\bigl(\{R_{\ell}^{\ul}\},\{R_{k}^{\dl}\}\bigr)$ in \eqref{eq: prob. general form} is replaced by  $ R_{\mathtt{min}}= \underset{{\ell\in\mathcal{L},k \in \mathcal{K}}}{\min}\bigl\{R_{\ell}^{\ul},\frac{1}{\eta}R_{k}^{\dl}\bigr\} $. $\eta\geq 1 $   denotes the asymmetric coefficient on DL  compared to UL since the required transmission rate of UL is often less than that of DL  \cite{3GPP}. The special case of $\eta =  1$ was also mentioned in \cite{Nguyen:JSAC:18}, but in a different context.

Similarly to  \eqref{eq: prob. max sum SE}, problem \eqref{eq: prob. max-min SE} also belongs to  the difficult class of mixed-integer  non-convex programming. Fortunately, the developments presented in Algorithm \ref{alg: for max sum SE} is readily extendable to solve the  max-min problem \eqref{eq: prob. max-min SE}. Specifically, by following the similar steps given in Section III, the  subproblems for \eqref{eq: prob. max-min SE} are given by
	\begin{subequations} \label{eq: prob. max min SE 1'}
	\vspace{7pt}
	\begin{IEEEeqnarray}{cl}
		\underset{\mathbf{w}, \mathbf{P},\tau}{\text{maximize}} & \quad \underset{{\ell\in\mathcal{L},k \in \mathcal{K}}}{\min}\bigl\{R_{\ell}^{\ul}\bigl(\mathbf{w}, \mathbf{P},\tau|\boldsymbol{\Omega}\bigr),\frac{1}{\eta}R_{k}^{\dl}\bigl(\mathbf{w}, \mathbf{P},\tau|\boldsymbol{\Omega}\bigr)\bigr\} 
		\qquad \label{eq: prob. max min SE 1':objective}\\
		\text{subject to} 
		& \quad \eqref{eq: prob. general form f},  \eqref{eq: DL power constraints}, \eqref{eq: UL power constraints}
	\end{IEEEeqnarray}
\end{subequations}
which are equivalently rewritten as
%\begingroup\allowdisplaybreaks
\begin{subequations} \label{eq: prob. max-min SE:eqi}
	\begin{IEEEeqnarray}{cl}
		\underset{\mathbf{w}, \mathbf{P},\tau,\phi}{\text{maximize}} & \qquad \phi \label{eq: prob. max-min SE a:eqi} \\
		\text{subject to} 
		& \quad \eqref{eq: prob. general form f},  \eqref{eq: DL power constraints}, \eqref{eq: UL power constraints}, \label{eq: prob. max-min SE b:eqi}\\
		&\quad R_{\ell}^{\ul}\bigl(\mathbf{w}, \mathbf{P},\tau|\boldsymbol{\Omega}\bigr) \geq \phi, \forall \ell \in \mathcal{L}, \label{eq: prob. max-min SE c:eqi} \\
		&\quad \frac{1}{\eta}R_{k}^{\dl}\bigl(\mathbf{w}, \mathbf{P},\tau|\boldsymbol{\Omega}\bigr) \geq \phi, \forall k \in \mathcal{K}\label{eq: prob. max-min SE d:eqi}
	\end{IEEEeqnarray}
\end{subequations}%\endgroup
where $ \phi $ is a newly introduced slack variable. The equivalence between \eqref{eq: prob. max min SE 1'} and \eqref{eq: prob. max-min SE:eqi} is guaranteed since constraints \eqref{eq: prob. max-min SE c:eqi} and \eqref{eq: prob. max-min SE d:eqi} must hold with equality at optimum. Clearly, the non-concave functions $R_{\ell}^{\ul}\bigl(\mathbf{w}, \mathbf{P},\tau|\boldsymbol{\Omega}\bigr)$ and $R_{k}^{\dl}\bigl(\mathbf{w}, \mathbf{P},\tau|\boldsymbol{\Omega}\bigr)$ were  addressed in \eqref{eq: convexifying UL rate} and \eqref{eq: convexifying DL rate}, respectively, while the non-convex constraints \eqref{eq: DL power constraints a} and \eqref{eq: UL power constraints a} were convexified in \eqref{eq: power constraints approx. - altered variables}. 
Consequently, at the $ (\kappa+1) $-th iteration, we solve the following  convex  program to provide  a minorant maximization for \eqref{eq: prob. max-min SE:eqi}:
	\begin{subequations} \label{eq: prob. max-min SE 1}
		\begin{IEEEeqnarray}{cl}
			\underset{\mathbf{w}, \mathbf{P},\boldsymbol{\mu}, \boldsymbol{\vartheta}, \phi}{\text{maximize}} & \quad \phi \label{eq: prob. max-min SE 1 a} \\
			\text{subject to} 
			& \quad \eqref{eq: prob. max sum SE 2 b}, \\
			& \quad \sum\nolimits_{j=1}^{2}\mathcal{R}_{\ell, j}^{\ul, (\kappa)} \geq \phi, \forall \ell \in \mathcal{L}, \label{eq: prob. max-min SE 1 c} \\
			& \quad \sum\nolimits_{j=1}^{2}\mathcal{R}_{k, j}^{\dl, (\kappa)} \geq \eta \phi, \forall k \in \mathcal{K}. \label{eq: prob. max-min SE 1 d}
		\end{IEEEeqnarray}
	\end{subequations}
The proposed algorithm to solve \eqref{eq: prob. max-min SE} is summarized in Algorithm \ref{alg: for max min SE}. Similarly to Algorithm \ref{alg: for max sum SE}, it can be shown that once initialized from a feasible point (i.e., Step 3), Algorithm \ref{alg: for max min SE} converges after finitely many iterations upon reaching $ |\phi^{(\kappa)}-\phi^{(\kappa-1)}| < \varepsilon $.

\begin{algorithm}[t]
	\begin{algorithmic}[1]
		\fontsize{11}{12}\selectfont
		\protect\caption{Proposed Iterative Algorithm for the Max-Min  Problem \eqref{eq: prob. max-min SE}}

		\label{alg: for max min SE}
		
		%		\global\long\def\algorithmicrequire{\textbf{Initialization:}}
		\STATE \textbf{Initialization:} Set $ \phi^{\max} := -\infty $, $ \varepsilon:=10^{-3} $,  $ \gamma_{\varepsilon} := 10^{-3}$,   $ \boldsymbol{\alpha}_{\ell}^{\ul} = \boldsymbol{\alpha}_k^{\dl} := \mathbf{1} $, and $ \bigl(\boldsymbol{\Omega}^{*}, \mathbf{w}^{*},\mathbf{P}^{*},\tau^{*}\bigr) :=\mathbf{0}$.
		
		\FOR[solving subproblem \eqref{eq: prob. max-min SE:eqi}] {each $ \boldsymbol{\Omega} \in \mathcal{S}_{\boldsymbol{\Omega}} $}	
		
		\STATE\textbf{Generating an initial point:} Set $ \kappa:=0 $, $ \phi^{(0)}:=0 $ and generate $ (\mathbf{w}^{(0)},\mathbf{P}^{(0)},\boldsymbol{\mu}^{(0)},\boldsymbol{\vartheta}^{(0)}) $.
				\REPEAT
		\STATE Solve \eqref{eq: prob. max-min SE 1} to obtain  $ \bigl(\mathbf{w}^{\star},\mathbf{P}^{\star},\boldsymbol{\mu}^{\star},\boldsymbol{\vartheta}^{\star}, \phi^{\star}\bigr) $.
		
		%\STATE {\color{red}Calculate $ \Delta\phi:=|\phi^{\star}-\phi^{(\kappa)}| $.}
		
		\STATE Update $(\mathbf{w}^{(\kappa+1)},\mathbf{P}^{(\kappa+1)}, \boldsymbol{\mu}^{(\kappa+1)}, \boldsymbol{\vartheta}^{(\kappa+1)}) :=(\mathbf{w}^{\star},   \mathbf{P}^{\star}, \boldsymbol{\mu}^{\star}, \boldsymbol{\vartheta}^{\star} $).
		\STATE Set $ \kappa := \kappa + 1 $.
		\UNTIL $ |\phi^{(\kappa)}-\phi^{(\kappa-1)}| < \varepsilon $
		\IF[if true, then there exists a better solution for  \eqref{eq: prob. max-min SE}] {$ \phi^{(\kappa)}>\phi^{\max} $}
		\STATE Update $ \phi^{\max}:= \phi^{(\kappa)} $ and $ \bigl(\boldsymbol{\Omega}^{*}, \mathbf{w}^{*},\mathbf{P}^{*},\tau^{*}\bigr) := \bigl(\boldsymbol{\Omega}, \mathbf{w}^{(\kappa)},\mathbf{P}^{(\kappa)},\bigl(1-1/\mu_2^{(\kappa)} \bigr) \bigr) $.
		\STATE Update $ \boldsymbol{\alpha}_{\ell}^{\ul} $ and $ \boldsymbol{\alpha}_k^{\dl} $ as in \eqref{eq: users assignments}, and  $ R_{\mathtt{min}} $ in \eqref{eq: prob. max-min SE a}.
		\ENDIF
		\ENDFOR
%		\STATE {\color{red}\textbf{Output:} The objective value $ \mathcal{R}_{\mathtt{m}} $, and optimal solution $ \bigl(\boldsymbol{\Omega}^{*}, \mathbf{w}^{*},\mathbf{P}^{*},\tau^{*},\boldsymbol{\alpha}_{\ell}^{\ul}, \boldsymbol{\alpha}_k^{\dl}\bigr) $.}
	\end{algorithmic} 
\end{algorithm}

\textit{Complexity analysis:} Compared to \eqref{eq: prob. max sum SE 2}, the max-min  problem  \eqref{eq: prob. max-min SE 1} has the same number of constraints, but it requires one additional variable $ \phi $. In each iteration of Algorithm \ref{alg: for max min SE}, the  computational complexity of solving \eqref{eq: prob. max-min SE 1} is thus $ \mathcal{O}\bigl((6L+7K+8)^{2.5}(2L+(4N+2)K+3)^2+(6L+7K+8)^{3.5}\bigr) $.

\section{Robust Transmission Design under Channel Uncertainty}\label{sec: Channel Uncertainty}
The solutions of the previous sections require perfect CSI of all channels, which rely on the channel estimations at the BS and users. In practice, it is important to realize that the perfect CSI assumption is  highly idealistic due to limited information exchange between the transceivers. Therefore, our next endeavor is to consider a robust design for the SR maximization problem \eqref{eq: prob. max sum SE} under the imperfect CSI assumption as that for the max-min problem \eqref{eq: prob. max-min SE} is straightforward. For the stochastic error model \cite{Maurer:TSP:Jan2011}, the channels can be modeled as
\begin{IEEEeqnarray}{rCl}\label{channelmodeling}
			\mathbf{h}_{\ell}^{\ul} =  \mathbf{\hat{h}}_{\ell}^{\ul} + \boldsymbol{\Delta}\mathbf{h}_{\ell}^{\ul},\;
			\mathbf{h}_{k}^{\dl} =  \mathbf{\hat{h}}_{k}^{\dl} + \boldsymbol{\Delta}\mathbf{h}_{k}^{\dl},\;\text{and}\;
			g_{\ell k} = \hat{g}_{\ell k} + \Delta g_{\ell k}
		\end{IEEEeqnarray}
where $(\mathbf{h}_{\ell}^{\ul},\mathbf{h}_{k}^{\dl},g_{\ell k})$ are the perfect channels; $(\mathbf{\hat{h}}_{\ell}^{\ul},\mathbf{\hat{h}}_{k}^{\dl},\hat{g}_{\ell k})$ are the channel estimates; $(\boldsymbol{\Delta}\mathbf{h}_{\ell}^{\ul},\boldsymbol{\Delta}\mathbf{h}_{k}^{\dl},$ $\Delta g_{\ell k})$ represent the CSI errors which are assumed to be independent of the channel estimates and distributed as $ \boldsymbol{\Delta}\mathbf{h}_{\ell}^{\ul} \sim \mathcal{CN}(0,\epsilon_{\ell}^{\ul} \mathbf{I}) $, $ \boldsymbol{\Delta}\mathbf{h}_{k}^{\dl} \sim \mathcal{CN}(0,\epsilon_{k}^{\dl} \mathbf{I}) $, and $ \Delta g_{\ell k} \sim \mathcal{CN}(0,\epsilon_{\ell,k}) $, respectively, where $(\epsilon_{\ell}^{\ul},\epsilon_{k}^{\dl},\epsilon_{\ell,k})$ denote the variances of the CSI errors caused by estimation inaccuracies \cite{Maurer:TSP:Jan2011}. The SI channel $\mathbf{G}_{\SI}$ is assumed to be the same as before, since the Tx and Rx antennas are co-located at the BS.

Herein, we consider a worst-case robust formulation by  treating the CSI errors as noise. For UL transmission, let $ \mathbf{\hat{y}}_{\ell,j}^{\ul}$ be the post-SIC received signal of $\ULU$ subject to the channel uncertainties in \eqref{channelmodeling}. The  $\ULU$'s signal received at the BS is estimated as $ \hat{x}_{\ell,j}^{\ul} = \mathbf{u}_{\ell, j}^H\mathbf{\hat{y}}_{\ell,j}^{\ul} $, where $ \mathbf{u}_{\ell, j} $ is an optimal weight vector that minimizes the MSE, $ \mathbb{E}\bigl[|x_{\ell,j}^{\ul}-\hat{x}_{\ell,j}^{\ul}|^2\bigr] $.

\begin{lemma}
	The optimal weight vector $ \mathbf{u}_{\ell, j} $ is derived as
	\begin{align} \label{eq: optimal weight vector}
		\mathbf{u}_{\ell, j} = p_{\ell, j}\bigl(p_{\ell,j}^2  \mathbf{\bar{h}}_{\ell,j}^{\ul}\bigl(\mathbf{\bar{h}}_{\ell,j}^{\ul}\bigr)^H + \boldsymbol{\hat{\Psi}}_{\ell,j} \bigr)^{-1}\mathbf{\bar{h}}_{\ell,j}^{\ul}
	\end{align}
	where $ \boldsymbol{\hat{\Psi}}_{\ell,j}\triangleq\sum_{\ell'=\ell+1}^{L}p_{\ell',j}^2  \mathbf{\bar{h}}_{\ell,j}^{\ul}\bigl(\mathbf{\bar{h}}_{\ell,j}^{\ul}\bigr)^H + \sum_{\ell'=1}^{L}p_{\ell',j}^2 \epsilon_{\ell'}^{\ul}\boldsymbol{\bar{\Lambda}}_j + \rho^2\sum_{k=1}^{K} \mathbf{\tilde{G}}_{j}^H \mathbf{w}_{k,j} \mathbf{w}_{k,j}^H\mathbf{\tilde{G}}_{j} + \sigma_{\mathtt{U}}^2\mathbf{I} $ and $ \mathbf{\bar{h}}_{\ell,j}^{\ul} \triangleq \boldsymbol{\bar{\Lambda}}_j \mathbf{\hat{h}}_{\ell}^{\ul} $.  The SINR of $\ULU$ in phase $ j $ under channel uncertainty can be expressed as
	\begin{align} \label{eq: SINR ULUs - chan. uncertainty}
		\hat{\gamma}_{\ell}^{\ul}(\boldsymbol{\omega}_j,\mathbf{w}_j , \mathbf{p}_j) = p_{\ell,j}^2 \bigl(\mathbf{\bar{h}}_{\ell,j}^{\ul}\bigr)^H\boldsymbol{\hat{\Psi}}_{\ell,j}^{-1}\mathbf{\bar{h}}_{\ell,j}^{\ul}.
	\end{align}
\end{lemma}
\begin{proof}
	Please see Appendix \ref{app: UL SINR}.
\end{proof}
It can be seen that if the estimation error is $ \epsilon_{\ell}^{\ul}=0 $,  \eqref{eq: SINR ULUs - chan. uncertainty} reduces to \eqref{eq: SINR ULUs}. For  DL transmission,  the SINR of $\DLU$ in  phase $j$ can be formulated as
\begin{align} \label{eq: SINR DLUs - chan. uncertainty}
	\hat{\gamma}_{k}^{\dl}(\boldsymbol{\omega}_j,\mathbf{w}_j , \mathbf{p}_j) = \frac{|(\mathbf{\bar{h}}_{k,j}^\dl)^H \mathbf{w}_{k,j}|^2}{\hat{\psi}_{k,j}}
\end{align}
where $ \mathbf{\bar{h}}_{k,j}^{\dl} \triangleq \boldsymbol{\bar{\Lambda}}_j \mathbf{\hat{h}}_{k}^{\dl} $ and
\begin{equation}
\hat{\psi}_{k,j} \triangleq \sum\nolimits_{k'=1,k'\neq k}^{K} \bigl|\bigl(\mathbf{\bar{h}}_{k,j}^\dl\bigr)^H \mathbf{w}_{k',j}\bigr|^2 + \sum\nolimits_{k'=1}^{K} \epsilon_{k'}^{\dl} \bigl\|\boldsymbol{\Lambda}_j\mathbf{w}_{k',j}\bigr\|^2 + \sum\nolimits_{\ell=1}^{L} p_{\ell,j}^2 |\hat{g}_{\ell k}|^2 + \sum\nolimits_{\ell=1}^{L} p_{\ell,j}^2 \epsilon_{\ell,k} + \sigma_k^2. \nonumber
\end{equation}
From \eqref{eq: SINR ULUs - chan. uncertainty} and \eqref{eq: SINR DLUs - chan. uncertainty},  by incorporating  the channel uncertainties, the worst-case information rates (nats/s/Hz) of $\ULU$ and $\DLU$ in one transmission time block are given by 
\begin{align} 
\hat{R}_{\ell}^{\ul}  = \sum\nolimits_{j=1}^{2}\alpha_{\ell,j}^{\ul}\bar{\tau}_j \ln\bigl(1+\hat{\gamma}_{\ell}^{\ul}(\boldsymbol{\omega}_j,\mathbf{w}_j , \mathbf{p}_j)\bigr)  \; \text{and}\;
\hat{R}_{k}^{\dl}  = \sum\nolimits_{j=1}^{2} \alpha_{k,j}^{\dl}\bar{\tau}_j \ln\bigl(1+\hat{\gamma}_{k}^{\dl}(\boldsymbol{\omega}_j, \mathbf{w}_j , \mathbf{p}_j)\bigr) \label{eq: robust rates downlink} 
\end{align}
respectively.

%\subsection{Problem Formulation and Proposed Algorithm for Robust Design}

The robust counterpart of the SR maximization problem \eqref{eq: prob. max sum SE} is then formulated as
	\begin{subequations} \label{eq: prob. max sum SE robust}
	\begin{IEEEeqnarray}{cl}
		\underset{\boldsymbol{\Omega},\mathbf{w}, \mathbf{P},\tau,\boldsymbol{\alpha}_{\ell}^{\ul},\boldsymbol{\alpha}_{k}^{\dl}}{\text{maximize}} & \quad \hat{R}_{\Sigma}\triangleq\sum\nolimits_{\ell=1}^{L}\hat{R}_{\ell}^{\ul} + \sum\nolimits_{k=1}^{K}\hat{R}_{k}^{\dl} \label{eq: prob. max sum SE robust a} \\
		\text{subject to} 
		& \quad \eqref{eq: prob. general form b}-\eqref{eq: prob. general form j}, \\
		& \quad \hat{R}_{\ell}^{\ul} \geq \bar{R}_\ell^{\ul}, \forall \ell \in \mathcal{L}, \label{eq: prob. max sum SE robust c} \\
		& \quad \hat{R}_{k}^{\dl} \geq \bar{R}_k^{\dl}, \forall k \in \mathcal{K} \label{eq: prob. max sum SE robust d}
	\end{IEEEeqnarray}
\end{subequations}
and its subproblems are given by
\begin{subequations} \label{eq: prob. robust max sum SE 1'}
	\vspace{7pt}
	\begin{IEEEeqnarray}{cl}
		\underset{\mathbf{w}, \mathbf{P},\tau}{\text{maximize}} & \quad \sum\nolimits_{\ell=1}^{L}\hat{R}_{\ell}^{\ul}\bigl(\mathbf{w}, \mathbf{P},\tau|\boldsymbol{\Omega}\bigr) + \sum\nolimits_{k=1}^{K}\hat{R}_{k}^{\dl}\bigl(\mathbf{w}, \mathbf{P},\tau|\boldsymbol{\Omega}\bigr), \qquad \label{eq: prob. robust max sum SE 1':objective}\\
		\text{subject to} 
		& \quad \eqref{eq: prob. general form f}, \eqref{eq: DL power constraints}, \eqref{eq: UL power constraints}, \\
		& \quad \hat{R}_{\ell}^{\ul}\bigl(\mathbf{w}, \mathbf{P},\tau|\boldsymbol{\Omega}\bigr) \geq \bar{R}_\ell^{\ul}, \; \forall \ell \in \mathcal{L}, \label{eq: prob. robust max sum SE 1 c} \\
		& \quad \hat{R}_{k}^{\dl}\bigl(\mathbf{w}, \mathbf{P},\tau|\boldsymbol{\Omega}\bigr) \geq \bar{R}_k^{\dl},\; \forall k \in \mathcal{K} \label{eq: prob. robust max sum SE 1 d}
	\end{IEEEeqnarray}
\end{subequations}
where $\hat{R}_{\ell}^{\ul}\bigl(\mathbf{w}, \mathbf{P},\tau|\boldsymbol{\Omega}\bigr)  = \sum_{j=1}^{2}\bar{\tau}_j \ln\bigl(1+\hat{\gamma}_{\ell}^{\ul}(\mathbf{w}_j , \mathbf{p}_j |\boldsymbol{\omega}_j)\bigr)$ and $\hat{R}_{k}^{\dl}\bigl(\mathbf{w}, \mathbf{P},\tau|\boldsymbol{\Omega}\bigr)  = \sum_{j=1}^{2}\bar{\tau}_j \ln\bigl(1+\hat{\gamma}_{k}^{\dl}(\mathbf{w}_j , \mathbf{p}_j |\boldsymbol{\omega}_j)\bigr)$. We are now in position to expose the hidden  convexity of the non-convex problem \eqref{eq: prob. robust max sum SE 1'}. It can be observed that the functions $\hat{R}_{\ell}^{\ul}\bigl(\mathbf{w}, \mathbf{P},\tau|\boldsymbol{\Omega}\bigr)$ and $\hat{R}_{k}^{\dl}\bigl(\mathbf{w}, \mathbf{P},\tau|\boldsymbol{\Omega}\bigr)$ admit the same form as $R_{\ell}^{\ul}\bigl(\mathbf{w}, \mathbf{P},\tau|\boldsymbol{\Omega}\bigr)$ and $R_{k}^{\dl}\bigl(\mathbf{w}, \mathbf{P},\tau|\boldsymbol{\Omega}\bigr)$, respectively. Therefore, we will show in the following that the proposed
approach to the SR maximization problem \eqref{eq: prob. max sum SE} can be extended to its robust counterpart.

To handle the non-concavity of function $\hat{R}_{\ell}^{\ul}\bigl(\mathbf{w}, \mathbf{P},\tau|\boldsymbol{\Omega}\bigr)$, in the same manner as \eqref{eq: convexifying UL rate}, we derive the concave quadratic minorant of $\hat{R}_{\ell,j}^{\ul}(\mathbf{w}_j, \mathbf{p}_j,\mu_j)\triangleq\frac{\ln\bigl(1+\hat{\gamma}_{\ell}^{\ul}(\mathbf{w}_j, \mathbf{p}_j|\boldsymbol{\omega}_j)\bigr)}{\mu_j}$ as
\begin{IEEEeqnarray} {cl} \label{eq: convexifying robust UL rate}
\hat{\mathcal{R}}_{\ell, j}^{\ul,(\kappa)} := -\varpi(\hat{z}_{\ell, j}^{(\kappa)},\mu_j^{(\kappa)})\mu_j  + \zeta(\hat{z}_{\ell, j}^{(\kappa)},\mu_j^{(\kappa)}) - \frac{\hat{\varphi}_{\ell}^{(\kappa)}(\mathbf{w}_j, \mathbf{p}_j|\boldsymbol{\omega}_j)}{\mu_j^{(\kappa)}} +  \frac{2\hat{z}_{\ell, j}^{(\kappa)}}{\mu_j^{(\kappa)}p_{\ell,j}^{(\kappa)}}p_{\ell,j}\qquad
	\end{IEEEeqnarray}
	where 
	\begin{align}
	&	\hat{z}_{\ell, j}^{(\kappa)}\triangleq\hat{\gamma}_{\ell}^{\ul}(\mathbf{w}_j^{(\kappa)}, \mathbf{p}_j^{(\kappa)}|\boldsymbol{\omega}_j), 
	\hat{\varphi}_{\ell}^{(\kappa)}(\mathbf{w}_j, \mathbf{p}_j|\boldsymbol{\omega}_j)  \triangleq
		\ds\tr\Bigl(\bigl(p_{\ell,j}^2\mathbf{\bar{h}}_{\ell,j}^{\ul}(\mathbf{\bar{h}}_{\ell,j}^{\ul})^H + \boldsymbol{\hat{\Psi}}_{\ell,j}\bigl)\bigl( (\boldsymbol{\hat{\Psi}}_{\ell,j}^{(\kappa)})^{-1} - (\boldsymbol{\hat{\Psi}}_{(\ell-1),j}^{(\kappa)})^{-1}\bigr)\Bigr),   \nonumber \\
		& \boldsymbol{\hat{\Psi}}_{\ell,j}^{(\kappa)} \triangleq  \sum\nolimits_{\ell'=\ell+1}^{L} \bigl(p_{\ell',j}^{(\kappa)}\bigr)^2 \mathbf{\bar{h}}_{\ell',j}^{\text{u}}\bigl(\mathbf{\bar{h}}_{\ell',j}^{\text{u}}\bigr)^H + \sum\nolimits_{\ell'=1}^{L}\bigl(p_{\ell',j}^{(\kappa)}\bigr)^2 \epsilon_{\ell'}^{\text{u}}\boldsymbol{\bar{\Lambda}}_j + \rho^2\sum\nolimits_{k=1}^{K} \mathbf{\tilde{G}}_{j}^H \mathbf{w}_{k,j}^{(\kappa)} \bigl(\mathbf{w}_{k,j}^{(\kappa)}\bigr)^H\mathbf{\tilde{G}}_{j} + \sigma^2\mathbf{I}.\nonumber
	\end{align}
Analogously to \eqref{eq: convexifying DL rate}, the concave  minorant of $\hat{R}_{k,j}^{\dl}(\mathbf{w}_j,\mathbf{p}_j, \mu_j)\triangleq\frac{\ln\bigl(1+\hat{\gamma}_{k}^{\dl}(\mathbf{w}_j, \mathbf{p}_j|\boldsymbol{\omega}_j)\bigr)}{\mu_j}$ reads as
\begin{align} \label{eq: convexifying DL rate robust}
\mathcal{\hat{R}}_{k, j}^{\dl, (\kappa)} := & \nu(\vartheta_{k,j}^{(\kappa)},\mu_j^{(\kappa)}) + \xi(\vartheta_{k,j}^{(\kappa)},\mu_j^{(\kappa)})\vartheta_{k,j} + \lambda(\vartheta_{k,j}^{(\kappa)},\mu_j^{(\kappa)})\mu_j 
\end{align}
with the additional convex constraints
\begin{subequations} \label{eq: vartheta constraints robust}
	\begin{gather}
	\hat{\psi}_{k,j} \leq \vartheta_{k,j} \hat{\gamma}_{\text{s}}^{(\kappa)}\bigl(\mathbf{w}_{k,j}\bigr),\; \forall k\in \mathcal{K},\; j=1,2, \label{eq: SOC constraint for SINR DLUs robust} \\
	\vartheta_{k,j} > 0,\; \forall k\in \mathcal{K},\; j=1,2 \label{eq: positive vartheta constraint robust}
	\end{gather}
\end{subequations}
over the trust region
\begin{align} \label{eq: positive effective channel uncert.}
	\hat{\gamma}_{\text{s}}^{(\kappa)}\bigl(\mathbf{w}_{k,j}\bigr) & \triangleq 2\Re\bigl\{\bigl(\mathbf{\bar{h}}_{k,j}^\dl\bigr)^H\mathbf{w}_{k,j}^{(\kappa)}\bigr\}\Re\bigl\{\bigl(\mathbf{\bar{h}}_{k,j}^\dl\bigr)^H \mathbf{w}_{k,j}\bigr\}  - \bigl(\Re\bigl\{(\mathbf{\bar{h}}_{k,j}^\dl)^H \mathbf{w}_{k,j}^{(\kappa)}\bigr\}\bigr)^2 > 0, \forall k, j.
\end{align}

With the above discussions, the following convex program solved at iteration $(\kappa+1)$ is minorant maximization for the non-convex problem \eqref{eq: prob. robust max sum SE 1'}. 
	\begin{subequations} \label{eq: prob. max sum SE 1 robust}
	\begin{IEEEeqnarray}{cl}
		\underset{\mathbf{w}, \mathbf{P},\boldsymbol{\mu}, \boldsymbol{\vartheta}}{\text{maximize}} & \quad \mathcal{\hat{R}}_{\Sigma}^{(\kappa+1)} \triangleq \sum\nolimits_{\ell=1}^{L}\sum\nolimits_{j=1}^{2}\mathcal{\hat{R}}_{\ell, j}^{\ul, (\kappa)} + \sum\nolimits_{k=1}^{K}\sum\nolimits_{j=1}^{2}\mathcal{\hat{R}}_{k, j}^{\dl, (\kappa)} \label{eq: prob. max sum SE 1 robust a} \\
		\text{subject to} 
		& \quad \eqref{eq: prob. general form f}, \eqref{eq: DL power constraints b}, \eqref{eq: UL power constraints b}, \eqref{eq: alternative var. constraint 1}, \eqref{eq: alternative var. constraint 2}, \eqref{eq: power constraints approx. - altered variables}, \eqref{eq: vartheta constraints robust}, \eqref{eq: positive effective channel uncert.}, \label{eq: prob. max sum SE 1 robust b} \\
		& \quad \sum\nolimits_{j=1}^{2}\mathcal{\hat{R}}_{\ell, j}^{\ul, (\kappa)} \geq \bar{R}_\ell^{\ul}, \forall \ell \in \mathcal{L}, \label{eq: prob. max sum SE 1 robust c} \\
		& \quad \sum\nolimits_{j=1}^{2}\mathcal{\hat{R}}_{k, j}^{\dl, (\kappa)} \geq \bar{R}_k^{\dl}, \forall k \in \mathcal{K}. \label{eq: prob. max sum SE 1 robust d}
	\end{IEEEeqnarray}
\end{subequations}
To solve \eqref{eq: prob. max sum SE robust}, we customize  Algorithm \ref{alg: for max sum SE} as follows. In Step 2, we solve the subproblems \eqref{eq: prob. robust max sum SE 1'} instead of \eqref{eq: prob.  max sum SE 1'}. In Step 5, we solve the convex program \eqref{eq: prob. max sum SE 1 robust} instead of \eqref{eq: prob. max sum SE 2}. The checking condition in Step 9 becomes $\hat{\mathcal{R}}_{\Sigma}^{(\kappa)}>\hat{\mathcal{R}}_{\text{SR}}^{\max} $ and the objective function value in Step 11 is now $\hat{R}_{\Sigma}$ in \eqref{eq: prob. max sum SE robust a}. We  refer to this customized algorithm as Algorithm 3.

{\hili The proposed algorithms in this paper provide centralized solutions, where a central processing station located at the BS carries out the computations. We assume that the CSI of users and the status of QoS requirements are locally known at the BS, which is done by using implicit beamforming and relying on channel reciprocity, e.g., as in the IEEE 802.11ac. At the beginning of transmission block, the CSI between the BS and all users can be estimated through the handshaking processes, in which the pilot signals are sent from all users to the BS. Then, the CSI estimates can be updated as UL users periodically send pilots to the BS and DL users send the acknowledgments to the BS while they successfully receive the data packets. For CCI estimation, the DL users might receive the training signals from UL users, and then send the CCI estimates to the BS at the beginning of the next transmission block. For a multicell scenario, decentralized solutions are more practically appealing and will be reported in future work. Specifically, each BS is required to carry out the computations locally and the information sharing among BSs to arrive to the optimal solution can be done via the LTE-X2 interface.}

\section{Numerical Results}\label{NumericalResults}
To evaluate the performance of the proposed algorithms through computer simulations, a small cell topology as illustrated in Fig. \ref{fig: User layout} is set up. There are 8 UL and 8 DL users randomly placed within a 100-meter radius from the BS. The channel vector between  HA$_i$ and a user at a distance $d_{\mathtt{BS,U}}$ is generated as $\mathbf{h}=\sqrt{10^{-\mathtt{PL}_{\mathtt{BS,U}}/10}}\ddot{\mathbf{h}}$, where $\mathbf{h}\in\{\mathbf{h}_{\ell,i}^{\ul},\mathbf{h}_{k,i}^{\dl}\}$ and $ \mathtt{U} \in \{\ULU, \DLU\} $. {\hili The flat fading channel between  HA$_i$ and a user at a distance $d_{\mathtt{BS,U}}$ is assumed to be Gaussian distributed, i.e., $\mathbf{h}\sim\mathcal{CN}(0,\mathtt{PL}_{\mathtt{BS,U}}\mathbf{I})$ with $\mathbf{h}\in\{\mathbf{h}_{\ell,i}^{\ul},\mathbf{h}_{k,i}^{\dl}\}$ and $ \mathtt{U} \in \{\ULU, \DLU\} $, where $\mathtt{PL}_{\mathtt{BS,U}}$ represents the path loss  \cite{Dan:TWC:14, Dinh:Access, 3GPP}. Similarly, the channel response from $\ULU$ to $\DLU$  at a distance $d_{\ell k}$ is generated as $ g_{\ell k} \sim\mathcal{CN}(0,\mathtt{PL}_{\ell k})$.} The entries of the SI channel $ \mathbf{G}_{\SI} $ are modeled as independent and identically distributed Rician random variables with Rician
factor of 5 dB. For ease of cross-referencing,  the other parameters given in Table~\ref{parameter} are also used in \cite{Bharadia13,Duarte:TWC:12,3GPP}. Without loss of generality, the rate thresholds  for all users are set as  $ \bar{R}^{\ul}_\ell = \bar{R}^{\dl}_k \equiv \bar{R}$.  The achieved rates divided by $\ln2$ are shown in all figures, to be presented in the unit of bps/Hz instead of nats/s/Hz. Except for the cumulative distribution functions (CDFs), we use the scenario in Fig. \ref{fig: User layout} and generate 500 random channels to calculate the average rates in the rest of simulations.

To conduct a performance assessment, we examine the proposed algorithms (Algorithms 1, 2, and 3) and other three schemes:
\begin{enumerate}
	\item An FD scheme with user grouping and fractional time (FD-UGFT) \cite{Dinh:Access};
	\item An FD scheme based on a conventional design in \cite{Dan:TWC:14} with QoS requirements (FD-CDQS);
	\item An HD system, in which UL and DL transmissions are independently executed in two transmission blocks with full-array antennas, and each block time is normalized to 1. In this case, the optimal rate of HD scheme is calculated as the average of UL and DL rates.
\end{enumerate}

\begin{table}[t]
	\begin{varwidth}[h]{0.5\linewidth}
%		\begin{table}
			\centering
			\captionof{table}{Simulation Parameters}
			\label{parameter}
			\vspace{-15pt}
			\scalebox{0.9}{
			\begin{tabular}{l|l}
				\hline
				Parameters & Value \\
				\hline\hline
				Radius of small cell                   & 100 m \\
				Bandwidth & 10 MHz \\
				Noise power spectral density at  receivers& -174 dBm/Hz \\
				PL  between BS and a user,  $\mathtt{PL}_{\mathtt{BS,U}}$ & 103.8 + 20.9$\log_{10}(d_{\mathtt{BS,U}})$ dB\\
				PL from $\ULU$ to $\DLU$,	$\mathtt{PL}_{\ell k} $ & 145.4 + 37.5$\log_{10}(d_{\ell k})$ dB\\
				Power budgets, $ P_{t}^{\text{max}}$ and $P_{\ell}^{\text{max}},\forall\ell $ & 26 dBm and 23 dBm \\
				Number of antennas at  BS, $ 2N $ & 16\\
				Rate threshold, $ \bar{R} $	&  1 bps/Hz\\
					\hline		   				
			\end{tabular} }
%		\end{table}
	\end{varwidth}
	\hfill
	\begin{minipage}[ht]{0.42\linewidth}
		\centering
		\includegraphics[width=0.9\linewidth, trim={0cm, 0cm, 0cm, 0cm}]{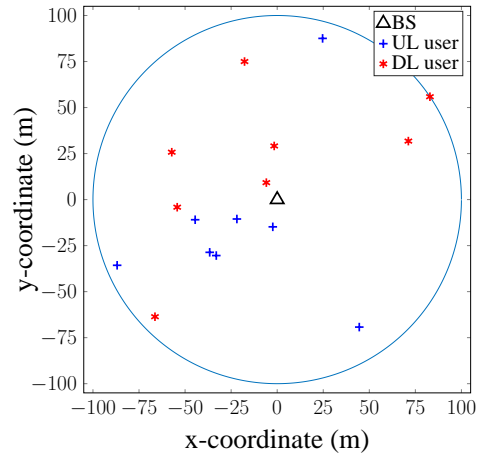}
		\vspace{-20pt}
		\captionof{figure}{A small cell setup with 8 DL and 8 UL users used in our numerical examples.}
		\label{fig: User layout}
	\end{minipage}
%\vspace{-10pt}
\end{table}

Fig. \ref{fig: convergence} shows the typical convergence behavior of the algorithms for FD schemes, at $ \rho^2=-30 $ dB for convenience. We do not consider the convergence behavior of HD due to the independence of UL and DL SR problems. As can be seen, Algorithm  \ref{alg: for max sum SE} always needs only a few tens of iterations to reach the converging value. The per-iteration optimized value of the proposed algorithm almost saturates from the 30-th iteration, while those of FD-UGFT and FD-CDQS algorithms still increase slightly. In addition, the proposed scheme provides the best performance among the FD-based systems. A method of brute-force search for the integer variables (BFS-IV) based on the proposed algorithm is presented in Fig. \ref{fig: convergence} as a baseline. In the BFS-IV method, the values of $(\boldsymbol{\Omega}, \boldsymbol{\alpha}_{\ell}^{\ul}, \boldsymbol{\alpha}_{k}^{\dl})$ are fixed to each of all possible combinations, and then \eqref{eq: prob. max sum SE 1'} is used for each case as the proposed FD. In most cases of BFS-IV, $ \boldsymbol{\alpha}_{\ell}^{\ul} $ and $ \boldsymbol{\alpha}_{k}^{\dl} $ are fixed to zero, and thus the associated UL transmit power and DL beamforming variables can be set to zero no matter how the other variables are optimized. Therefore, the convergence rate of BFS-IV is faster than those of the three other methods. However, a large number of possible cases for BFS-IV are inapplicable to the FD system. It is further observed that the proposed FD scheme provides the SR close to that of the BFS-IV as Proposition \ref{prop: Constraint Reduction} and Lemma \ref{lem: Bound Tightening} hold the optimality.

\begin{figure}[t] 
	%	\begin{minipage}[t]{\columnwidth}
	\centering
	\vspace{-30pt}
	\includegraphics[width=0.45\columnwidth, trim={0cm, 0cm, 0cm, -2cm}]{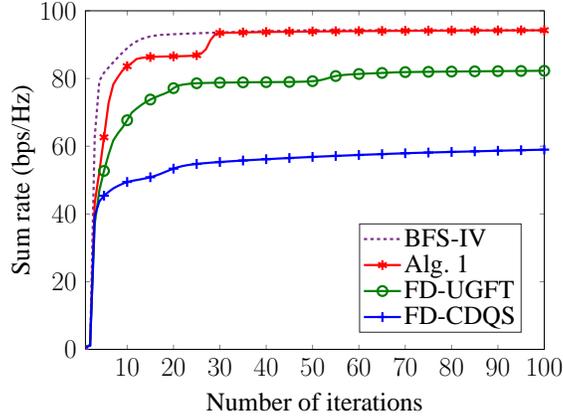}
	\vspace{-20pt}
	\caption{Typical convergence rate of FD schemes with SR maximization problem.}
	\label{fig: convergence}
%	\vspace{-20pt}
	%	\end{minipage}
	%	\hfill
\end{figure}

\begin{table}[b]
	\setlength\extrarowheight{5pt}
	
	\centering
	\color{blue}{
		\vspace{10pt}
		\captionof{table}{\color{blue} Complexity Comparison}
		\vspace{-10pt}
		\label{Complexity}
		\scalebox{1}{
			\begin{tabular}{l|c|c|c|c}
				\hline
				Methods & \#Sub-problems & \#Variables & \#Constraints & Per-iteration complexity \\
				\hline\hline
				BFS-IV & $ s=4^2.2^{2K}.2^{2L} $ & $v_1=4NK+2K+2L+2 $ &  $ c_1=7K+6L+8 $ & $ s.\mathcal{O}\bigl(c_1^{2.5}v_1^{2}+c_1^{3.5}\bigr) $\\
				\hline
				FD-CDQS & 1 & $ v_2=N^2K+2K+3L $ &  $ c_2=6K+7L+1 $ & $ \mathcal{O}\bigl(c_2^{2.5}v_2^{2}+c_2^{3.5}\bigr) $ \\ %c_1-2K-7
				\hline	
				FD-UGFT & 1 & $ v_3=2NK+12K+14L+4 $ & $ c_3=15K+16L+6 $ & $ \mathcal{O}\bigl(c_3^{2.5}v_3^{2}+c_3^{3.5}\bigr) $ \\ %c_1+8K+10L-2
				\hline
				Algorithm 1 & 8 & $ v_1 $ &  $ c_1 $ & $ 8.\mathcal{O}\bigl(c_1^{2.5}v_1^{2}+c_1^{3.5}\bigr) $ \\
				\hline	   				
	\end{tabular} } }
\end{table}

{\hili Besides the convergence speed, an important measure for computation is the per-iteration complexity as it has a significant impact on the total complexity. In each iteration one or more sub-problems are solved, depending on the algorithm structure. As illustrated in Table \ref{Complexity}, we provide a per-iteration complexity comparison among the algorithms used for FD schemes, in which the complexity of each sub-problem is represented in term of big-O notation as given in \cite{SeDuMi:2002}. It is true that the number of variables and constraints are the same in BFS-IV and the proposed method, since the power control for BFS-IV is based on Algorithm \ref{alg: for max sum SE}. However, the number of sub-problems for BFS-IV is very high which makes it inapplicable to practical implementation. In the next simulations, we no longer consider BFS-IV due to its extremely high complexity. The number of constraints required by FD-CDQS and the proposed algorithm are comparable. However, since FD-CDQS must handle the rank-1 relaxed variables for the beamforming vectors, the proposed scheme needs much less number of variables, leading to lower per-iteration complexity. As compared to the FD-UGFT method, although Algorithm \ref{alg: for max sum SE} needs twice the number of the beamforming vector variables, it has a huge payoff related to the number of users. In particular, the largest exponent in big-O notation is with respect to the number of constraints which is significantly reduced in Algorithm \ref{alg: for max sum SE}. Therefore, the proposed method provides a lower complexity, especially when the number of users is large. The per-iteration complexity of the algorithms also has an impact on the required number of iterations to capture a feasible starting point in \eqref{eq: prob. max sum SE - initialization} as investigated in Fig. \ref{fig: Finding a feasible starting point}. In Fig. \ref{fig: Feasible point}, we plot the typical progresses of algorithms to achieve a feasible point from the same initial point which involves the same values of beamforming vector and UL transmit power. Despite starting at the same value, the per-user rate with the proposed algorithm increases to the feasible region (positive per-user rate) more quickly than that with other algorithms. Fig. \ref{fig: Average No. Iterations} further depicts the average number of iterations until the feasible point is reached through examining 500 random channels. Algorithm \ref{alg: for max sum SE} requires the least iterations for the feasible point, as it guarantees less constraints and uses less decision variables. The above results confirm that the proposed method provides a lower complexity in terms of both convergence speed and per-iteration complexity.} 

%\begin{table}[b]
%%	\setcounter{table}{1}
%	\setlength\extrarowheight{5pt}
%	
%	\centering
%	\color{blue}{
%		\captionof{table}{\color{blue} Complexity Comparison}
%		\vspace{-10pt}
%		\label{Complexity}
%		\scalebox{1}{
%			\begin{tabular}{l|c|c|c|c}
%				\hline
%				Methods & \#Sub-problems & \#Variables & \#Constraints & Per-iteration complexity \\
%				\hline\hline
%				BFS-IV & $ s=4^2.2^{2K}.2^{2L} $ & $v_1=4NK+2K+2L+2 $ &  $ c_1=7K+6L+8 $ & $ s.\mathcal{O}\bigl(c_1^{2.5}v_1^{2}+c_1^{3.5}\bigr) $\\
%				\hline
%				FD-CDQS & 1 & $ v_2=N^2K+2K+3L $ &  $ c_2=6K+7L+1 $ & $ \mathcal{O}\bigl(c_2^{2.5}v_2^{2}+c_2^{3.5}\bigr) $ \\ %c_1-2K-7
%				\hline	
%				FD-UGFT & 1 & $ v_3=2NK+12K+14L+4 $ & $ c_3=15K+16L+6 $ & $ \mathcal{O}\bigl(c_3^{2.5}v_3^{2}+c_3^{3.5}\bigr) $ \\ %c_1+8K+10L-2
%				\hline
%				Algorithm 1 & 8 & $ v_1 $ &  $ c_1 $ & $ 8.\mathcal{O}\bigl(c_1^{2.5}v_1^{2}+c_1^{3.5}\bigr) $ \\
%				\hline	   				
%	\end{tabular} } }
%\end{table}

\begin{figure}[t] 
	\centering
%	\setcounter{figure}{4}
%	\captionsetup{labelfont={myblue}}
	
	%		\vspace{-20pt}
	\begin{subfigure}[]
		{
%							\vspace{-20pt}
			\includegraphics[width=0.43\linewidth, trim={0cm, 0.8cm, 0cm, 0cm}]{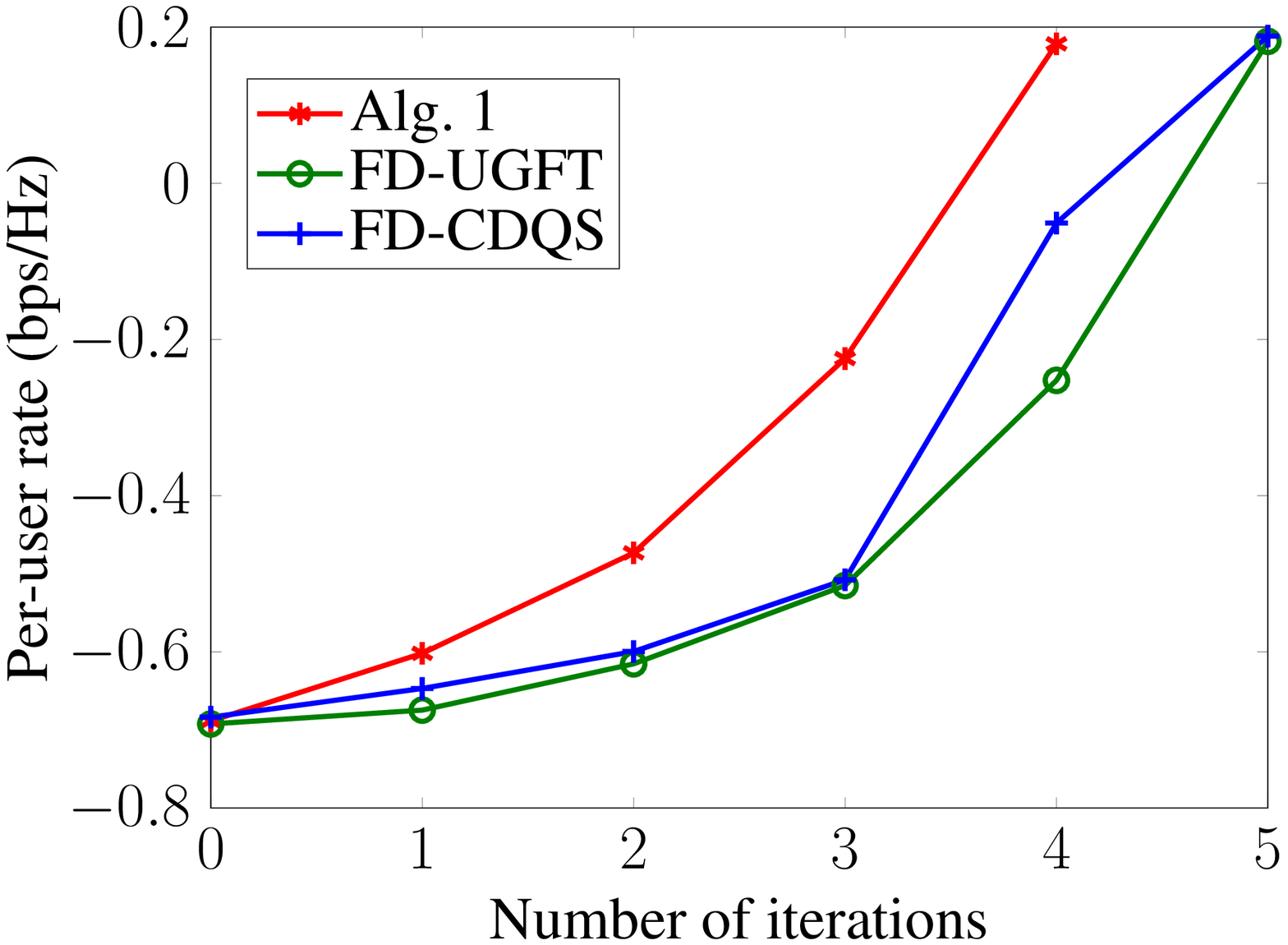}
					\vspace{-20pt}
			%		\caption{}
			\label{fig: Feasible point}
			%		\vspace{-10pt}
		}
	\end{subfigure}
	\hfill
	\begin{subfigure}[]
		{
%							\vspace{-20pt}
			\includegraphics[width=0.43\linewidth, trim={0cm, 0.8cm, 0cm, 0cm}]{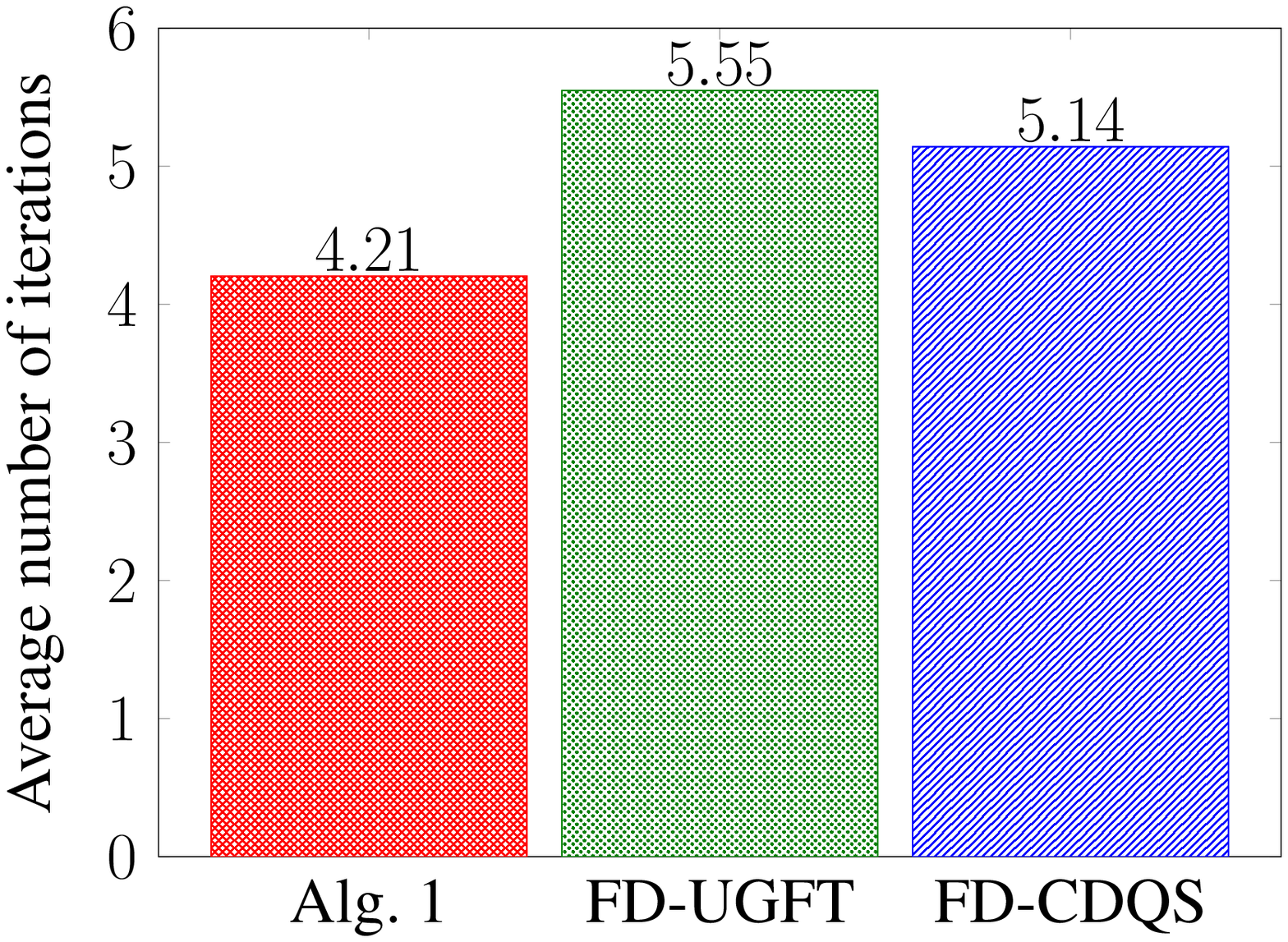}
			\vspace{-20pt}
			
			\label{fig: Average No. Iterations}
			%		\vspace{-10pt}
		}
	\end{subfigure}
			\vspace{-10pt}
	\caption{\hili The number of iterations for finding a feasible starting point. (a) A typical iterative progress for finding a feasible point. (b) The average number of iterations to capture a feasible starting point.}
	\label{fig: Finding a feasible starting point}
\end{figure}

\begin{figure}[t] 
%	\begin{minipage}[t]{0.48\columnwidth}
		\centering
		\vspace{-45pt}
		\includegraphics[width=0.45\columnwidth, trim={0cm, 0cm, 0cm, -2cm}]{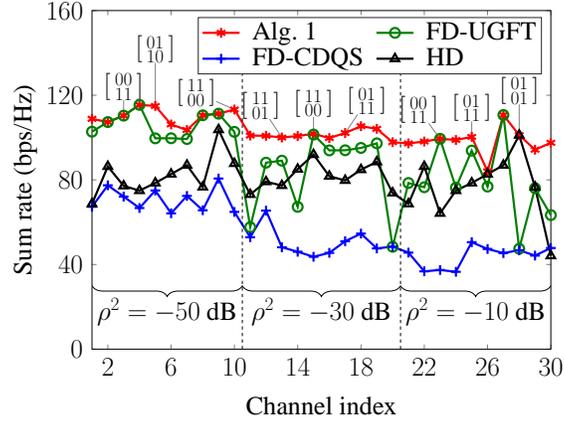}
		\vspace{-20pt}
		\caption{Sum rate versus random channels.}
		\label{fig: SR vs random channels}
%		\vspace{-20pt}
%	\end{minipage}
\end{figure}

%By using the joint HA antenna selection and two-phase transmission, the proposed scheme efficiently exploits the channel information to improve the achievable rates. In fact, the proposed FD {\hili scheme} either outperforms other methods or collects the best performance among them in every channel. To clarify this point, we further provide the performance of all schemes by observing the HA mode selection through the value of the matrix $ \boldsymbol{\Omega} $, under the effect of the residual SI. 

{\hili To deeply understand the algorithm, we capture the HA mode selection in two phases under different random channels and residual SI levels. Through the values of $ \boldsymbol{\Omega} $, we can perceive the mode changes of HAs so that the system is well adapted to various conditions.} In Fig. \ref{fig: SR vs random channels}, we randomly generate 30 channels: 10 for every three values of $ \rho^2\in\{-50, -30,  -10 \}$ dB. The SRs of the FD-CDQS are worse than those for the other schemes, while the proposed FD provides at least the same performance as the HD or FD-UGFT. If we observe the HA mode selection at the channel index 5, for instance, HA$_1$ and HA$_2$ in phase 1 (1-st column of $ \boldsymbol{\Omega} $) are set to the Rx (0) and Tx (1) modes, respectively; and vice versa in phase 2 (2-nd column of $ \boldsymbol{\Omega} $). Such settings for transmission provide the best performance, as compared to the other FD schemes without Tx/Rx-mode antenna selection. Notably, when $ \rho^2 = -50 $ dB, the HA mode selection tends to have FD transmission in both phases. Therefore, the performance of the proposed FD is at least equal to that of FD-UGFT. On the other hand, the trend of the mode selection shifts from the two-phase FD transmission to the hybrid-mode, where both FD and HD are successively used in two phases, and finally to the HD-mode transmission when $ \rho^2 $ becomes large, i.e., $ \rho^2 $ = -10 dB. It demonstrates that the proposed FD is more stable in the change of random channels.

\begin{figure}[t] 
	\centering
	\vspace{-20pt}
	\begin{subfigure}[Sum rate versus residual SI.]%{0.48\columnwidth}
		%		\centering
%				\vspace{-20pt}
		{
			\includegraphics[width=0.45\linewidth, trim={0cm, .8cm, 0cm, 0cm}]{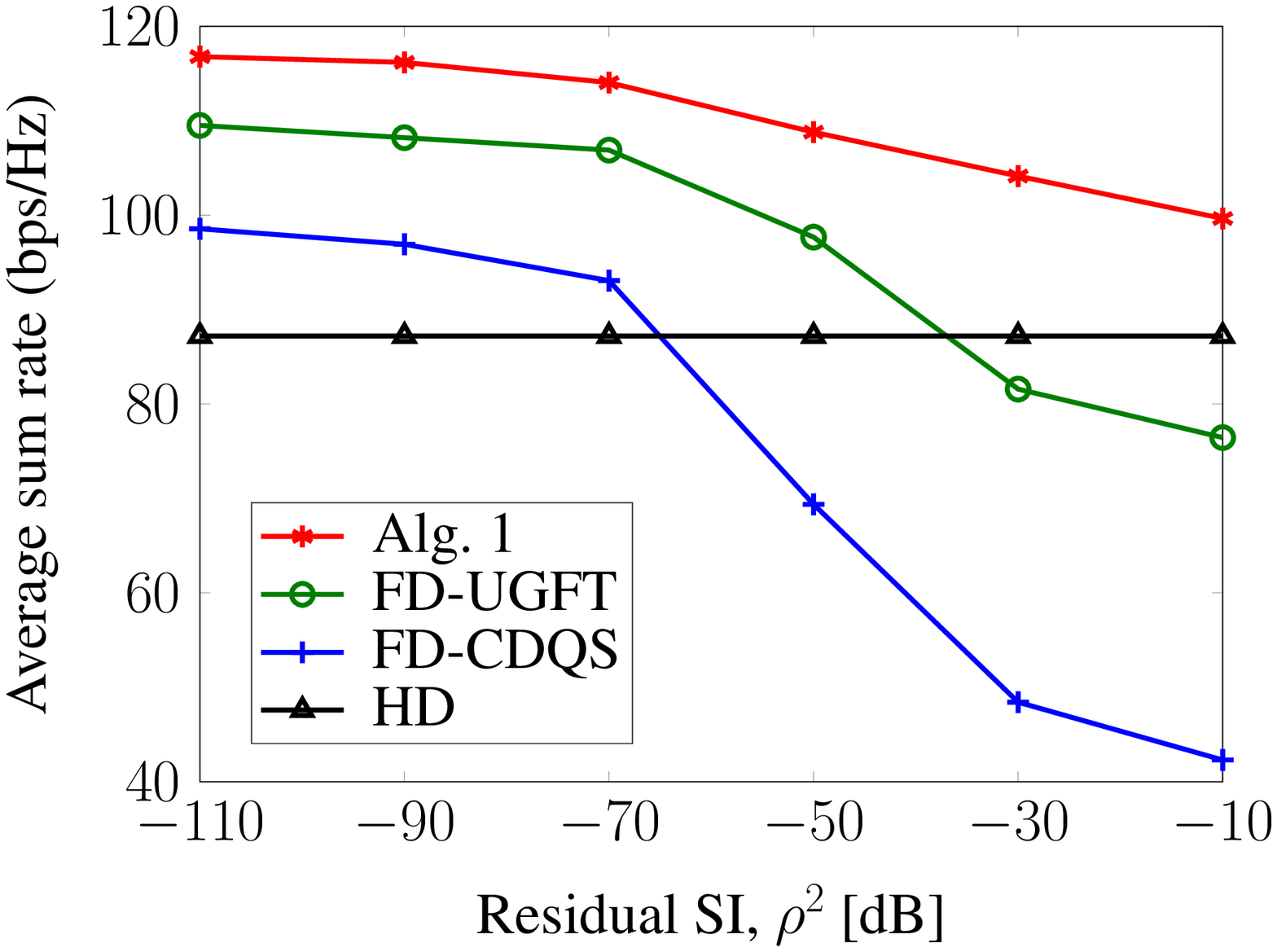}
			%		\vspace{-40pt}
			%		\caption{Average sum rate versus the residual SI, $ \rho^2 $.}
			\label{fig: SR vs SI}
			%		\vspace{-10pt}
		}
	\end{subfigure}
	\hfill
	\begin{subfigure}[Max-min rate versus residual SI.]%{0.48\columnwidth}
		%		\centering
%				\vspace{-20pt}
		{\includegraphics[width=0.45\linewidth, trim={0cm, .8cm, 0cm, 0cm}]{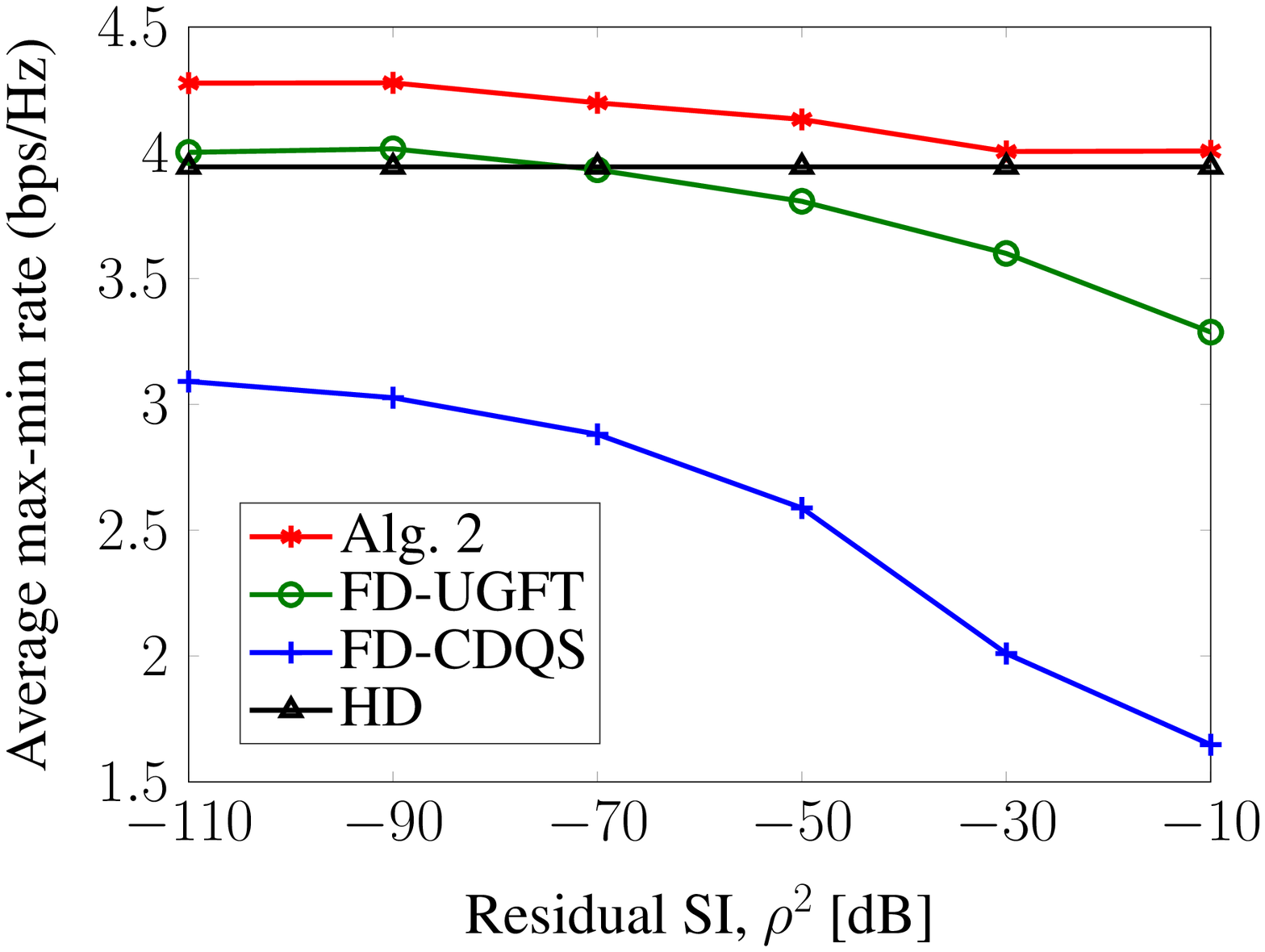}
			%		\vspace{-40pt}
			%		\caption{Average max-min rate versus the residual SI, $ \rho^2 $.}
			\label{fig: Max-min vs SI}
			%		\vspace{-10pt} 
		}
	\end{subfigure}
	\vspace{-10pt} 
	\caption{The effect of residual SI on the achievable rate.}
	\vspace{10pt}
\end{figure}

Fig. \ref{fig: SR vs SI} plots the SR as a function of the residual SI, $\rho^2$. Obviously, as the residual SI increases, the performance of three FD schemes deteriorates, while that of the HD system keeps unchanged. As expected, the proposed FD always achieves the highest SR compared to the others. The proposed FD enables the hybrid modes with a dynamically-timed two-phase transmission, and thus, it outperforms the HD scheme even when $\rho^2$ becomes more stringent. Otherwise, the advantages of the proposed FD  are efficiently exploited by the HA-enabled Tx/Rx-antenna mode. In addition,    other FD schemes plummet to under the HD in terms of the SR when the residual SI becomes larger than -40 dB.  Fig. \ref{fig: Max-min vs SI} depicts the max-min rate for both UL and DL users as $\rho^2$ increases.  As compared to HD, the performance of FD-UGFT is slightly better when $\rho^2$ is very small, e.g., $\rho^2=\{-110, -90 \}$ dB, while FD-CDQS is always inferior to HD in terms of max-min rate because the HD scheme suffers from no SI and CCI, and it has more DoF due to the full-array antennas for transmission. The max-min rate of the proposed FD using a hybrid approach is at least equal to that of HD in case of unfavorable channels for using FD radio.

\begin{figure}[t] 
	\centering
%	\vspace{-20pt}
	\begin{subfigure}[CDF of the sum rate.]
		{
%			\vspace{-25pt}	
			\includegraphics[width=0.45\linewidth, trim={0cm, .9cm, 0cm, 1cm}]{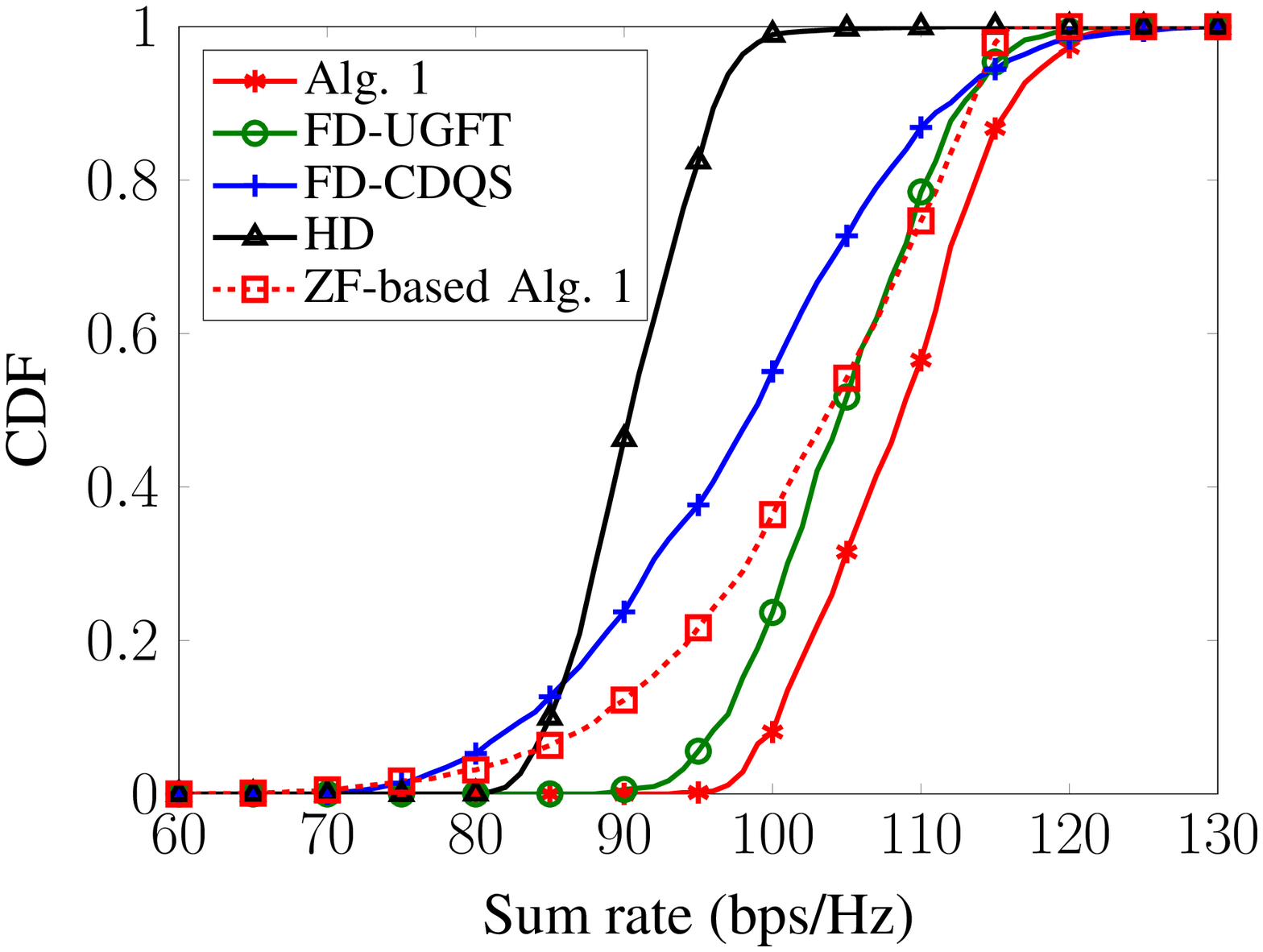}
			%		\vspace{-40pt}
			%		\caption{}
			\label{fig: CDF vs Sum rate}
			%		\vspace{-10pt}
		}
	\end{subfigure}
	\hfill
	\begin{subfigure}[CDF of the max-min rate.]
		{
%					\vspace{-25pt}
			\includegraphics[width=0.45\linewidth, trim={0cm, 0.9cm, 0cm, 1cm}]{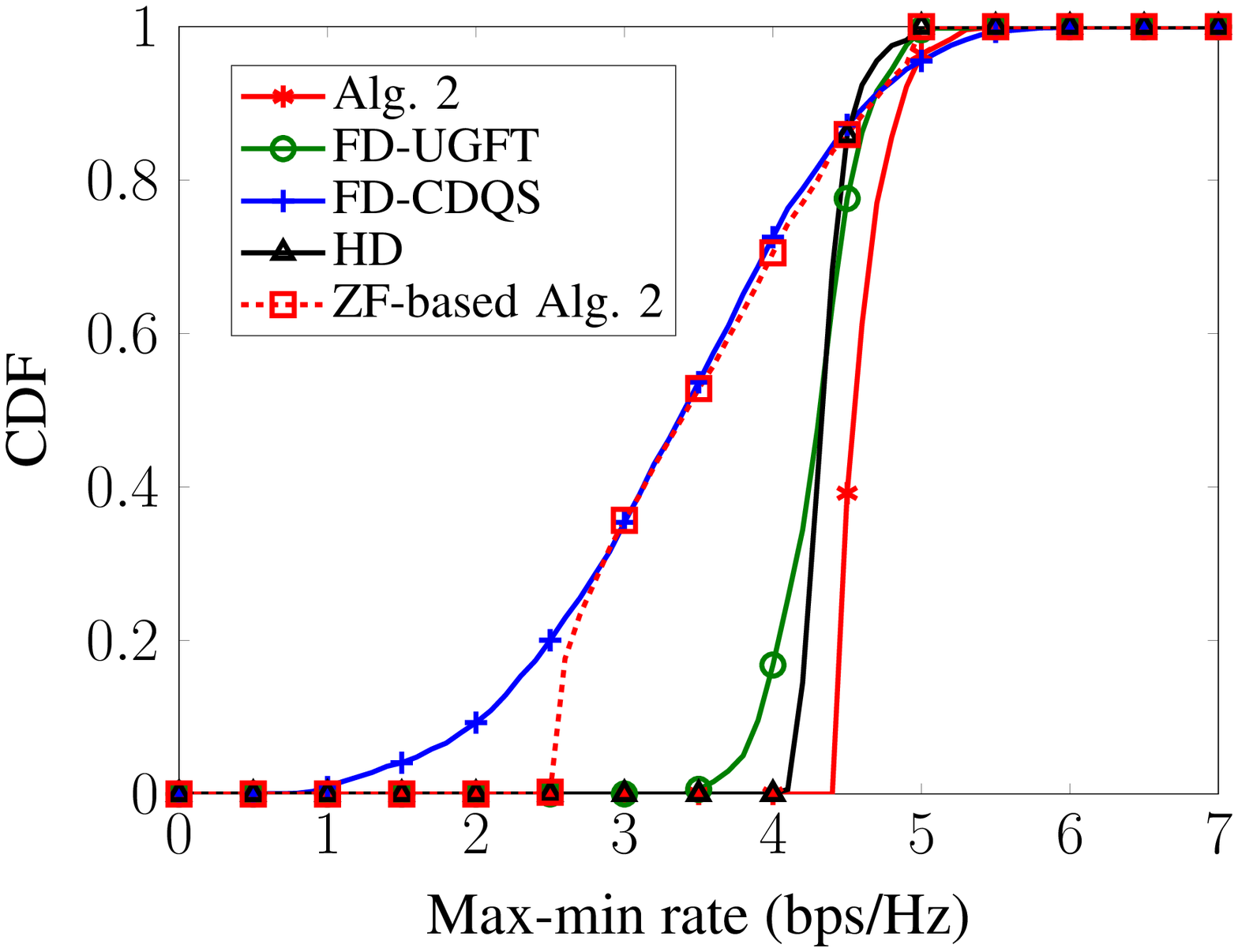}
			%		\vspace{-40pt}
			
			\label{fig: CDF vs Max-min rate}
			%		\vspace{-10pt}
		}
	\end{subfigure}
	\vspace{-10pt}
	\caption{Cumulative distribution functions (CDFs). (a) CDF of the sum rate with $\bar{R}=1 $ bps/Hz. (b) CDF of max-min rate.}
	\vspace{-0pt}
	\label{fig: CDF}
\end{figure}

Fig. \ref{fig: CDF} shows the CDFs of the SR and max-min rate. We examine 1000 topologies of users' random positions and 500 random channels for each topology. In both cases, the proposed FD scheme outperforms the others for moderate rate requirements. Interestingly, Fig. \ref{fig: CDF vs Sum rate} indicates that the SR optimization is highly flexible for the channel information, including users' positions and random channels. In fact, the differences between the 5- and 95-percentile of all methods are larger than 20 bps/Hz, which is more than 1 bps/Hz per user on average. Meanwhile, the interval within the same range of percentiles in Fig. \ref{fig: CDF vs Max-min rate} is lower than 0.5 bps/Hz, except for the FD-CDQS method. This is a result of the SR optimization having the tendency to improve the sum of the channel efficiency as long as the QoS constraints are satisfied, while the max-min  optimization improves all users' rates in association to each other. {\hilise To examine insight the property of proposed scheme, we provide the CDFs of a baseline method with Alg. \ref{alg: for max sum SE} using zero-forcing for both UL and DL transmission (ZF-based Alg. \ref{alg: for max sum SE}). As depicted in Fig. \ref{fig: CDF vs Sum rate}, the ZF-based Alg. \ref{alg: for max sum SE} method gives the worst performance when compared to the other two-phase FD schemes (FD-UGFT and Alg. \ref{alg: for max sum SE}), which verifies the advantages of MMSE-SIC and beamforming design. Interestingly, the CDF of ZF-based Alg. \ref{alg: for max min SE} in Fig. \ref{fig: CDF vs Max-min rate} reflects both properties of ZF and two-phase transmission. The ZF-based Alg. \ref{alg: for max min SE} may utilize the different transmission modes to reduce the SI and CCI, while FD-CDQS still suffers from the interference in the whole time block. Accordingly, though the FD-CDQS method uses MMSE-SIC and beamforming design, ZF-based Alg. \ref{alg: for max min SE} with merely canceling the multiuser interference takes more advantage through managing the SI and CCI with two phases. However, Alg. \ref{alg: for max min SE} with MMSE-SIC and beamforming design further facilitates the data transmission as compared to ZF-based Alg. \ref{alg: for max min SE}.}

\begin{figure}[t] 
	\centering
	\vspace{-20pt}
	\begin{subfigure}[]
		{
			\vspace{-20pt}
			\includegraphics[width=0.42\linewidth, trim={0cm, 1cm, 0cm, 0cm}]{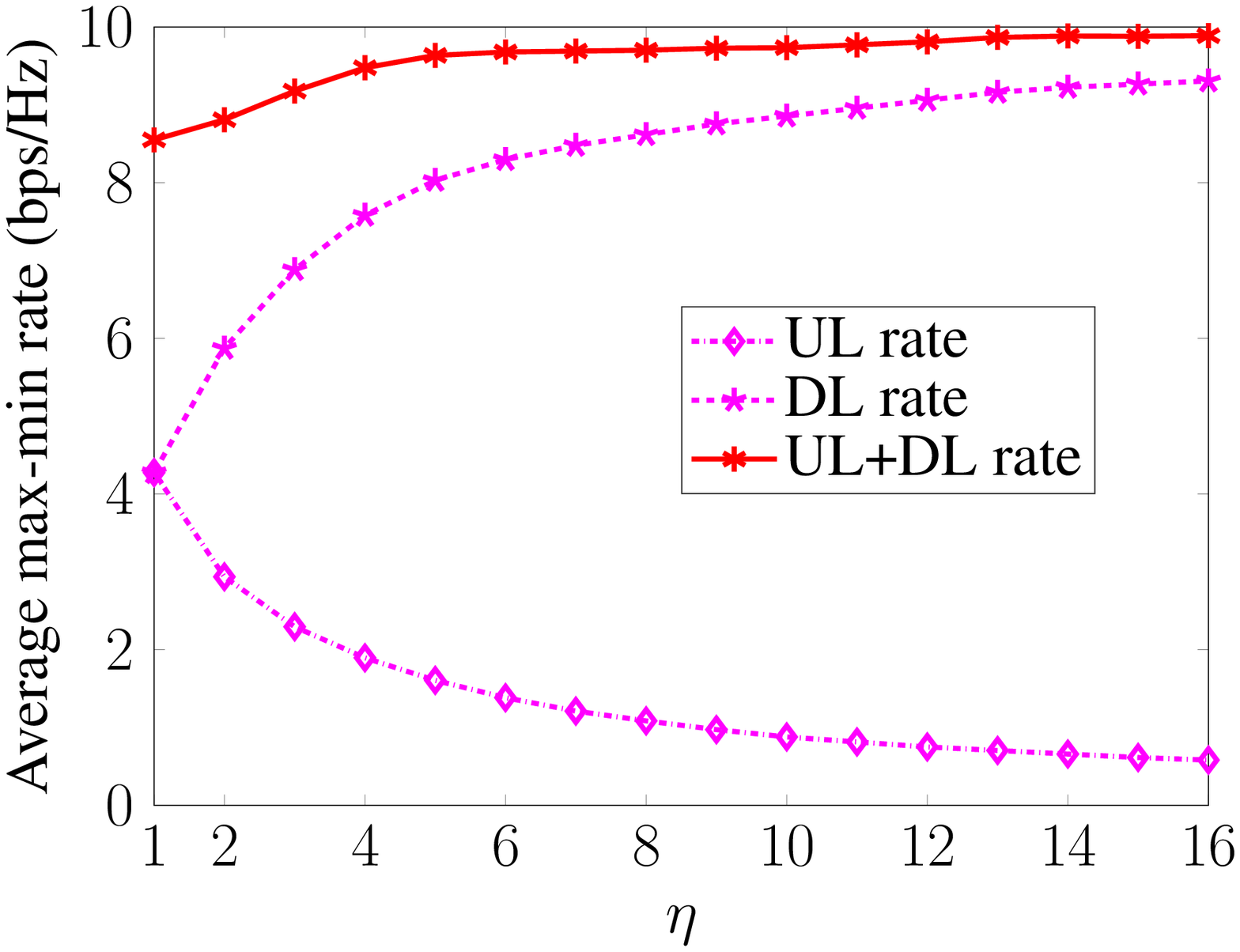}
			%		\vspace{-40pt}
			%		\caption{}
			\label{fig: Rate vs eta}
			%		\vspace{-10pt}
		}
	\end{subfigure}
	\hfill
	\begin{subfigure}[]
		{
			\vspace{-20pt}
			\includegraphics[width=0.42\linewidth, trim={0cm, 1cm, 0cm, 0cm}]{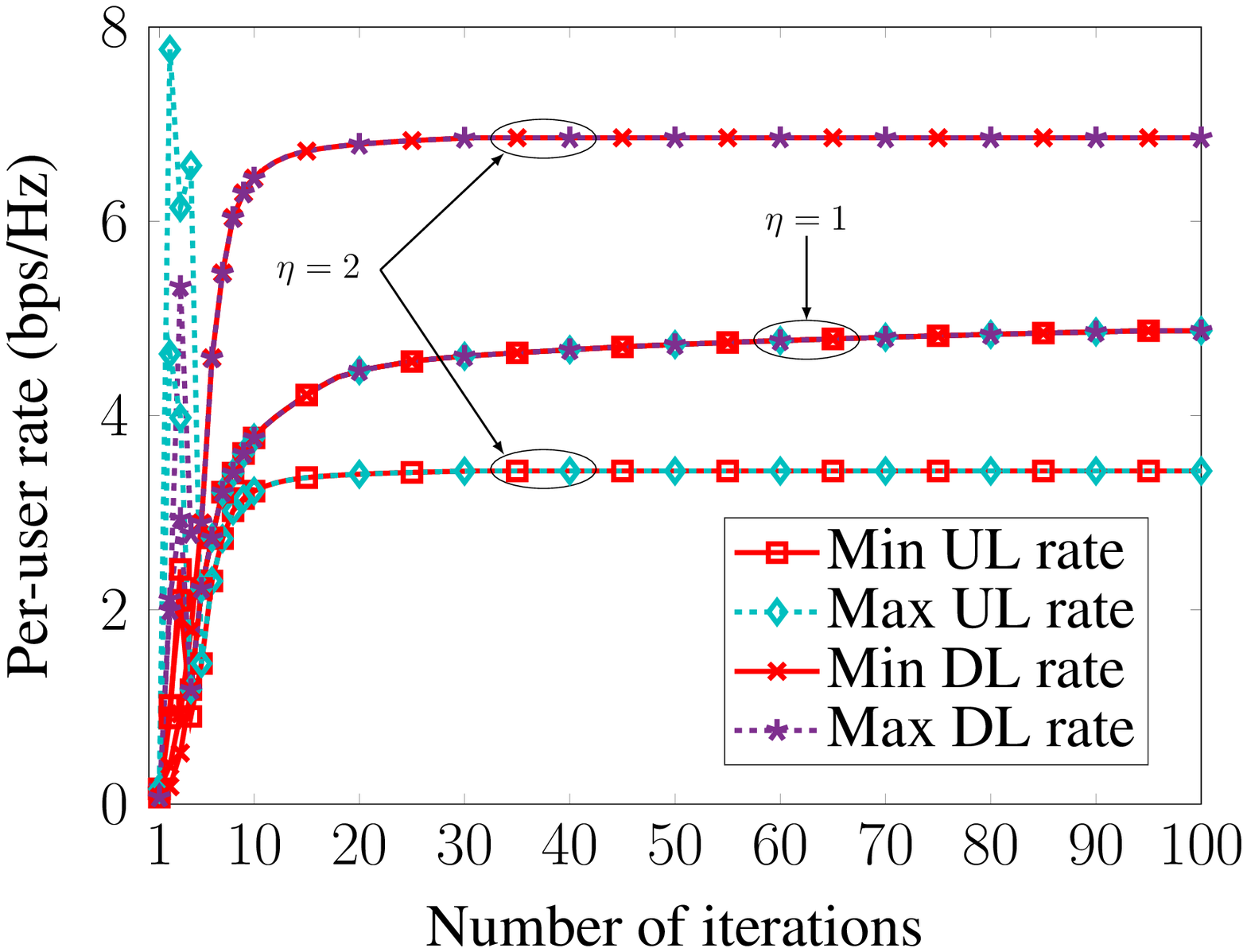}
%					\vspace{-40pt}
			
			\label{fig: Max-min convergence}
			%		\vspace{-10pt}
		}
	\end{subfigure}
	\vspace{-10pt}
	\caption{Per-user rate for max-min optimization. (a) Average max-min rate as $ \eta=1 \rightarrow 16 $. (b) Typical convergence behavior.}
	\label{fig: Max-min eta}
\end{figure}

By using Algorithm \ref{alg: for max min SE}, Fig. \ref{fig: Max-min eta}(a) shows the per-user rate versus $\eta$. It is clear that the UL users keep the max-min rate, which is the best value of $ \phi $ among eight subproblems \eqref{eq: prob. max-min SE:eqi}, and the DL users achieve the scaled rate of $ \eta\phi $. Although the max-min rate decreases as $ \eta $ increases, the total of UL and DL per-user rates increases. This is because when $ \eta $ becomes larger, the BS with array antennas effectively favors DL users over UL users. Subsequently, the amount of the increase in the DL rate is more than that of the decrease in the UL rate, which carries the total rate up. To derive further insight into the rates among UL (or DL) users, we investigate the convergences of their maximum and minimum per-user rates, as shown in Fig. \ref{fig: Max-min convergence}. For a given value of $ \eta $, the maximum and minimum of UL per-user rates converge at the same value. It indicates that the UL per-user rates are equal to each other in the max-min optimization, and similar observation can be made for DL users, as constraints \eqref{eq: prob. max-min SE c:eqi} and \eqref{eq: prob. max-min SE d:eqi} hold. In fact, it first takes several iterations to find a feasible point, and from the 5-th iteration, the maximum and minimum UL (or DL) rates keep almost the same value. It is also true that when $ \eta=2 $, the UL rates diverge from the DL rates, verifying that the DL gain rates are higher than the UL loss rates.

In Fig. \ref{fig: SR vs Max-min}, we study the relationship between the SR and max-min rate under different levels of $ \eta $ and $ \rho^2 $. As can be observed, the SR is directly proportional to $ \eta $ while the max-min rate follows the inverse trend, which is evident from UL rate in Fig. \ref{fig: Max-min eta}(a). The gap  of  the SR or max-min rate plummets as $ \eta $ increases. At $ \rho^2=-70 $ dB,  the gap between SRs (or max-min rates) is around 2 bps/Hz (or 1.25 bps/Hz) and less than 1 bps/Hz (or 0.5 bps/Hz) when $ \eta $ varies from 1 to 2 and from 8 up to 16, respectively.

\begin{figure}[t] 
	\begin{minipage}[t]{0.48\columnwidth}
		\centering
		\vspace{-30pt}
		\includegraphics[width=0.9\columnwidth, trim={0cm, 0cm, 0cm, -2cm}]{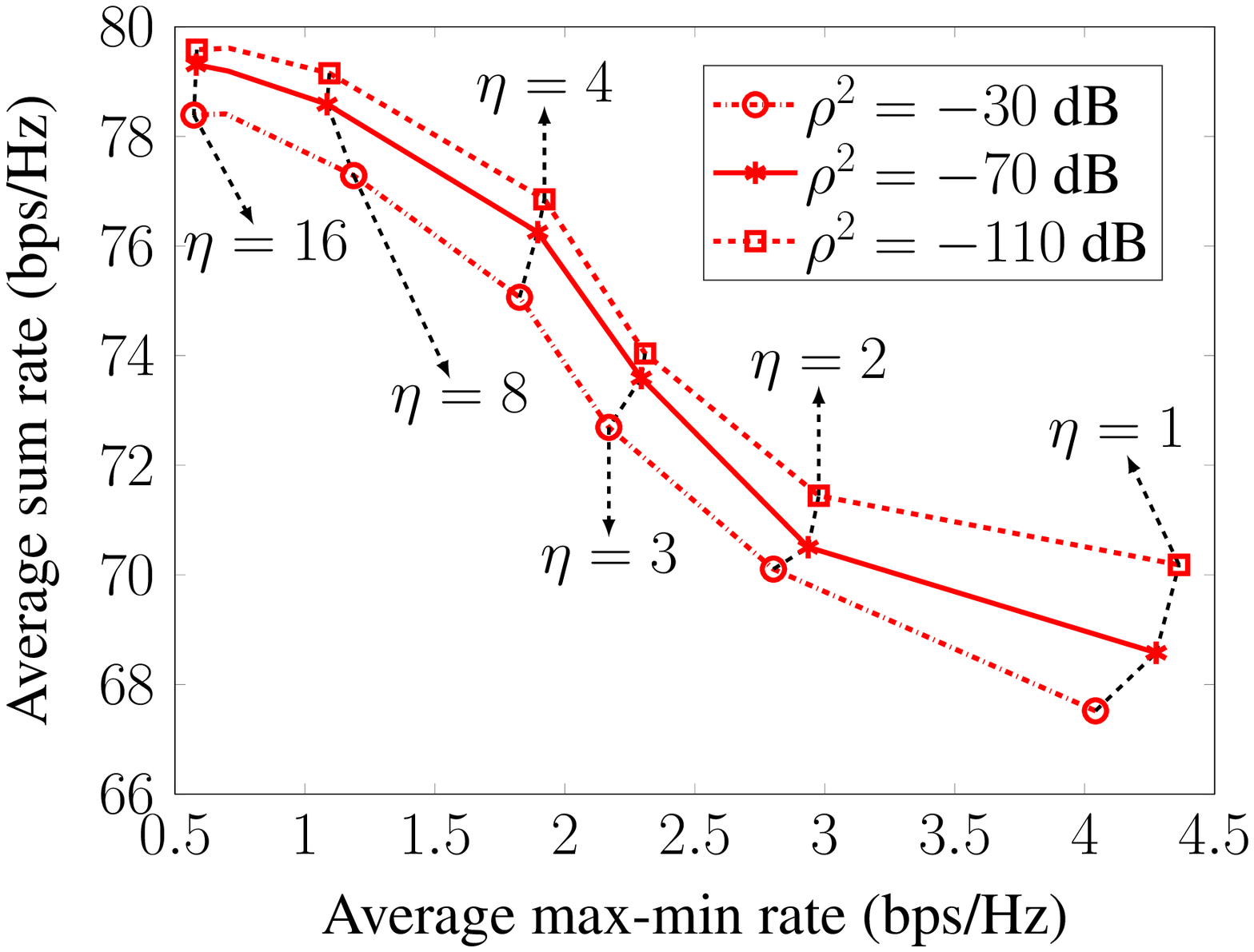}
		\vspace{-20pt}
		\caption{Relationship between sum rate and max-min rate.}
		\label{fig: SR vs Max-min}
%				\vspace{-20pt}
	\end{minipage}
	\hfill
	\begin{minipage}[t]{0.48\columnwidth}
		\centering
		\vspace{-30pt}
		\includegraphics[width=0.9\columnwidth, trim={0cm, 0cm, 0cm, -2cm}]{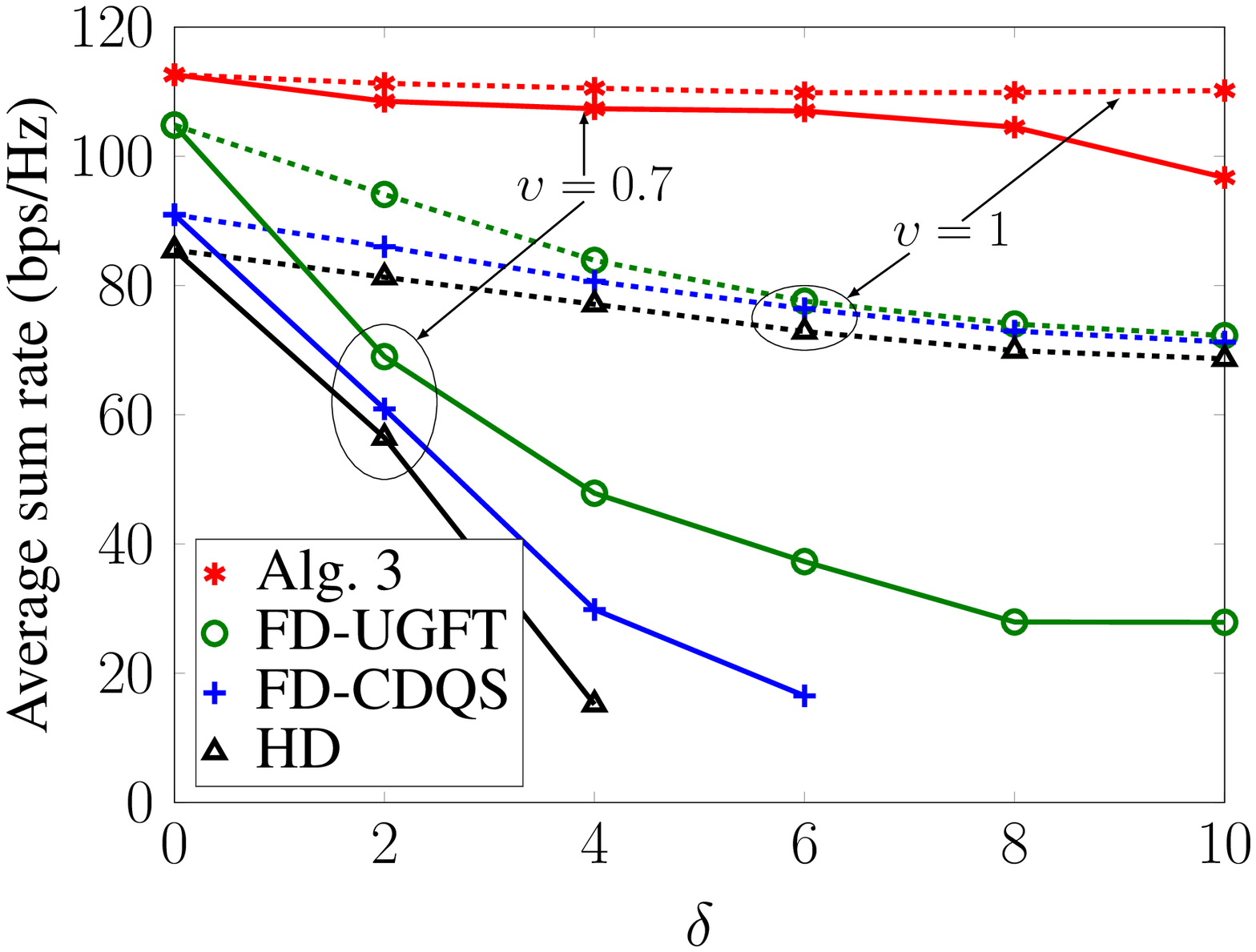}
		\vspace{-20pt}
		\caption{Sum rate with different stochastic error levels.}
		\label{fig: SR vs Delta}
%				\vspace{-20pt}
	\end{minipage}
\end{figure}

As the last numerical example, we investigate the effect of different values of channel uncertainties on the SR in Fig. \ref{fig: SR vs Delta}. The variances of the CSI errors are modeled as $ \epsilon_{\ell}^{\ul} = \epsilon_{k}^{\dl} = \epsilon_{\ell,k} = \epsilon = \delta r_{sn}^{-\upsilon} $, where  $ \delta, \upsilon \geq 0 $ and $ r_{sn} $  is the signal-to-noise ratio \cite{Maurer:TSP:Jan2011}. Clearly, the perfect CSI scenario is obtained with $ \epsilon=0 $ when $ \delta=0 $ for  an arbitrary value of $ \upsilon $.  To evaluate the performance, we use Algorithm 3 for the proposed robust FD and  extend the FD-UGFT  \cite{Dinh:Access} and the FD-CDQS  \cite{Dan:TWC:14}  with the estimation errors for robust designs. As seen from Fig. \ref{fig: SR vs Delta}, the SRs of the existing designs drop quickly as $ \delta $ increases. For $ \upsilon=0.7 $, the existing designs surprisingly deteriorate, and it is even impossible to find a feasible solution at $ \delta=4 $ (HD) or at $ \delta=6 $ (FD-CDQS). In contrast, the SR of the proposed FD design is less sensitive to the stochastic error levels, which confirms its robustness against the effect of imperfect CSI as compared to the others.

\section{Conclusion}\label{Conclusion}
In this paper, we have proposed new joint BS-UE association, time phase and power control designs for the FD MU-MISO systems. The considered problem is either: 1) to maximize the total sum rate subject to the minimum rate requirements for UL and DL users, or 2) to maximize the minimum rate among all users, which is formulated as a mixed-integer non-convex programming. To address the design problem, we have transformed the original problem into subproblems of lower dimensions. Then, iterative low-complexity algorithms have been proposed based on the  inner approximation method to solve the sequence of convex programs, exploiting its special properties. Our algorithms with realistic parameters monotonically improve the objective functions and convergence to a stationary point is guaranteed. Through numerical experiments, we have demonstrated the usefulness of our proposed algorithms in improving both the SR and max-min rate, even under strong effects of residual SI. The extension to the case of  imperfect CSI  has also been considered to further confirm the robustness of the proposed design.

\appendices
\renewcommand{\thesectiondis}[2]{\Alph{section}:}
%\vspace{-5pt}
\section{Proof of Proposition \ref{prop: Constraint Reduction}} \label{app: Constraint Reduction}
\vspace{-5pt}
{\hili It is clear that  function $ f_\omega $ in \eqref{eq: function f_omega} is well-defined since any $ \mathbf{z}=[z_1\;z_2]^T\in\{0,1\}^{2\times 1} $ can derive an $ f_\omega(\mathbf{z}) $ in co-domain $ \mathcal{B} $. For an arbitrary value of $ \hat{z}=\mathrm{b}`b_1b_0 \in \mathcal{B}$, we easily obtain the bit-wise matching as $ z_1=b_1 $ and $ z_2=b_0 $. For $ \mathbf{z}=[z_1\;z_2]^T$, there exists an inverse mapping $ \mathbf{z}=f_\omega^{-1}(\hat{z})=[b_1\;b_0]^T \in \{0,1\}^2,\; \forall \hat{z}\in\mathcal{B} $, indicating that $ f_\omega $ is surjective. Suppose that $ f_\omega(\mathbf{z}), f_\omega(\mathbf{z}') \in \mathcal{B} $ and $ f_\omega(\mathbf{z})= f_\omega(\mathbf{z}') =\mathrm{b}`b_1b_0 $. According to the definition of $ f_\omega $ in \eqref{eq: function f_omega}, we get $ \mathbf{z}=\mathbf{z}'=[b_1\;b_0]^T $, which represents an injection. This verifies the fact that the function $ f_\omega $ is bijective, and thus a function $ g_\omega:\; \{0,1\}^{2\times 1}\times\{0,1\}^{2\times 1}\rightarrow\mathcal{B}\times\mathcal{B},\; g_\omega(\boldsymbol{\Omega})=[f_\omega(\boldsymbol{\omega}_1)\;f_\omega(\boldsymbol{\omega}_2)] $ is also bijective.}

Now, we consider $ \boldsymbol{\omega}_j, \; j=1,2, $ in the objective function of \eqref{eq: prob. max sum SE}. Based on the structure of the formulation problem, $ \boldsymbol{\omega}_j \; (j=1,2) $ are independently selected in $ \{0,1\}^{2\times 1} $ and play the same role as mode selection in  two phases. In this manner, $ \boldsymbol{\Omega} $ is fixed as a 2-tuple $ (\boldsymbol{\omega}_1, \boldsymbol{\omega}_2) $, so that we can remove \eqref{eq: prob. general form b} and \eqref{eq: prob. general form c} to derive  subproblems of lower dimensions. 

To reduce the complexity of \eqref{eq: prob. max sum SE}, we further reconstruct the objective function  as $ f_R(\mathbf{X}_1, \mathbf{X}_2)= R_{\Sigma}(\mathbf{X}_1)+R_{\Sigma}(\mathbf{X}_2) $, where $ R_{\Sigma}(\mathbf{X}_j)\triangleq \sum_{\ell\in\mathcal{L}}\alpha_{\ell,j}^{\ul}\bar{\tau}_j \ln\bigl(1+\gamma_{\ell}^{\ul}\bigl(\boldsymbol{\omega}_j, \mathbf{w}_j, \mathbf{p}_j\bigr)\bigr)+\sum_{k\in\mathcal{K}}\alpha_{k,j}^{\dl}\bar{\tau}_j \ln\bigl(1+\gamma_{k}^{\dl}\bigl(\boldsymbol{\omega}_j, \mathbf{w}_j, \mathbf{p}_j\bigr)\bigr) $ is the sum rate of  phase $j$ and $ \mathbf{X}_j \triangleq \{\boldsymbol{\omega}_j, \mathbf{w}_j, \mathbf{p}_j, \tau, \boldsymbol{\alpha}_{\ell,j}^{\ul}, \boldsymbol{\alpha}_{k,j}^{\dl} \} $, $ j=1,2 $.  It is readily seen that $f_R(\mathbf{X}_1, \mathbf{X}_2) $ is a symmetric function, meaning $ f_R(\mathbf{X}_1^*, \mathbf{X}_2^*)= f_R(\mathbf{X}_2^*, \mathbf{X}_1^*) $. {\hili Moreover, since $ g_\omega(\boldsymbol{\Omega}) $ is a bijection from the finite set $ \{0,1\}^{2\times 1}\times \{0,1\}^{2\times 1} $ to the finite set $ \mathcal{B}\times \mathcal{B} $, there exists a subset $ \mathcal{S}_{\boldsymbol{\Omega}}\subset\{0,1\}^{2\times 1}\times \{0,1\}^{2\times 1} $ and a subset $ \mathcal{B}' \subset \mathcal{B}\times \mathcal{B} $ such that $ g_\omega^{-1}(\mathcal{B}')=\mathcal{S}_{\boldsymbol{\Omega}} $ and $ g_\omega(g_\omega^{-1}(\mathcal{B}'))=\mathcal{B}' $. From the symmetry property, we construct a set $ \mathcal{B}' $ as in \eqref{eq: fixing omega - HA assignment}, which is $ \mathcal{B}' \subset \mathcal{B}\times \mathcal{B} $, and the search area is thus reduced to a half of $ \mathcal{B}\times \mathcal{B} $. For each image $ \boldsymbol{\hat{\omega}}\in\mathcal{B}' $, the argument $ \boldsymbol{\Omega}=g_\omega^{-1}(\boldsymbol{\hat{\omega}})\in\mathcal{S}_{\boldsymbol{\Omega}}  $ satisfying constraints \eqref{eq: prob. general form b} and \eqref{eq: prob. general form c} is used as a constant parameter corresponding to a subproblem, which completes the proof.}

\vspace{-10pt}
\section{Bound Tightening in Lemma \ref{lem: Bound Tightening}} \label{app: Bound Tightening}
\vspace{-5pt}
To examine the effect of $ \mathbf{w}_{k,i,j} $ on the optimal value, we consider the beamforming vectors for DL users at the optimum of the subproblem. Let us reconstruct the constraints \eqref{eq: prob. general form g} and \eqref{eq: prob. general form h} at the optimum as
\begin{subequations} \label{eq: power constraints as optimum}
	\begin{gather} 
	\sum\nolimits_{k=1}^{K}\bar{\tau}_j \|\mathbf{w}_{k,i,j}^*\|^2 + \sum\nolimits_{k=1}^{K}\underset{(i',j')\neq (i,j)}{\sum\nolimits_{i'=1}^{2}\sum\nolimits_{j'=1}^{2}}\bar{\tau}_j \|\mathbf{w}_{k,i',j'}^*\|^2 \leq P_t^{\text{max}}, \\
	\sum\nolimits_{k=1}^{K}\|\mathbf{w}_{k,i,j}^*\|^2 + \sum\nolimits_{k=1}^{K}\|\mathbf{w}_{k,i',j}^*\|^2 \leq P_t^{\infty}, \; i'\neq i.
	\end{gather}
\end{subequations}
It is implicit that $ \mathbf{w}_{k,i,j}^* \in \mathcal{D}_{\mathbf{w}_{k,i,j}} \triangleq \{\mathbf{w}_{k,i,j}\in \mathbb{C}^{N\times1}| \sum_{k=1}^{K}\bar{\tau}_j \|\mathbf{w}_{k,i,j}\|^2 \in [0,a_{i,j}]\; \& \; \sum_{k=1}^{K}\|\mathbf{w}_{k,i,j}^*\|^2 \in [0,b_{i,j}] \}, \forall k \in \mathcal{K} $, where $ a_{i,j}\triangleq P_t^{\text{max}}-\sum\nolimits_{k=1}^{K}\underset{(i',j')\neq (i,j)}{\sum\nolimits_{i'=1}^{2}\sum\nolimits_{j'=1}^{2}}\bar{\tau}_j \|\mathbf{w}_{k,i',j'}^*\|^2 $ and $ b_{i,j}\triangleq P_t^{\infty}-\sum\nolimits_{k=1}^{K}\|\mathbf{w}_{k,i',j}^*\|^2 $. On the other hand, since $ \omega_{ij}=0 $, the channels corresponding to the $ i $-th HA are inactive. According to \eqref{eq: SINR DLUs}, any value of the beamforming vector $ \mathbf{w}_{k,i,j}^*\in \mathcal{D}_{\mathbf{w}_{k,i,j}} $ satisfying \eqref{eq: power constraints as optimum} impacts neither $ \gamma_{k}^{\dl}\bigl(\boldsymbol{\omega}_j, \mathbf{w}_j, \mathbf{p}_j\bigr) $ nor the $ k $-th user's rate. In other words, for $ \omega_{ij}=0 $, if we change the value of the vector $ \mathbf{w}_{k,i,j}^* $ at the optimum such that \eqref{eq: power constraints as optimum} holds, the optimal value would remain unchanged.

\section{Uplink SINR} \label{app: UL SINR}
To retrieve the UL information, the optimal weight vector is determined by minimizing the MSE function, i.e.,
\begin{align}
	f^\text{MSE}_{\ell, j} & =\mathbb{E}\Bigl[\bigl|x_{\ell,j}^{\ul}-\mathbf{u}_{\ell, j}^H\mathbf{\hat{y}}_{\ell,j}^{\ul}\bigr|^2\Bigr] = 1 + \mathbf{u}_{\ell, j}^H\mathbf{C}_{\mathbf{\hat{y}}_{\ell,j}^{\ul} \mathbf{\hat{y}}_{\ell,j}^{\ul}}\mathbf{u}_{\ell, j} - 2p_{\ell,j} \mathbf{u}_{\ell, j}^H \mathbf{\bar{h}}_{\ell,j}^{\ul}
\end{align}
where the covariance matrix of $ \mathbf{\hat{y}}_{\ell,j}^{\ul} $ is calculated by 
\begin{align}
	\mathbf{C}_{\mathbf{\hat{y}}_{\ell,j}^{\ul}\mathbf{\hat{y}}_{\ell,j}^{\ul}} & = \mathbb{E}\bigl[\mathbf{\hat{y}}_{\ell,j}^{\ul}\bigl(\mathbf{\hat{y}}_{\ell,j}^{\ul}\bigr)^H\bigr] = p_{\ell,j}^2  \mathbf{\bar{h}}_{\ell,j}^{\ul}\bigl(\mathbf{\bar{h}}_{\ell,j}^{\ul}\bigr)^H + \boldsymbol{\hat{\Psi}}_{\ell,j}
\end{align}
with $ \boldsymbol{\hat{\Psi}}_{\ell,j}\triangleq\sum_{\ell'=\ell+1}^{L}p_{\ell',j}^2  \mathbf{\bar{h}}_{\ell',j}^{\ul}\bigl(\mathbf{\bar{h}}_{\ell',j}^{\ul}\bigr)^H + \sum_{\ell'=1}^{L}p_{\ell',j}^2 \epsilon_{\ell'}^{\ul}\boldsymbol{\bar{\Lambda}}_j + \rho^2\sum_{k=1}^{K} \mathbf{\tilde{G}}_{j}^H \mathbf{w}_{k,j} \mathbf{w}_{k,j}^H\mathbf{\tilde{G}}_{j} + \sigma_{\mathtt{U}}^2\mathbf{I} $ and $ \mathbf{\bar{h}}_{\ell,j}^{\ul} \triangleq \boldsymbol{\bar{\Lambda}}_j \mathbf{\hat{h}}_{\ell}^{\ul} $. The optimal solution for MSE minimization satisfying $ \nabla_{\mathbf{u}_{\ell, j}^H} f^\text{MSE}_{\ell, j} = 0 $ is
\begin{align} \label{eq: weight vector for MMSE est.}
	\mathbf{u}_{\ell, j} = p_{\ell,j} \mathbf{C}_{\hat{\mathbf{y}}_{\ell,j}^{\ul} \hat{\mathbf{y}}_{\ell,j}^{\ul}}^{-1} \mathbf{\bar{h}}_{\ell,j}^{\ul}
\end{align}
showing  \eqref{eq: optimal weight vector}.

Furthermore, the information detection is performed as
$	\hat{x}_{\ell,j}^{\ul} = \mathbf{u}_{\ell, j}^H\mathbf{\hat{y}}_{\ell,j}^{\ul} = s_{\ell,j} + z_{\ell,j},
$
where the expected signal and interference-plus-noise are respectively given by 
\begin{align}
	s_{\ell,j} & = \mathbf{u}_{\ell, j}^H \mathbf{\bar{h}}_{\ell,j}^{\ul}x_{\ell,j}^{\ul} =  p_{\ell,j}^2 \bigl(\mathbf{\bar{h}}_{\ell,j}^{\ul}\bigr)^H \mathbf{C}_{\mathbf{\hat{y}}_{\ell,j}^{\ul} \mathbf{\hat{y}}_{\ell,j}^{\ul}}^{-1} \mathbf{\bar{h}}_{\ell,j}^{\ul}x_{\ell,j}^{\ul}, \label{eq: expected signal - chan. uncertainty} \\
	z_{\ell,j} & = \mathbf{u}_{\ell, j}^H\bigl(\sum\nolimits_{\ell'>\ell}^{L} p_{\ell',j}  \mathbf{\bar{h}}_{\ell'}^{\ul}x_{\ell',j}^{\ul} + \sum\nolimits_{\ell'=1}^{L}p_{\ell',j} \boldsymbol{\bar{\Lambda}}_j\boldsymbol{\Delta}\mathbf{h}_{\ell'}^{\ul} x_{\ell',j}^{\ul} + \rho\sum\nolimits_{k=1}^{K} \mathbf{\tilde{G}}_j^H \mathbf{w}_{k,j} x_{k,j}^{\dl} + \mathbf{n}_{\ell,j}\bigr). \label{eq: interf. plus noise - chan. uncertainty}
\end{align}
From \eqref{eq: weight vector for MMSE est.}-\eqref{eq: interf. plus noise - chan. uncertainty}, the SINR of $\ULU$ in phase $j$ by treating the CSI errors as noise is derived as
\begin{align} \label{eq: UL SINR - chan. uncertainty}
	\hat{\gamma}_{\ell}^{\ul}(\boldsymbol{\omega}_j,\mathbf{w}_j , \mathbf{p}_j) = \frac{\mathbb{E}\bigl[s_{\ell,j}s_{\ell,j}^H\bigr]}{\mathbb{E}\bigl[z_{\ell,j}z_{\ell,j}^H\bigr]} & = \frac{p_{\ell,j}^2 \bigl(\mathbf{\bar{h}}_{\ell,j}^{\ul}\bigr)^H \mathbf{C}_{\mathbf{\hat{y}}_{\ell,j}^{\ul} \mathbf{\hat{y}}_{\ell,j}^{\ul}}^{-1} \mathbf{\bar{h}}_{\ell,j}^{\ul}}{1-p_{\ell,j}^2 \bigl(\mathbf{\bar{h}}_{\ell,j}^{\ul}\bigr)^H \mathbf{C}_{\mathbf{\hat{y}}_{\ell,j}^{\ul} \mathbf{\hat{y}}_{\ell,j}^{\ul}}^{-1} \mathbf{\bar{h}}_{\ell,j}^{\ul}}.
\end{align}
By applying the matrix inversion lemma to the denominator of \eqref{eq: UL SINR - chan. uncertainty}, i.e., $1-p_{\ell,j}^2 \bigl(\mathbf{\bar{h}}_{\ell,j}^{\ul}\bigr)^H \mathbf{C}_{\mathbf{\hat{y}}_{\ell,j}^{\ul} \mathbf{\hat{y}}_{\ell,j}^{\ul}}^{-1} \mathbf{\bar{h}}_{\ell,j}^{\ul}  = \bigl(1+p_{\ell,j}^2 (\mathbf{\bar{h}}_{\ell,j}^{\ul})^H\boldsymbol{\hat{\Psi}}_{\ell,j}^{-1}\mathbf{\bar{h}}_{\ell,j}^{\ul}\bigr)^{-1}$,
\eqref{eq: UL SINR - chan. uncertainty} can be further simplified as
\begin{align}
	\hat{\gamma}_{\ell}^{\ul}(\boldsymbol{\omega}_j,\mathbf{w}_j , \mathbf{p}_j) & = \frac{1-\bigl(1+p_{\ell,j}^2 (\mathbf{\bar{h}}_{\ell,j}^{\ul})^H\boldsymbol{\hat{\Psi}}_{\ell,j}^{-1}\mathbf{\bar{h}}_{\ell,j}^{\ul}\bigr)^{-1}}{\bigl(1+p_{\ell,j}^2 (\mathbf{\bar{h}}_{\ell,j}^{\ul})^H\boldsymbol{\hat{\Psi}}_{\ell,j}^{-1}\mathbf{\bar{h}}_{\ell,j}^{\ul}\bigr)^{-1}} = p_{\ell,j}^2 \bigl(\mathbf{\bar{h}}_{\ell,j}^{\ul}\bigr)^H\boldsymbol{\hat{\Psi}}_{\ell,j}^{-1}\mathbf{\bar{h}}_{\ell,j}^{\ul}
\end{align}
showing \eqref{eq: SINR ULUs - chan. uncertainty}, and the proof is completed.

\titlespacing{\section}{-10pt}{-5pt}{-10pt}
\bibliographystyle{IEEEtran}
\bibliography{IEEEfull}
\end{document}